\newif\iftwocolumnversion
\twocolumnversionfalse
\iftwocolumnversion
  \documentclass[10pt,twocolumn]{article}
\else
  \documentclass[11pt]{article}
\fi

\usepackage[utf8]{inputenc}

\usepackage{dcolumn}
\usepackage{bm}

\usepackage{xcolor}
\definecolor{blueviolet}{rgb}{0.2, 0.2, 0.6}
\definecolor{webgreen}{rgb}{0,.5,0}
\definecolor{webbrown}{rgb}{.6,0,0}
\usepackage[utf8]{inputenc}
\usepackage[T1]{fontenc}
\usepackage[margin=1in]{geometry}
\usepackage{microtype}
\usepackage{amsmath,amssymb,amsfonts,amsthm,mathtools,bm}
\usepackage{graphicx}
\usepackage{subcaption}
\usepackage{booktabs}
\usepackage{tabularx}
\usepackage{xcolor}
\definecolor{blueviolet}{rgb}{0.2,0.2,0.6}
\definecolor{webgreen}{rgb}{0,.5,0}
\definecolor{webbrown}{rgb}{.6,0,0}
\definecolor{softboxblue}{RGB}{232,241,249}
\definecolor{softboxblueframe}{RGB}{79,119,158}
\usepackage[colorlinks=true,urlcolor=webbrown,linkcolor=blueviolet,citecolor=webgreen]{hyperref}
\usepackage[nameinlink,capitalize]{cleveref}
\usepackage[square,sort&compress,numbers]{natbib}
\usepackage[most]{tcolorbox}
\usepackage{cleveref}

\usepackage{titletoc}

\usepackage{enumerate}
\usepackage{comment}
\usepackage{xpatch}
\usepackage{braket}
\usepackage{physics}
\usepackage{amsthm}
\usepackage{tikz}
\usepackage{commath}
\usepackage{mathtools}
\usepackage{qcircuit}
\usepackage{algorithm}
\usepackage{algorithmicx}
\usepackage{algpseudocode}
\usepackage{tabto}
\usepackage{pgf-umlsd}
\usepackage{bm}
\usepackage{makecell}
\usepackage{subcaption}
\usepackage{pgfplots}
\pgfplotsset{compat=1.18}
\usepgfplotslibrary{fillbetween}

\usepackage{caption}
\usepackage{graphicx}
\usepackage{tabularx}
\usepackage{booktabs}
\usepackage{tabularray}

\let\Pr\relax
\DeclareMathOperator*{\Pr}{\mathbf{Pr}}

\newcommand{\eps}{\varepsilon}

\newtheorem{theorem}{Theorem}
\newtheorem{prop}{Proposition}
\newtheorem{lemma}{Lemma}
\newtheorem{corollary}{Corollary}

\theoremstyle{remark}

\theoremstyle{definition}
\newtheorem{definition}{Definition}

\newtheorem{example}{Example}

\numberwithin{equation}{section}

\newcommand{\nocontentsline}[3]{}
\let\origcontentsline\addcontentsline
\newcommand\stoptoc{\let\addcontentsline\nocontentsline}
\newcommand\resumetoc{\let\addcontentsline\origcontentsline}

\crefname{prop}{Proposition}{Propositions}
\Crefname{prop}{Proposition}{Propositions}
\crefname{theorem}{Theorem}{Theorems}
\Crefname{theorem}{Theorem}{Theorems}
\crefname{lemma}{Lemma}{Lemmas}
\Crefname{lemma}{Lemma}{Lemmas}
\crefname{definition}{Definition}{Definitions}
\Crefname{definition}{Definition}{Definitions}

\newenvironment{widefigure}{\iftwocolumnversion\begin{figure*}[t]\else\begin{figure}[t]\fi}{\iftwocolumnversion\end{figure*}\else\end{figure}\fi}
 
\definecolor{todocol}{rgb}{1,0,0.0}
\definecolor{fixmecol}{rgb}{1.0,0.3,0.3}
\definecolor{ideacol}{rgb}{1.0,0.75,0.0}
\definecolor{probcol}{rgb}{1.0,0.1,0.0}
\definecolor{green-munsell}{rgb}{0.0, 0.66, 0.47}
\definecolor{navyblue}{rgb}{0.0, 0.0, 0.5}

\usepackage[framemethod=TikZ]{mdframed}
\usepackage{xcolor}
\definecolor{softboxblueframe}{RGB}{79,119,158}
\newmdenv[
  linecolor=softboxblueframe,
  linewidth=0.8pt,
  roundcorner=5pt,
  innertopmargin=10pt,
  innerbottommargin=10pt,
  innerrightmargin=10pt,
  innerleftmargin=10pt,
  backgroundcolor=white,
  frametitleaboveskip=-0.5\ht\strutbox,
  frametitlealignment=\raggedright,
  nobreak=true,
]{styledbox_base}
\newenvironment{informal}[1]{
  \begin{styledbox_base}[frametitle={\colorbox{white}{\space#1\space}}]
}{
  \end{styledbox_base}
}

\title{Routing Anonymity and Identifiability \\ of Noisy Quantum Hardware}

\author{
	Ben Priestley\\
	\small \textit{Centre for Quantum Information and Foundations, DAMTP, University of Cambridge} \\
	\small \textit{Quantum Software Lab, School of Informatics, University of Edinburgh}
	\and
	Mina Doosti\\
	\small \textit{Quantum Software Lab, School of Informatics, University of Edinburgh}
}

\date{}

\begin{document}
\maketitle

\begin{abstract}

    Present-day quantum computing---and its very likely future---is cloud-based, where a user submits a circuit to be executed by a service provider on proprietary backend hardware. While providers may wish to hide implementation details, scheduling choices, or even which physical device was used, noisy finite-shot outputs can carry backend-specific “fingerprints”---information imprinted in the classical output distribution that can reveal the backend identity. So far, such fingerprints have mostly been studied from a benchmarking perspective, for example as a tool for verification, with only limited attention to the privacy considerations for both users and providers in these scenarios.
    
    This work develops the first formal framework for backend identifiability and the corresponding privacy notion. We introduce an operational backend-identifiability game and use it to formalise \textit{routing anonymity} as a security notion for quantum cloud services. We show that backend identifiability is exactly a hypothesis-testing problem and prove that, under passive i.i.d. access to a single backend, routing anonymity decays exponentially at the Chernoff rate. We also establish a utility-anonymity trade-off, imposing fundamental limits on how much backend-specific information can be removed from classical outputs without degrading their usefulness. In addition, we observe that, for noisy quantum hardware, identifying fingerprints are inherently an intermediate-depth phenomenon, and establish a formal depth principle using Pauli-transfer-matrix tools.
    
    We complement the theory with experiments on a platform that reflects this real-life scenario, namely Amazon Braket on AWS, and run our experiments on different hardware platforms, including ion-trap and superconducting quantum processors. We observe 87--90\% classification between superconducting backends and 96--100\% classification across physical platforms, and find that identifiability can survive several natural forms of post-processing. Overall, these results establish routing anonymity as a distinct security requirement for quantum cloud computing, and provide a framework for quantifying and controlling the resulting utility-anonymity trade-off.
\end{abstract}

\newpage

\newpage
\section{Introduction}

The prospect of quantum computation has long transitioned from a theoretical novelty to an active practical pursuit of useful quantum advantages across many different domains: cryptography and cryptanalysis \cite{Shor-1994, proos2003shor, grassl2016applying, roetteler2017quantum, kaplan2015quantum, kaplan2016breaking, laarhoven2015finding, Priestley_2026, chailloux2021lattice}; quantum chemistry and material science \cite{aspuru2005simulated, peruzzo2014variational, kandala2017hardware, reiher2017elucidating, google2020hartree, von2021quantum, clinton2024towards}; many-body, condensed matter, and high-energy physics \cite{bernien2017probing, lloyd1996universal, martinez2016real, ebadi2021quantum, atas20212, meth2025simulating}, fundamental physics and quantum gravity \cite{Zhu-2025a, li2019quantum, east2021spin, priestley2025lqg, mielczarek2019spin, van2023experimental}, combinatorial optimisation \cite{farhi2014quantum, lucas2014ising, ebadi2022quantum}, quantum machine learning \cite{rebentrost2014quantum, schuld2019quantum, havlivcek2019supervised, liu2021rigorous, huang2021power, huang2022quantum, yin2025experimental}, etc. Unlike the traditional (classical) computing paradigm of individual hardware ownership,\footnote{At least prior to the unprecedented demand for incredibly large-scale compute that has followed the development of \textit{large} large language models; even in the classical computing world, remote/cloud services are becoming standard practice, and increasingly so!} large-scale quantum technologies are inevitably more suited to cloud-based access and ``quantum-as-a-service (QaaS)'' settings \cite{nguyen2024quantumcloudcomputingreview, Golec_2024}. In this paradigm, a user writes a description for a quantum circuit---implementing some task that she is interested in---and submits it to a service provider, who then goes off and executes this circuit for her on some quantum processing unit (QPU) and later returns a finite-shot sample output.

As far as the user is concerned, cloud-based quantum computing offers an interface through which implementation details are abstracted to such an extent that the hardware itself can feel almost interchangeable. The same interface may expose superconducting, trapped-ion, neutral-atom, simulator backends, etc. often from several hardware providers all through a single high-level account \cite{AWSBraketDocs, AWSBraketService}. Physically, of course, the QPUs are decidedly \textit{not} interchangeable to the service provider, who must care a great deal about each QPU's calibration history, native gates, topology, crosstalk, readout asymmetries, compilation path, slowly drifting environment, and such related hardware-specific details. Consequently, in the noisy near-term of quantum computation, while the provider can feign some privacy about his choice of backend by hiding the label on his execution routing, it is has been observed that a `fingerprint' specific to the chosen backend will be left in the classical outputs that he returns to the user \cite{Martina2021learningNoiseFingerprint, RoyGhoshGhosh2025forensics, smith2023fast, mi2021short, mutolo2025quantumcomputerfingerprintingusing, Wu-2024}. There are two perspectives to take about this observation: (\textit{user-side}) the user herself may want to be able to verify that her circuit has actually be run on the claimed backend hardware without having to rely on a promise alone; and (\textit{provider-side}) the provider may want to keep secret his choices of backend routing, with some guarantee about the user's inability to violate that privacy. In either case, the operational question is equivalent:
\begin{center}
	\emph{How much route information is leaked by classical outputs of noisy quantum hardware?}
\end{center}

There is a large literature on security notions and verifiability for delegated quantum computation, but its dominant emphasis is on \textit{user-side} privacy and correctness. Blind quantum computation and verifiable blind quantum computation protect the user's input, output, and computation from an untrusted quantum server, while also allowing the user to detect incorrect behaviour in the verifiable variants \cite{Childs2005secure, Broadbent2009UBQC, FitzsimonsKashefi2017VUBQC, Dunjko2014composable, badertscher2020security, poshtvan2025selectively}. Quantum verification protocols ask whether a classical verifier, or a verifier with limited quantum capabilities, can certify the correctness of a quantum computation performed by a quantum prover \cite{Gheorghiu2019verification, Aharonov2017interactive, Mahadev2018verification, chia2020classical, bartusek2022succinct, zhang2021succinct, kalai2026classically, leichtle2021verifying}. Related approaches based on remote or oblivious state preparation, quantum homomorphic encryption, and hardware-assisted secure execution similarly aim to protect the user's computation or to certify the service outcome under additional cryptographic or hardware assumptions \cite{GheorghiuVidick2019RSP, Mahadev2018homomorphic, Alagic2018vQFHE, Ma2022QEnclave}. Indeed, the aforementioned observations about backend-specific fingerprints \cite{Martina2021learningNoiseFingerprint, RoyGhoshGhosh2025forensics, smith2023fast, mi2021short, mutolo2025quantumcomputerfingerprintingusing, Wu-2024} largely take the user-side perspective in their experimental or analytical approaches. Even where provider-side concerns are discussed, the protected mathematical object is never formally given as the anonymity of the provider's route \cite{coupel2025securityvulnerabilitiesquantumcloud}.

In this work, we take the much-neglected \textit{provider-side} perspective and ask the complementary security questions; namely, to what extent can the provider's routing choice be anonymised without invalidating the promised service? It is worth emphasising the distinction of this question from the standard user-side perspective, in that we reverse the goals to describe mechanisms through which information is removed/obscured rather than inferred (either experimentally or by analysis of known/assumed noisy behaviour). Our novelty is then the proposal of a unified framework for backend identification which establishes, to our knowledge, the first formalisation of \textit{routing anonymity} for cloud quantum computing, with supporting statements about how this privacy notion interacts with the information implicit in backend-specific noise fingerprints. We describe precisely how privacy depends jointly on the workload of circuits submitted by the user, the finite-shot resolution of the executing hardware, repeated access, any post-processing of the measured output, and the promised utility. In fact, for the latter, it is interesting in and of itself to describe his privacy as a property of the very service he intends to provide, which we note is explicitly missing from the current (user-side) literature.

We propose a game of backend identifiability in which the provider secretly selects a backend, executes a user-chosen circuit from some agreed ensemble, and returns a classically post-processed outcome distribution from which the user should guess the backend label (\cref{sec:routing_game}). Our framework formalises how this game can be reduced exactly to a hypothesis testing question, and thus can be modeled statistically and sufficiently described as combinations of binary identification games (\cref{sec:statistical_reduction}). Under a `persistent routing' extension of the game (which assumes a fixed backend label and passive i.i.d. probing from the user), we can observe a similar reduction to see that anonymity decays at the Chernoff rate. Following this, we describe how post-processing can act as a  suppression mechanism for route-specific noise, and then incorporate the promised utility of the service to place fundamental bounds on the degree of anonymity that can be obtained (\cref{sec:post_processing_and_utility}). This is formalised as a utility-anonymity no-free-lunch theorem that informs the provider how he should mathematically model his utility with respect to the anonymity he desires. We can understand backend identifiability as existing in an `intermediate-depth' window, and prove this characteristic of distinguishability in a Pauli-transfer-matrix model (\cref{sec:intermediate_depth}). Finally, we also proposes a new channel (psuedo-)distance which is designed to act as a tight measure on backend distinguishabilty tailored to the user's workload, which can then be used to provide some simple sufficient conditions on anonymity (\cref{sec:noise_channels}).

To support our framework and theoretical results, we design a series of experiments using Amazon Web Services (AWS) Braket, running on Rigetti's \textit{Ankaa-3}, IQM's \textit{Garnet}, and IonQ's \textit{Aria-1}, to demonstrate how finite-shot transcripts carry route-specific information in practice (\cref{sec:experiments}). These are similar in spirit to existing fingerprinting work \cite{Martina2021learningNoiseFingerprint, RoyGhoshGhosh2025forensics, smith2023fast, mi2021short, mutolo2025quantumcomputerfingerprintingusing, Wu-2024}, although we distinguish our experimentation in a few key ways to better suit our unique perspective. Firstly, we design multiple suites of experiments to understand different mechanisms of distinguishability; our `depth-varied' experiments use random circuits of varying depth---a far more challenging classification task than we see in most existing literature---while our `time-varied' experiments more closely match those of other works. Secondly, we investigate multiple forms of post-processing on classical outcomes to understand the utility-based arguments of our work, and observe the intermediate-depth principle's reaction over different transcript forms. Thirdly, we make an explicit and unique effort to visualise how noise fingerprints are learned by simple classifiers, giving an intuition for how separability between backends presents with respect to the workload circuits. We also offer some preliminary ideas and exploration of fingerprint forecasting, in which we indicate how we may be able to predict the evolution of the fingerprint in time, without explicit assumptions on the noise model.

\subsection{Security Motivations and Applications} \label{sec:example_applications}

Motivations for a notion of `routing anonymity' can be relatively direct, particularly in analogy to the classical case, in which routing metadata is ubiquitously security-sensitive. For example, mix networks and onion routing aim to hide communication paths, unlink senders and receivers, or limit what observers can infer from traffic metadata \cite{Chaum1981mix, Reed1998onion, Dingledine2004Tor}. More generally, the provider may have any number of valid reasons for wanting to keep his routing choices hidden; e.g. keeping private his scheduling policies, confidentiality agreements that he may have with his hardware providers which require particular implementation details to be kept private from competing parties, preventing users from interpreting QPU loads or inferring times of increased stress, etc. 

If the reader would only humour us for a moment, we will now describe a select few toy examples in which a motivation for provider-side security is required immediately by design. The motivation need not be so contrived in typical cases of quantum cloud computation, as we note above, but it is nevertheless interesting to acknowledge specialised settings in which anonymity of routing choice is not only desirable but also application critical. The reader who is convinced about the relevance of routing anonymity can skip this section, humourlessly.

In the following examples, we use the backend as a private \textit{verifier/oracle} (\cref{example:quantum_lock_and_key}) of a user-supplied object; as a private \textit{issuer} whose signatures must be resistant to imitation (\cref{example:noisy_tokens}); and as a hidden \textit{target} whose identity reveals vulnerabilities (\cref{example:probing}).

\begin{example}[Quantum Lock-and-Key] \label{example:quantum_lock_and_key}	
	Suppose that Alice guards a collection of locked vaults, whose contents are so valuable that even she herself is not allowed to know the secret keys which open them. She is only permitted to blindly try presented keys on the requested vault's lock and observe whether it opens. We can model this setting as follows: define some quantum circuit to be the `key' and associate with each vault a noisy quantum backend to be the `lock', and allow Alice to observe the backend-specific output of the circuit prior to locking the vault. Now, when Bob comes along and presents a circuit and asks to open a particular vault, Alice routes it to the appropriate backend (assuming a private mapping between the vault and backend label), then passes the output through a decision function to decide whether the key fits the lock to open the vault.
	
	If Bob can collect decisions over an ensemble of candidate key circuits, then we better guarantee that such a collection does not provide him with enough information to identify the lock! Knowing the lock makes designing the key somewhat trivial. Hence, in this setting, the anonymity of Alice's routing choice is critical and should be preserved, even after (finitely-)many attempts.
\end{example}

The above quantum lock-and-key example can be viewed more generally as a private oracle whose hidden physical backend defines a decision rule, and whose output is then a backend-specific accept/reject bit. This specific set up is interesting not only for its overt requirements for routing anonymity, but also because it demonstrates a problem setting in which multiple layers of security are required. Notably, Alice is essentially asked to be a middle-woman routing a given key to a privately-known lock, thereby setting up privacy of the key for Bob (assuming that Alice cannot look at the given key, she herself is not able to unlock the vault, except maybe after many attempts) and privacy of the lock for Alice. A `double-blind' security setting, if you like.

\begin{example}[Noisy-Issued Tokens] \label{example:noisy_tokens}
	Suppose that Alice is a bank who secretly chooses one of several noisy backends to act as the active mint for each time period. To issue a token, she picks a public classical serial number $s$ which generates a corresponding quantum circuit $C_s$, then executes $C_s$ on the active backend and compresses its output into a short classical signature of the token to give to a customer, Bob. To later verify a presented signature, Alice can use her privileged knowledge of the then-active backend to test whether its particular fingerprint is sufficiently present in the signature.
	
	Now suppose that Bob collects a genuine token and attempts to produce a counterfeit with a new serial number $s'$. If he is able to identity the active mint from his valid signature, he can tailor his counterfeit to imitate its fingerprint on $C_{s'}$. But, if the routing can be anonymised and he cannot infer the active mint, then his counterfeit signature will fail to carry the appropriate fingerprint and Alice will deem it to be invalid.
\end{example}

\begin{example}[Selective Probing] \label{example:probing}
	Suppose that Alice has a collection of $N$ backends which are each very slow on a disjoint family of circuits, and we assume that they are arbitrarily fast at everything else. We might also assume that the union of these families covers all possible circuits; i.e. for any circuit that a user, Bob, can define, exactly one backend will execute it slowly. She allocates each backend to a distinct, non-overlapping time slot lasting $(1/N)$-th of the day, and routes any of Bob's requests to the corresponding backend whenever he happens to ask. She waits until the end of the time slot before returning the classical outcome distributions of any executions within that slot.
	
	Assume that Bob knows exactly which family each backend is slow on, but that he is only allowed to prepare $M\ll N^2$ probe circuits to be submitted whenever he likes. Upon receiving a response from the backend at the end of slot $t$, he can look at the output and design new probes for $t+1$. His task is to overload the $N$-th and final backend such that it takes longer than its prescribed $1/N$ time and Alice cannot go home at the end of the day. 
\end{example}

In the above example, the condition on the number of probes that Bob can submit prevents him from brute-force finding a slow circuit for each backend in turn; instead he must learn something about how his probe was routed. After $N-1$ time slots, he will ideally have guessed at which backend each slot had been assigned to, and be left at the final slot with an idea about which backend remains. He can then design a probe which he knows will be slow for this final backend, and if his guess is right, he will win the game and make Alice late for bed. Clearly then, it is in Alice's interest to ensure that Bob cannot learn how she routes his requests in each slot. This example could be viewed in analogy to a distributed denial-of-service (DDoS) attack, tailored to the quantum cloud.

As a final applications remark, note that here we are \textit{not} interested in quantum process tomography, due to its demanding resource requirements. Full generic process tomography scales exponentially with system size, gate-set tomography is deliberately invasive, and scalable benchmarking or noise-learning methods make structural choices about which figures of merit to estimate \cite{Magesan2011scalable, BlumeKohout2017GST, Hamilton2020scalable, Harper2020efficient, VandenBerg2023pec, eisert2020quantum}. System-level benchmarks, such as quantum volume and cycle benchmarking, compare hardware capabilities and error rates, but they are not designed from a security or privacy perspective~\cite{Cross2019quantumvolume, Erhard2019cyclebenchmarking, Resch2021benchmarking}. In practice, a user trying to identify a backend may succeed using a much cheaper-to-extract notion.

\subsection{Contributions}

We summarise our primary contributions as follows:

\begin{itemize}
	\item \textbf{Unified framework and formal privacy notion.} We introduce the backend identifiability game and formally define routing anonymity as an operational privacy notion for hiding the service provider's routing choices in the quantum cloud setting. We extend to a multi-round version of the game under assumed i.i.d. passive access and describe the rate of anonymity decay via Chernoff information.
	
	\item \textbf{Statistical characterisation of backend identifiability.} We reduce optimal backend identification to classical hypothesis testing over route-induced transcript laws, allowing us to use known statistical results to describe distinguishing bias via total variation distance, and give a sufficient condition for general routing anonymity with respect to constituent binary identification games.
	
	\item \textbf{Utility-anonymity trade-off.} Formalising post-processing as the provider's principal mechanism for suppressing route-specific information, we introduce ideas about utility preservation and prove a corresponding no-free-lunch result about its relationship with routing anonymity.
	
	\item \textbf{Workload-relative conditions for anonymity.} We introduce a workload-probed channel (pseudo-)distance that measures average channel separation over the permitted probe workloads, yielding a tighter practical bound on routing anonymity than worst-case diamond norm.
	
	\item \textbf{Intermediate-depth principle.} In a contractive Pauli-transfer-matrix model with a common mixing component and small backend-specific perturbations, we prove that route-specific noise signals initially accumulate in depth before rapidly decaying. This formalises an intermediate-depth window in which backend identification is most effective.
	
	\item \textbf{Experimental implementation.} Using both random and structured workloads, several forms of transcript post-processing, and designing temporal experiments, we demonstrate our framework in practice to show that finite-shot outputs from real QPUs can expose substantial route information. Routing anonymity is then of considerable practical interest.
\end{itemize}

\section{Framework Overview and Informal Theoretical Results} \label{sec:overview_framework_and_theoretical_results}

In this section, we present our theoretical framework intuitively and informally to give a technical overview of our main results. The formal framework is deferred to Appendix \ref{appendix:framework}, and precise statements and theoretical results to Appendix \ref{appendix:theoretical_results}. 

\subsection{A Game of Routing} \label{sec:routing_game}

    We formalise the discussed security notion as a game played between a \textit{user} with a desire to learn backends, and a \textit{provider} with the converse desire to keep secret his routing choice. In this context, a \textit{`backend'} is any machinery with the capacity to execute a quantum circuit and return measurements thereof; e.g. a single quantum hardware device, a cohort of devices, a portion of a device, or some subset of registers, etc. We call the provider's choice of backend a \textit{`routing'}---he selects a \textit{route} through which to execute a given circuit. Generally, we use the terms `backend' and `route' interchangeably, but the rule of thumb is that the user wants to \textit{identify the backend} that runs her computation while the provider wants to \textit{anonymise his route} (to a backend) from the user.

    The interaction between these parties looks like the following: the user probes the provider with a \textit{workload (quantum) circuit} of her choosing, and the provider then dutifully executes that circuit through a route of his choosing over a finite precision of \textit{shots} to produce an empirical probability distribution, which is thus naturally associated with both the backend and circuit. Before handing this distribution back to the user, the provider passes it through some \textit{(deterministic) post-processing map} in accordance with the service that he is promising to provide. The idea is that the user is not generally interested in the raw outcome sequences of her circuit, but rather some property; e.g. expectations of observables, energies of Hamiltonians, ground states, correlator functions, decision outcomes, etc.; hence it is prudent on the provider's part to reveal only the property of interest when his privacy is of concern.
    
    Following this interaction, the user has received a \textit{transcript} in association with her known circuit and the unknown backend. The question---and the unmistakable fun in this game---is whether there is enough information accessible in the transcript to identify this backend. Other questions linger on the periphery; e.g. which ensemble of workload circuits makes accessible the most amount of usable information for this identification; which classes of post-processing maps are most hiding of information in the exposed transcript; what relationship does the capacity for identification have with the number of shots? All very interesting and natural aspects of the game.

    The security game can be described as follows:

    \begin{informal}{Game 1: Backend/Route Identification Game}
        The user fixes a collection of $n$-qubit workload circuit ensembles $\{\mu_d\}_{d\geq 0}$, and a finite number $m\in\mathbb{N}$ of shots. The provider fixes a known (deterministic) post-processing map $\phi:\mathcal{Y}\to\mathcal{X}$, and a prior $\pi\in\Delta(k)$ over a backend set $\mathcal{D}=\{D_1,\dots,D_k\}$. The game is then played like:
        \begin{enumerate}
            \item Provider randomly samples a hidden route $I\sim \pi$.
            \item User chooses a circuit depth $d\in\mathbb{N}$ and randomly samples a depth-$d$ workload circuit $C\sim\mu_d$, then submits it to the provider.
            \item Provider executes $C$ over $m$ independent shots on backend $D_I$, producing a raw empirical distribution $\hat{p}_{I,C}\in\Delta_m(\{0,1\}^n)$.
            \item Provider produces a transcript $X=\phi(C,\hat{p}_{I,C})$, then returns it to the user.
            \item User observes $X$ and outputs a guess $\tilde{I}$, winning the game if $\tilde{I}=I$.
        \end{enumerate}
    \end{informal}
    
    We can extend this game into a more realistic setting by allowing the user to probe the backend over multiple rounds. There are many plausible constructions for a multi-round version of the above identification game; we consider a \textit{persistent routing} setting in which the provider fixes a backend prior to playing, allows the user to submit multiple workload circuits to the backend, and then executes them in bulk before returning a lengthened transcript. Importantly, this is a passive and i.i.d. access model as opposed to e.g. adaptive access models that permit the user to sample her $t$-th circuit after having observed the $(t-1)$-th transcript. This assumption simplifies our later observations, but is certainly worth developing in future work.

    \begin{informal}{Game 2: Persistent Routing Game}
        The backend/route identification game can be extended over $T$ rounds as follows:
        \begin{enumerate}
            \item The user samples a sequence $(C_1,\dots,C_T)$ of workload circuits, with each $C_t\sim\mu_{d}$ sampled from the same depth-$d$ ensemble, and submits them \textit{all} the provider; 
            \item The provider then produces a sequence $(X_1,\dots,X_T)$ of independent transcripts in turn by executing and post-processing each circuit on the same fixed backend $D_I$; and
            \item The user finally receives the lengthened transcript and again guesses $\tilde{I}$.
        \end{enumerate}
    \end{informal}

    At this point, we can explicitly state the provider's goal in this game. In contrast to the user's goal of identifying the backend, the provider succeeds if the user cannot identify the backend with probability substantially higher than random guessing. A central contribution of this work is to show how routing anonymity emerges through different instantiations of the game and its parameters: which classes of post-processing maps, which circuit ensembles, how many shots, and related choices preserve anonymity of the route, and to what degree? First, we define the routing anonymity with respect to the described games as follows.

    \begin{informal}{(Informal \cref{def:anonymity}) Routing Anonymity}
        We say that the (persistent) backend/route identification game has \textit{$\varepsilon$-anonymity} if the user's probability to guess the chosen route is bounded to within $\varepsilon>0$ of the baseline random guess.
    \end{informal}

\subsection{Identification and Anonymity via Statistical Tests} \label{sec:statistical_reduction}

    Any ensemble $\mu$ of workload circuits, executed on a particular backend $D_i$ over $m$ shots, has an associated \textit{raw transcript law} $Q_i(\mu,m)$ by which the empirical distribution is sampled. After being passed through a post-processing map $\phi$, we can instead refer to an \textit{induced transcript law} $P_i(\mu,m,\phi)=\phi_\# Q_i(\mu,m)$. For brevity, write $Q_i$ and $P_i$ when the context is clear.

    \begin{informal}{(Informal \cref{thm:backend_identifiability_reduces_to_hypothesis_testing}) Backend Identifiability Reduces to Hypothesis Testing}
        Identifying the backend producing a received transcript $X$ is exactly a statistical hypothesis test with hypotheses of the form $H_i:X\sim P_i$; accepting $H_i$ implies guessing backend $D_i$. 
        
        In particular, between two equally-likely backends $D_1$ and $D_2$, the probability $p_s^\star$ to correctly identify the backend is given by,
        \[
        p^\star_s = \frac{1}{2}\Big( 1 + \mathrm{TV}(P_1, P_2) \Big) \quad ,
        \]
        where $\mathrm{TV}(\cdot,\cdot)$ denotes the usual total variation (TV) distance between probability distributions.
    \end{informal}

    Our first result, \cref{thm:backend_identifiability_reduces_to_hypothesis_testing}, makes the conceptually simple connection between classical discrimination of probability distributions from an observed sample and the backend identifiability game, as a security notion. The observation that bounds guessing probability is then imported immediately from the abundance of known classical literature in hypothesis testing, and allows us to obtain a clean relation for \textit{distinguishing bias} $\beta_{D_1,D_2}$ in the uniform-binary case as exactly the TV distance between the respective induced laws; i.e. we have $\beta_{D_1,D_2}\equiv \mathrm{TV}(P_1, P_2)$ bias, beyond random guessing.

    Throughout this work, as in \cref{thm:backend_identifiability_reduces_to_hypothesis_testing}, it will be convenient to restrict our attention to this `uniform-binary' setting between a pair of equally-likely backends. In fact, the following \cref{prop:pairwise_indistinguishability_is_sufficient_for_global_anonymity} should settle our nerves about this restriction by showing that pairwise indistinguishability is sufficient to obtain a global anonymity in the fully general setting among an entire set of finitely-many backends with arbitrary priors. Hence, we can reason about anonymity via simple pairwise arguments, at least in that we can provide sufficient conditions for the general setting.

    \begin{informal}{(Informal \cref{prop:pairwise_indistinguishability_is_sufficient_for_global_anonymity}) Pairwise Indistinguishability is Sufficient}
        The optimal probability to correctly identify among a set of finitely-many backends with arbitrary priors can be decomposed as a sum of pairwise distinguishing biases. 
        
        In particular, the excess advantage $\mathrm{Adv}$ (i.e. beyond random guessing) in the fully general setting can be bounded by the worse-case distinguishing bias among pairs of backends;
        \[
        \mathrm{Adv} \leq \max_{i\neq j} \mathrm{TV}(P_i, P_j) \quad .
        \]
    \end{informal}

    Our second important result, \cref{thm:persistent_routing_reduces_to_chernoff_testing}, then extends \cref{thm:backend_identifiability_reduces_to_hypothesis_testing} into the persistent routing setting to observe that, under our passive i.i.d. assumptions, the provider's anonymity decays exponentially at the Chernoff rate between the relevant induced laws. 

    \begin{informal}{(Informal \cref{thm:persistent_routing_reduces_to_chernoff_testing}) Persistent Routing Reduces to Chernoff Testing}
        In the passive i.i.d. access model of persistent routing over $T$ rounds, the hypotheses of \cref{thm:backend_identifiability_reduces_to_hypothesis_testing} become $H_i:(X_1,\dots,X_T)\sim P_i^{\otimes T}$, again with acceptance of $H_i$ corresponding to guessing $D_i$. 

        Then, between two equally-likely backends $D_1$ and $D_2$, the probability $p_s^\star(T)$ to correctly identify the backend after $T$ rounds satisfies,
        \[
        1-p_s^\star(T) = \exp\Big( -T \cdot D_\mathrm{Ch}(P_1,P_2) + o(T) \Big) \quad ,
        \]
        where $D_\mathrm{Ch}(\cdot,\cdot)$ denotes the usual Chernoff information between probability distributions.
    \end{informal}

\subsection{Post-Processing and a Utility-Anonymity Trade-Off} \label{sec:post_processing_and_utility}

    A fundamental observation about post-processing is that there do not exist deterministic maps which can \textit{increase} information. Generically speaking, the information contained within a raw observation, relevant to identifying the backend, tends to be stripped out by the post-processing. A natural corollary is then that raw observations are informally maximal for backend identifiability.

    This reveals a very useful role for post-processing: it is perhaps the principal means by which the provider can reduce the amount of information leaked to the user in the returned transcripts, and thus improve routing anonymity. But to what end? Trivially, choosing a deliberately destructive map, say the constant map $\phi(\cdots)=1$, clearly reduces both the TV distance and Chernoff information to zero and so preserves \textit{anonymity} indefinitely, but there is a clear lack of \textit{utility} for the user. Unless the promised service is similarly trivial, this is an abuse of power by the provider in the backend/route identification game that abandons the practical motivation for such a setting.

    Introduce a \textit{utility map} $u$ representing the service that user is actually looking for. Subsequently, call the induced law $P_i(\dots,u)$ under this map the \textit{utility law}. Now, we can say that the provider's post-processing map map $\phi$ \textit{exactly preserves} the utility $u$ if the user is able to decode the output of $\phi$ to obtain her desired utility output (as if from $u$). We should note that $u$ is a quite abstract object, in much the same way that $\phi$ itself is left abstract, and indeed it may be natural to discuss entire classes of utility and post-processing maps, for example if the service is left reasonably flexible. We now have our third important result: that anonymity and utility have a trade-off relationship. 

    \begin{informal}{(Informal \cref{thm:utility_preserving_no_free_lunch}) Utility-Preserving No-Free-Lunch}
        If a deterministic post-processing map $\phi$ exactly preserves a utility $u$, then the induced laws $P_i(\dots,\phi)$ under $\phi$ must contain at least as much information as the utility laws $P_i(\dots,u)$. Any further loss of information necessarily degrades the utility.

        Hence, if the provider is restricted to choose only utility-preserving $\phi$, then he is only able to remove route-specific information which is extraneous to the promised utility $u$; he cannot remove route-specific information already encoded in the utility output under this restriction.
    \end{informal}

    To be clear, by `degrades the utility', we mean that the transcript that the user ultimately receives contains only partial information about the property of interest represented by the utility, and hence she can only partially reconstruct the property. For example, the provider may only be providing an estimate about the ground-state energy, with an error proportional to the degree to which the utility is not preserved. 

    We can further note approximate versions of utility preservation and the accompanying no-free-lunch argument of \cref{thm:utility_preserving_no_free_lunch}. The definitions of such versions should be delicately crafted with respect to the utility, as the approximation measure should yield a suitable interpretation for how the utility is degraded. We further discuss this problem and suggest solutions in Appendix \ref{appendix:post_processing_and_utility}.

\subsection{Identifiability as an Intermediate-Depth Phenomenon} \label{sec:intermediate_depth}

    An important observation can at this point be made about the role of depth $d$ in the user's workload circuits. In writing her ensemble of circuits $\mu_d$ with respect to a chosen depth $d$, we should ask what interval she should consider choosing $d$ to lie in. We have the following intuition:
    \begin{itemize}
    	\item At very low depths (or in  the noiseless setting), backends are indistinguishable because they implement very close to (or exactly) the same ideal circuit; and
    	\item Beyond some large depth, if noise becomes dominated by a common strongly mixing component, then any route-specific information is typically washed out, yielding indistinguishability; but
    	\item In the `intermediate-depth' window, enough noise will have accumulated to reveal backend-specific structure, yet not so much that everything has collapsed to universal fixed-point behaviour---this is the interesting regime.
    \end{itemize}

    To formalise this idea, in Appendix \ref{appendix:depth_window_principle} we record a pair of idealised theorems in a simplified Pauli-transfer/noisy-channel model. We prove these results by moving to Pauli-transfer-matrix (PTM) coordinates, wherein we can represent the effect of each noisy layer of a circuit as a linear contraction of the traceless Pauli components of the state. We assume that this contraction can be suitably decomposed into a common depolarising-like mixing component and a small backend-specific perturbation, then look at how the difference between two backends evolves with depth.
    
    \begin{informal}{(Informal \cref{thm:identifiability_degrades_at_large_depths} and \ref{thm:identifiability_grows_at_small_depths}) Identifiability Lies in Intermediate Depths}
        Denote circuit depth by $d$, and let $\lambda,\varepsilon\in(0,1)$ be noise-controlling parameters such that $\lambda+\varepsilon<1$. At shallow depths, backend-specific noise fingerprints initially grow at a rate $\Omega(d\lambda^{d-1})$. Over all depths, they are bounded by $\mathcal{O}(d(\lambda+\varepsilon)^{d-1})$, vanishing exponentially at large depths. Hence, optimal backend identifiability lies in some `intermediate' depth window.
    \end{informal}
    
    \begin{figure}[h]
    	\centering
    	\begin{tikzpicture}
    		\begin{axis}[
    			width=0.8\textwidth,
    			height=0.5\textwidth,
    			xmin=0, xmax=30,
    			ymin=0, ymax=0.9,
    			xlabel={Circuit Depth $d$},
    			ylabel={Expected Distinguishing Bias $\beta_{D_i,D_j}(d)$},
    			yticklabels={},
    			xticklabels={},
    			axis lines=left,
    			domain=0:30,
    			samples=100,
    			smooth,
    			legend style={
    				at={(0.95,0.7)},
    				anchor=east,
    				draw=none,
    				fill=none
    			},
    			legend cell align=left,
    			legend image post style={line width=1.2pt},
    			clip=false
    			]
    			
    			\pgfmathsetmacro{\lam}{0.75}      
    			\pgfmathsetmacro{\eps}{0.05}      
    			\pgfmathsetmacro{\clow}{0.15}     
    			\pgfmathsetmacro{\cup}{0.35}      
    			\pgfmathsetmacro{\dzero}{8}       
    			
    			\addplot[name path=upperfull, very thick, red]
    			{\cup*x*(\lam+\eps)^(x-1)};
    			\addlegendentry{\cref{thm:identifiability_degrades_at_large_depths}; $\mathcal{O}(d(\lambda+\varepsilon)^{d-1})$}
    			
    			\addplot[name path=upperpert, draw=none, domain=0:\dzero]
    			{\cup*x*(\lam+\eps)^(x-1)};
    			
    			\addplot[name path=lowerpert, very thick, black, domain=0:\dzero]
    			{\clow*x*(\lam)^(x-1)};
    			\addlegendentry{\cref{thm:identifiability_grows_at_small_depths}; $\Omega(d\lambda^{d-1})$}
    			
    			\addplot[very thick, black, dashed, domain=\dzero:30]
    			{\clow*x*(\lam)^(x-1)};
    			
    			\addplot[
    			fill=gray!25,
    			draw=none
    			] fill between[
    			of=lowerpert and upperpert
    			];
    			
    			\addplot[
    			black,
    			dashed,
    			thick
    			] coordinates {(\dzero,0) (\dzero,0.9)};
    			
    			\node[anchor=south] at (axis cs:\dzero,-0.09) {$d_0$};
    			
    		\end{axis}
    	\end{tikzpicture}
    	\caption{Illustration of the lower and upper bounds on expected distinguishing bias by circuit depth, as given by \cref{thm:identifiability_degrades_at_large_depths} and \cref{thm:identifiability_grows_at_small_depths}. The shaded region between the curves covers the perturbative range $0\leq d\leq d_0$ of depths satisfying \cref{thm:identifiability_grows_at_small_depths}'s additional assumptions.}
    	\label{fig:intermedaite_depth_phenomenon_illustration}
    \end{figure}
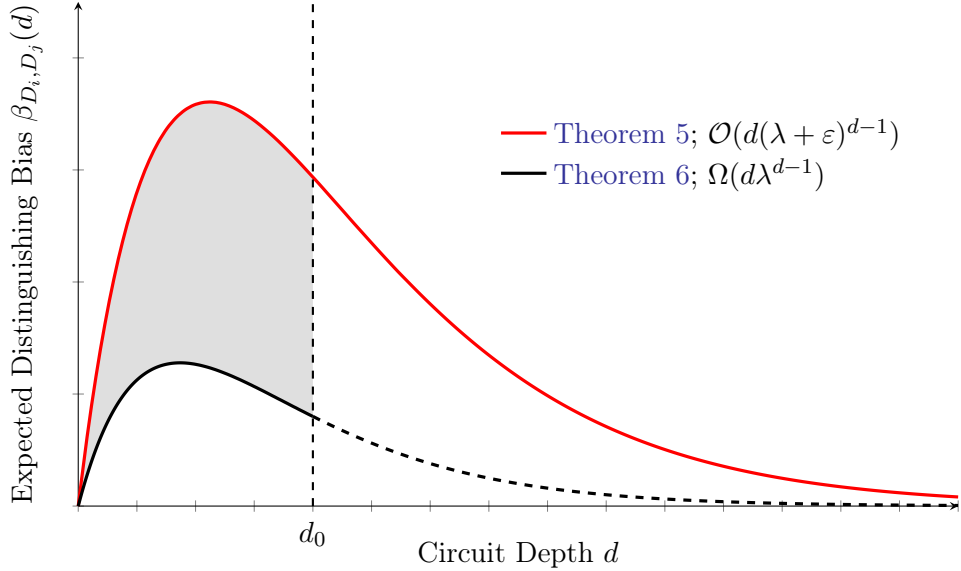
    
\subsection{Noise Channel Bounds and Relation to Noise Characterisation} \label{sec:noise_channels}

    We observe (in \cref{prop:distinguishing_bias_is_an_average_case_notion}) that the distinguishability of backends in our framework is governed by average-case notions of TV distance on only the states actually induced by the user's probing protocol, rather than any worst-case separation of the full channels over all possible inputs and ancillas. The natural channel distance to consider is thus not e.g. diamond norm but instead a particular \textit{workload-probed channel (pseudo-)distance}, denoted by $\delta_\mu(\cdot,\cdot)$ and defined with respect to the user's chosen ensemble $\mu$ of workload circuits.

    Usefully, the following \cref{prop:workload_probed_distance_tighter_than_diamon_norm} shows that the workload-probed channel (pseudo-)distance serves as a tighter bound on the distinguishing bias than the worst-case diamond norm, and is thus the more practical sufficient condition to target in the backend/route identification game. This is very much in agreement with recent ideas that these kinds of average-case distances are the more relevant quantities for NISQ-era applications than worst-case norms like diamond norm \cite{Maciejewski2023operational, Maciejewski2023exploring}.  

    \begin{informal}{(Informal \cref{prop:workload_probed_distance_tighter_than_diamon_norm})}
        Between any two backends, $D_i$ and $D_j$, the workload-probed (pseudo-)distance $\delta_\mu(\mathcal{N}_i,\mathcal{N}_j)$ between their associated noisy channels, $\mathcal{N}_i$ and $\mathcal{N}_j$, is a tighter bound than the diamond norm;
        \[
        \beta_{D_i,D_j}(\mu,m) \leq m\delta_\mu(\mathcal{N}_i,\mathcal{N}_j) \leq \frac{m}{2}\Big\| \mathcal{N}_i-\mathcal{N}_j \Big\|_\diamond \quad .
        \]
    \end{informal}

    A natural corollary, following \cref{prop:workload_probed_distance_tighter_than_diamon_norm}, is that two backends can be made indistinguishable by ensuring that the workload-probed (pseudo-)distance between their associated noisy channels is not too large (more immediately than ensuring the diamond norm is not too large). Recalling the pairwise argument of \cref{prop:pairwise_indistinguishability_is_sufficient_for_global_anonymity}, we can then give a sufficient condition for routing anonymity in the fully general setting---that no two pair of backends among the full set have large workload-probed channel (pseudo-)distance. Again, this is an experimentally-achievable condition in the NISQ setting that often cannot be said about diamond norm.

\section{Experimental Results} \label{sec:experiments}

In this section, we complement the above framework with experimental evidence that real finite-shot output distributions, from QPUs offered via current quantum cloud services, carry learnable route information. It is not the intent to reconstruct device noise models, or perform full tomography, but rather to instantiate the backend identification game of \cref{def:formal_identification_game} (and its extension to persistent routing via \cref{def:persistent_routing}) in a realistic setting and exemplify the risk to privacy.

Our experiments are run via AWS Braket on three QPUs: Rigetti's \textit{Ankaa-3} and IQM's \textit{Garnet}, both superconducting devices, and IonQ's \textit{Aria-1}, an ion-trap device. We categorise classification into two qualitatively different settings: the \textit{like-type} setting between the two superconducting devices, and the \textit{differing-type} setting generally between Ankaa-3 and Garnet for consistency. We also categorise experiments into two families of (5-qubit) probe circuits: varying depth circuits composed of Haar-random two-qubit gates arranged in an alternating brickwork pattern (``depth-varied''); and fixed GHZ preparation circuits repeated in time (``time-varied''). Broadly speaking, the former understands the backend identification problem in the case where the user cannot rely on hand-picked outcomes to probe deliberately, while the latter understands the extension to persistent routing and the temporal structure of noise signals.

For each circuit, we obtain a finite-shot histogram over computational-basis bitstrings. We then train simple classifiers on datasets of transcript representations (``features'' in the machine learning lexicon), including both for raw outcomes and several forms of problem-independent post-processing. Comparing how classification performance depends on feature type gleans empirical insight into how route-specific information may survive post-processing. Detailed experimental setup and a full account of our results is given in Appendices \ref{appendix:experiment_circuits}--\ref{appendix:classification_results}.

\begin{widefigure}
	\centering
	\begin{subfigure}[t]{0.48\linewidth}
		\centering
		\includegraphics[width=\columnwidth]{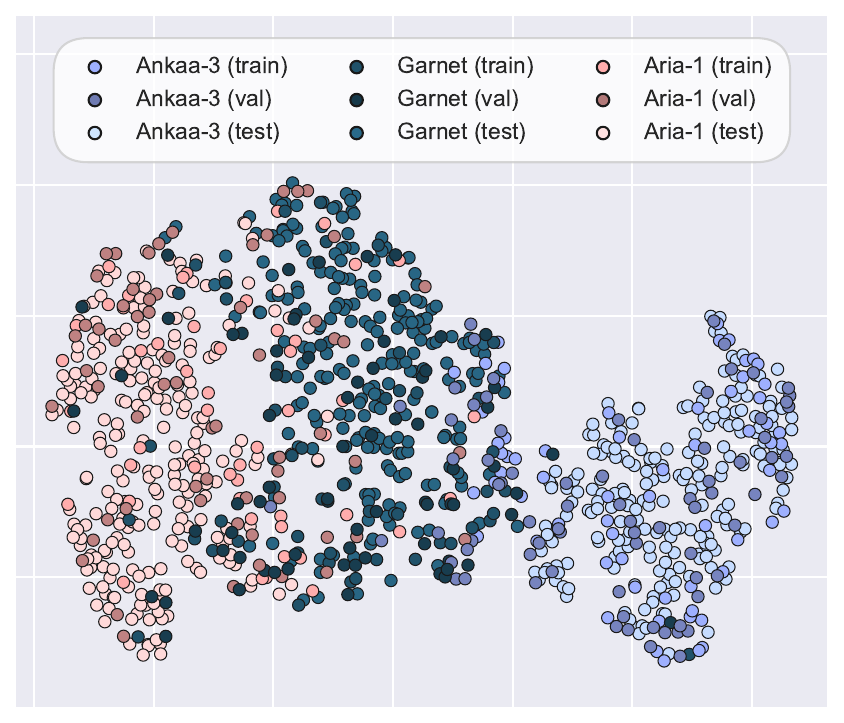}
		\caption{All-type classifier embeddings.}
		\label{fig:all_devices_raw_5}
	\end{subfigure}\hfill
	\begin{subfigure}[t]{0.48\linewidth}
		\centering
		\includegraphics[width=\columnwidth]{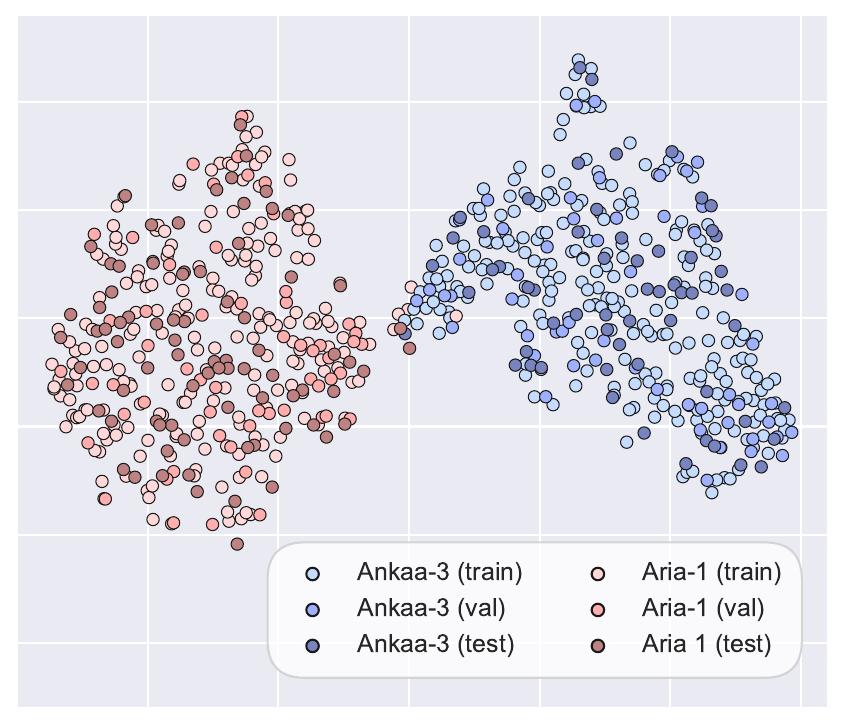}
		\caption{Differing-type classifier embeddings.}
		\label{fig:differ_type_raw_10}
	\end{subfigure}
	\caption{$t$-SNE dimensionality reduction plot visualising the 16-dimensional features extracted from the penultimate layer of the classifier, given raw depth-10 data. We plot the entire dataset, across all train-val-test splits to indicate the generality with which latent representations have been learned.}
	\label{fig:latent_representations}
\end{widefigure}

We find that, in the depth-varied experiments, backend labels can be learned well above random guessing, with like-type test classification accuracies around 87--90\%, and differing-type essentially perfect at around 96--98\% (each for raw features). High classification performance can be observed to survive several forms of post-processing; e.g. log-ratio features yield like-type accuracies in the range 66--77\% when pooling over depths. \cref{fig:latent_representations} visualises the latent representations learned by our classifiers in depth-varied experiments, indicating the separability that gives rise to the ease of classification on raw features. This latent space can be understood (loosely) as a subspace manifold of the full bitstring space in which intra-backend distances are minimised while inter-backend distances are maximised, giving a representation for noise signals \textit{with respect to the workload and backend set}.

Backend identification over a dataset of GHZ circuit data is a perfect 100\% across all features, as we expect from such structured circuits persistently probed very many times. This highlights the importance of the workload choice in the backend identification game---routing anonymity can be broken far more easily with carefully selected, highly structured circuits, so service providers should be conscious to the workload circuits they release transcript about, and how often they do so.

Further to backend classification, we also study temporal identifiability by grouping our time-varied GHZ data into batches and classifying by batch. A strong time-dependent fingerprint of devices can be detected among the superconducting devices at the batch-level, although classification within a batch---early vs late samples---is significantly more difficult. This could be interpreted practically as a low risk of leaking route-specific information over short-term persistent routing, particularly when coupled with the Haar-random observations that independent, random circuits probing in succession can mitigate the exponential decay of anonymity. At the very least, short-term temporal structure is dominated by long-term signal fluctuation in the QPUs we work on. 

\section{Discussion and Future Works}

We have introduced routing anonymity as the provider-side security notion for cloud quantum computation, together with the more familiar user-side counterpart in backend identifiability. It is understood that, in the NISQ setting, classical output distributions from quantum computations are not merely noisy approximations to the ideal distribution, but in fact laced with route-specific information that can be leveraged to violate anonymity of a service provider's routing choice if not sufficiently obscured by post-processing. Our framework packages this intuition into an operational game, which immediately reduces backend identification to classical hypothesis testing over transcript laws, and then elevate this into a persistent-routing setting, which shows that the rate of anonymity decay via independent, repeated probing is described by the Chernoff information between laws.

Post-processing emerges as the provider's principal mechanism for obscuring route-specific information to preserve the anonymity of their routing choices. We further formalise the degree to which a provider may remove this route-specific information from their released outputs via a no-free-lunch theorem attached to the promised utility of the service; anonymity may only be preserved up to the level of the law induced by the promised utility, and any further gains in anonymity necessarily degrade the service. All of this leads to an important practical question: which workloads are most identifying? From the user's perspective, this becomes a problem of designing workloads of probe circuits; from the provider's perspective, it becomes a problem of identifying exposing workloads, or of designing service interfaces that carefully choose how transcripts are released to remove route-specific information while preserving the promised utility. Such questions deserve incredible attention and care from cloud service providers for the foreseeable future, and we hope that this work provides a useful framework for uncovering answers in concrete applications.

We have also outlined a notable dependence on depth, proving that backend identifiability is naturally an intermediate-depth phenomenon: at shallow depths, too little backend-specific noise has accumulated, while at sufficiently large depths, common mixing can obscure route-specific structure. We do not design our experiments to explore this phenomenon directly, though nevertheless observe behaviour consistent with the principle and statements that we outline.

Our experimental observations support the theoretical results, finding that finite-shot classical outputs from real cloud quantum computing platforms can reliably identify routing choices well above random guessing. This identifiability exists even between backends implemented under the same physical platform architecture (namely, between superconducting devices), indicating the granularity and uniqueness of noise in current QPUs. While post-processing can be observed to improve routing anonymity in these experiments, we see that even workloads comprised of independent, random circuits can retain good identifiability, given enough persistence of routing.

Further time-dependent experimentation reveals aspects of the anonymity problem beyond those directly understood in the theoretical considerations of this work. Namely, strong classification between well-spaced batches of transcripts points to long-range temporal structure which itself may be learned effectively. This leads to practical questions about the duration of user-provider interactions, possibly indicating a preference for short-term bursts of probing to retain anonymity. 

Our blackbox model has a useful distinction from classical shadows \cite{Aaronson2018shadow, Huang2020classicalshadows}, which require a randomised measurement toolbox, e.g. random local Paulis or global Clifford measurements, and an inverse measurement channel. Our setting, and our experiments, take only computational-basis measurements and so should not, by themselves, construct a good classical shadow of the output state. The output object is thus closer to a \textit{circuit-conditioned `fingerprint'}. A clean diagnostic is to compare a classifier using only the unconditioned output object with one using the circuit-conditioned pair; if the former is near chance while the latter identifies the backend, then the signal is not shadow-like state reconstruction but rather a simpler input-output hardware `fingerprinting'.

\subsection{Extensions} 

Appendix \ref{appendix:forecasting} provides some preliminary exploration of noise \textit{forecasting}; i.e. predicting the presentation of noise in outcome distributions in future time steps by learning its evolution over a short prior. A sufficiently good forecasting model could have interesting applications, such as improving near-term error mitigation, retrospectively estimating ideal distributions, inferring the operational lifetime of a backend, detecting calibration changes, etc. However, such a model also introduces new privacy risks, such as giving users an improved capacity to pre-emptively adapt their workload when given repeated access (even without an explicitly adaptive access model). Security measures against such forecasting may look like, for example, randomly permuting the order of probe circuits.

An obvious progression of this work is to formalise concrete, provide-side, routing anonymity schemes; e.g. restricting the class of allowed workloads, reducing the number of shots or interactions, coarsening output histograms via post-processing, adding controlled classical noise, etc. Perhaps the main technical challenge is to make certifications of anonymity practical; worst-case channel distances are too strong and often experimentally inaccessible, so a workload-probed distinguishability metric seems most appropriate, but estimating or bounding it from finite data remains an open problem.

The suggested `quantum lock-and-key' setting of \cref{example:quantum_lock_and_key} may be a nice concrete application in which to develop ideas about such a scheme. In this setting, we have `locks' represented by backends, and `keys' by specific circuits about which the provider only releases decision bits indicating whether they open the requested lock. We can model this under our framework as a routing-anonymity game in which the utility is a precise decision function, and the security question is whether there exist workloads which leak the identity of any lock. A correctness condition should require valid keys to be accepted with high probability, and a soundness condition that invalid keys are rejected with high probability \textit{without leaking lock information}. We strongly encourage future work to develop the quantum lock-and-key idea into a complete scheme, including a concrete key distribution, acceptance rule, finite-query security bounds, and a better understanding for how the utility-anonymity trade-off can be modelled and utilised to obtain strong notions of completeness and soundness.

A possible restriction of our framework as proposed, which we have not yet discussed, is the ordering of events. Currently, we allow the provider to choose his route first and foremost, before seeing the workload submitted to him. This is perfectly acceptable in many circumstances (e.g. in all three of our examples in \cref{sec:example_applications}), but it necessarily forbids an adaptive scheduling policy on the provider's part that may well be extremely common in future quantum cloud practices. For example, if a provider has only one backend capable of running over some threshold number of qubits, then it could be impractical to turn away users submitting wide workload circuits simply because he has chosen the `wrong' backend. Instead, we generally foresee a setting wherein the user submits a circuit, and the provider then samples a backend according to some prior \textit{conditioned on the workload}. However, such a setting opens up the possibility for adversarial users to dictate the provider's backend choice to her possible identifying advantage (e.g. she might choose the widest circuit that she can imagine to be confident that the provider is forced to use his one capable backend). This kind of complication is conceptually challenging to overcome in the provider-side privacy perspective, and a key reason that we take the agnostic approach in selecting the backend label first, but is nonetheless important to understand for the security of quantum cloud computing.

\section*{Acknowledgements}
The authors thank James Mills for helpful discussions at the early stages of the project. BP acknowledges the financial support provided by the Jesus College Embiricos Trust Scholarship. MD acknowledges the support of the Quantum Advantage Pathfinder (QAP), with grant reference EP/X026167/1, and the
UK Engineering and Physical Sciences Research Council.  

\bibliographystyle{unsrtnat}
\bibliography{refs}

@PREAMBLE{
 "\providecommand{\noopsort}[1]{}" 
 # "\providecommand{\singleletter}[1]{#1}%" 
}

@article{Huang2020classicalshadows,
  author  = {Huang, Hsin-Yuan and Kueng, Richard and Preskill, John},
  title   = {Predicting many properties of a quantum system from very few measurements},
  journal = {Nature Physics},
  volume  = {16},
  pages   = {1050--1057},
  year    = {2020},
  doi     = {10.1038/s41567-020-0932-7}
}

@article{Aaronson2018shadow,
  author  = {Aaronson, Scott},
  title   = {Shadow Tomography of Quantum States},
  journal = {SIAM Journal on Computing},
  volume  = {49},
  number  = {5},
  pages   = {STOC18-368--STOC18-394},
  year    = {2020},
  doi     = {10.1137/18M120275X},
  note    = {Preliminary version in STOC 2018}
}

@article{Maciejewski2023operational,
  author  = {Maciejewski, Filip B. and Pucha{\l}a, Zbigniew and Oszmaniec, Micha{\l}},
  title   = {Operational Quantum Average-Case Distances},
  journal = {Quantum},
  volume  = {7},
  pages   = {1106},
  year    = {2023},
  doi     = {10.22331/q-2023-09-11-1106},
  eprint  = {2112.14283},
  archivePrefix = {arXiv},
  primaryClass = {quant-ph}
}

@article{Maciejewski2023exploring,
  author  = {Maciejewski, Filip B. and Pucha{\l}a, Zbigniew and Oszmaniec, Micha{\l}},
  title   = {Exploring Quantum Average-Case Distances: Proofs, Properties, and Examples},
  journal = {IEEE Transactions on Information Theory},
  year    = {2024},
  note    = {See also arXiv:2112.14284},
  eprint  = {2112.14284},
  archivePrefix = {arXiv},
  primaryClass = {quant-ph}
}

@article{Hamilton2020scalable,
  author  = {Hamilton, Kathleen E. and Kharazi, Tyler and Morris, Titus and McCaskey, Alexander J. and Bennink, Ryan S. and Pooser, Raphael C.},
  title   = {Scalable quantum processor noise characterization},
  journal = {Proceedings of the IEEE International Conference on Quantum Computing and Engineering},
  pages   = {430--440},
  year    = {2020},
  eprint  = {2006.01805},
  archivePrefix = {arXiv},
  primaryClass = {quant-ph}
}

@article{Magesan2011scalable,
  author  = {Magesan, Easwar and Gambetta, Jay M. and Emerson, Joseph},
  title   = {Scalable and Robust Randomized Benchmarking of Quantum Processes},
  journal = {Physical Review Letters},
  volume  = {106},
  pages   = {180504},
  year    = {2011},
  doi     = {10.1103/PhysRevLett.106.180504}
}

@article{BlumeKohout2017GST,
  author  = {Blume-Kohout, Robin and Gamble, John King and Nielsen, Erik and Rudinger, Kenneth and Mizrahi, Jonathan and Fortier, Kevin and Maunz, Peter},
  title   = {Demonstration of qubit operations below a rigorous fault tolerance threshold with gate set tomography},
  journal = {Nature Communications},
  volume  = {8},
  pages   = {14485},
  year    = {2017},
  doi     = {10.1038/ncomms14485}
}

@article{Harper2020efficient,
  author  = {Harper, Robin and Flammia, Steven T. and Wallman, Joel J.},
  title   = {Efficient learning of quantum noise},
  journal = {Nature Physics},
  volume  = {16},
  pages   = {1184--1188},
  year    = {2020},
  doi     = {10.1038/s41567-020-0992-8}
}

@article{VandenBerg2023pec,
  author  = {van den Berg, Ewout and Minev, Zlatko K. and Kandala, Abhinav and Temme, Kristan},
  title   = {Probabilistic error cancellation with sparse {Pauli}--{Lindblad} models on noisy quantum processors},
  journal = {Nature Physics},
  volume  = {19},
  pages   = {1116--1121},
  year    = {2023},
  doi     = {10.1038/s41567-023-02042-2}
}

@book{CoverThomas2006,
  author    = {Cover, Thomas M. and Thomas, Joy A.},
  title     = {Elements of Information Theory},
  publisher = {Wiley},
  year      = {2006},
  edition   = {2}
}

@misc{AWSBraketDocs,
  author       = {{Amazon Web Services}},
  title        = {{Amazon Braket}: Supported devices},
  howpublished = {\url{https://docs.aws.amazon.com/braket/latest/developerguide/braket-devices.html}},
  note         = {Accessed May 2026},
  year         = {2026}
}

@misc{AWSBraketService,
  author       = {{Amazon Web Services}},
  title        = {{Amazon Braket}: Quantum cloud computing service},
  howpublished = {\url{https://aws.amazon.com/braket/}},
  note         = {Accessed May 2026},
  year         = {2026}
}

@article{Resch2021benchmarking,
  author  = {Resch, S. and Gutierrez, A. and Smith, J. and Rech, P. and Tannu, S. S. and Han, S. and Tiwari, D. and Diniz, P. and Chong, F. T.},
  title   = {Benchmarking Quantum Computers and the Impact of Quantum Noise},
  journal = {ACM Computing Surveys},
  volume  = {54},
  number  = {7},
  pages   = {1--35},
  year    = {2021},
  doi     = {10.1145/3464420}
}

@article{Childs2005secure,
  author  = {Childs, Andrew M.},
  title   = {Secure Assisted Quantum Computation},
  journal = {Quantum Information and Computation},
  volume  = {5},
  number  = {6},
  pages   = {456--466},
  year    = {2005},
  eprint  = {quant-ph/0111046},
  archivePrefix = {arXiv}
}

@inproceedings{Broadbent2009UBQC,
  author    = {Broadbent, Anne and Fitzsimons, Joseph and Kashefi, Elham},
  title     = {Universal Blind Quantum Computation},
  booktitle = {Proceedings of the 50th Annual {IEEE} Symposium on Foundations of Computer Science},
  pages     = {517--526},
  year      = {2009},
  doi       = {10.1109/FOCS.2009.36},
  eprint    = {0807.4154},
  archivePrefix = {arXiv},
  primaryClass = {quant-ph}
}

@article{FitzsimonsKashefi2017VUBQC,
  author  = {Fitzsimons, Joseph F. and Kashefi, Elham},
  title   = {Unconditionally Verifiable Blind Quantum Computation},
  journal = {Physical Review A},
  volume  = {96},
  number  = {1},
  pages   = {012303},
  year    = {2017},
  doi     = {10.1103/PhysRevA.96.012303},
  eprint  = {1203.5217},
  archivePrefix = {arXiv},
  primaryClass = {quant-ph}
}

@article{Gheorghiu2019verification,
  author  = {Gheorghiu, Alexandru and Kapourniotis, Theodoros and Kashefi, Elham},
  title   = {Verification of Quantum Computation: An Overview of Existing Approaches},
  journal = {Theory of Computing Systems},
  volume  = {63},
  pages   = {715--808},
  year    = {2019},
  doi     = {10.1007/s00224-018-9872-3},
  eprint  = {1709.06984},
  archivePrefix = {arXiv},
  primaryClass = {quant-ph}
}

@inproceedings{Mahadev2018verification,
  author    = {Mahadev, Urmila},
  title     = {Classical Verification of Quantum Computations},
  booktitle = {Proceedings of the 59th Annual {IEEE} Symposium on Foundations of Computer Science},
  pages     = {259--267},
  year      = {2018},
  doi       = {10.1109/FOCS.2018.00033},
  eprint    = {1804.01082},
  archivePrefix = {arXiv},
  primaryClass = {quant-ph}
}

@article{Ma2022QEnclave,
  author  = {Ma, Yuhang and Huang, Yipeng and Ferracin, Samuele and Harper, Robin and Tannu, Swamit S.},
  title   = {{QEnclave}: A Practical Solution for Secure Quantum Cloud Computing},
  journal = {npj Quantum Information},
  volume  = {8},
  pages   = {128},
  year    = {2022},
  doi     = {10.1038/s41534-022-00612-5}
}

@article{Chaum1981mix,
  author  = {Chaum, David L.},
  title   = {Untraceable Electronic Mail, Return Addresses, and Digital Pseudonyms},
  journal = {Communications of the ACM},
  volume  = {24},
  number  = {2},
  pages   = {84--90},
  year    = {1981},
  doi     = {10.1145/358549.358563}
}

@article{Reed1998onion,
  author  = {Reed, Michael G. and Syverson, Paul F. and Goldschlag, David M.},
  title   = {Anonymous Connections and Onion Routing},
  journal = {{IEEE} Journal on Selected Areas in Communications},
  volume  = {16},
  number  = {4},
  pages   = {482--494},
  year    = {1998},
  doi     = {10.1109/49.668972}
}

@inproceedings{Dingledine2004Tor,
  author    = {Dingledine, Roger and Mathewson, Nick and Syverson, Paul},
  title     = {{Tor}: The Second-Generation Onion Router},
  booktitle = {Proceedings of the 13th {USENIX} Security Symposium},
  year      = {2004}
}

@article{Cross2019quantumvolume,
  author  = {Cross, Andrew W. and Bishop, Lev S. and Sheldon, Sarah and Nation, Paul D. and Gambetta, Jay M.},
  title   = {Validating Quantum Computers Using Randomized Model Circuits},
  journal = {Physical Review A},
  volume  = {100},
  pages   = {032328},
  year    = {2019},
  doi     = {10.1103/PhysRevA.100.032328},
  eprint  = {1811.12926},
  archivePrefix = {arXiv},
  primaryClass = {quant-ph}
}

@article{Erhard2019cyclebenchmarking,
  author  = {Erhard, Alexander and Wallman, Joel J. and Postler, Lukas and Meth, Michael and Stricker, Roman and Martinez, Esteban A. and Schindler, Philipp and Monz, Thomas and Emerson, Joseph and Blatt, Rainer},
  title   = {Characterizing Large-Scale Quantum Computers via Cycle Benchmarking},
  journal = {Nature Communications},
  volume  = {10},
  pages   = {5347},
  year    = {2019},
  doi     = {10.1038/s41467-019-13068-7}
}

@misc{nguyen2024quantumcloudcomputingreview,
	title={Quantum Cloud Computing: A Review, Open Problems, and Future Directions}, 
	author={Hoa T. Nguyen and Prabhakar Krishnan and Dilip Krishnaswamy and Muhammad Usman and Rajkumar Buyya},
	year={2024},
	eprint={2404.11420},
	archivePrefix={arXiv},
	primaryClass={cs.ET},
	url={https://arxiv.org/abs/2404.11420}, 
}

@misc{caro2022learning,
  author       = {Caro, Matthias C.},
  title        = {Learning Quantum Processes and Hamiltonians via the Pauli Transfer Matrix},
  year         = {2022},
  eprint       = {2212.04471},
  archivePrefix= {arXiv},
  primaryClass = {quant-ph}
}

@misc{greenbaum2015introduction,
  author       = {Greenbaum, Daniel},
  title        = {Introduction to Quantum Gate Set Tomography},
  year         = {2015},
  eprint       = {1509.02921},
  archivePrefix= {arXiv},
  primaryClass = {quant-ph}
}

@article{helsen2019spectral,
  author       = {Helsen, Jonas and Battistel, Francesco and Terhal, Barbara M.},
  title        = {Spectral quantum tomography},
  journal      = {npj Quantum Information},
  volume       = {5},
  number       = {1},
  pages        = {74},
  year         = {2019},
  doi          = {10.1038/s41534-019-0189-0},
  eprint       = {1904.00177},
  archivePrefix= {arXiv},
  primaryClass = {quant-ph}
}

@article{raginsky2002strictly,
  author       = {Raginsky, Maxim},
  title        = {Strictly contractive quantum channels and physically realizable quantum computers},
  journal      = {Physical Review A},
  volume       = {65},
  number       = {3},
  pages        = {032306},
  year         = {2002},
  doi          = {10.1103/PhysRevA.65.032306},
  eprint       = {quant-ph/0105141},
  archivePrefix= {arXiv}
}

@article{hiai2016contraction,
  author       = {Hiai, Fumio and Ruskai, Mary Beth},
  title        = {Contraction coefficients for noisy quantum channels},
  journal      = {Journal of Mathematical Physics},
  volume       = {57},
  number       = {1},
  pages        = {015211},
  year         = {2016},
  doi          = {10.1063/1.4936215},
  eprint       = {1508.03551},
  archivePrefix= {arXiv},
  primaryClass = {quant-ph}
}

@article{wallman2016noise,
  author       = {Wallman, Joel J. and Emerson, Joseph},
  title        = {Noise tailoring for scalable quantum computation via randomized compiling},
  journal      = {Physical Review A},
  volume       = {94},
  number       = {5},
  pages        = {052325},
  year         = {2016},
  doi          = {10.1103/PhysRevA.94.052325},
  eprint       = {1512.01098},
  archivePrefix = {arXiv},
  primaryClass = {quant-ph}
}

@article{ware2021experimental,
  author       = {Ware, Matthew and Ribeill, Guilhem and Rist{\`e}, Diego and Ryan, Colm A. and Johnson, Blake and da Silva, Marcus P.},
  title        = {Experimental Pauli-frame randomization on a superconducting qubit},
  journal      = {Physical Review A},
  volume       = {103},
  number       = {4},
  pages        = {042604},
  year         = {2021},
  doi          = {10.1103/PhysRevA.103.042604},
  eprint       = {1803.01818},
  archivePrefix = {arXiv},
  primaryClass = {quant-ph}
}

@article{schumann2024emergence,
  author       = {Schumann, Marco and Wilhelm, Frank K. and Ciani, Alessandro},
  title        = {Emergence of noise-induced barren plateaus in arbitrary layered noise models},
  journal      = {Quantum Science and Technology},
  volume       = {9},
  number       = {4},
  pages        = {045019},
  year         = {2024},
  doi          = {10.1088/2058-9565/ad6285},
  eprint       = {2310.08405},
  archivePrefix = {arXiv},
  primaryClass = {quant-ph}
}

@inproceedings{chia2020classical,
  title={Classical verification of quantum computations with efficient verifier},
  author={Chia, Nai-Hui and Chung, Kai-Min and Yamakawa, Takashi},
  booktitle={Theory of Cryptography Conference},
  pages={181--206},
  year={2020},
  organization={Springer}
}

@inproceedings{bartusek2022succinct,
  title={Succinct classical verification of quantum computation},
  author={Bartusek, James and Kalai, Yael Tauman and Lombardi, Alex and Ma, Fermi and Malavolta, Giulio and Vaikuntanathan, Vinod and Vidick, Thomas and Yang, Lisa},
  booktitle={Annual International Cryptology Conference},
  pages={195--211},
  year={2022},
  organization={Springer}
}

@inproceedings{zhang2021succinct,
  title={Succinct blind quantum computation using a random oracle},
  author={Zhang, Jiayu},
  booktitle={Proceedings of the 53rd Annual ACM SIGACT Symposium on Theory of Computing},
  pages={1370--1383},
  year={2021}
}

@article{kalai2026classically,
  title={How to Classically Verify a Quantum Cat without Killing It},
  author={Kalai, Yael Tauman and Khurana, Dakshita and Raizes, Justin},
  journal={arXiv preprint arXiv:2602.09282},
  year={2026}
}

@article{leichtle2021verifying,
  title={Verifying BQP computations on noisy devices with minimal overhead},
  author={Leichtle, Dominik and Music, Luka and Kashefi, Elham and Ollivier, Harold},
  journal={PRX Quantum},
  volume={2},
  number={4},
  pages={040302},
  year={2021},
  publisher={APS}
}

@inproceedings{badertscher2020security,
  title={Security limitations of classical-client delegated quantum computing},
  author={Badertscher, Christian and Cojocaru, Alexandru and Colisson, L{\'e}o and Kashefi, Elham and Leichtle, Dominik and Mantri, Atul and Wallden, Petros},
  booktitle={International Conference on the Theory and Application of Cryptology and Information Security},
  pages={667--696},
  year={2020},
  organization={Springer}
}

@article{poshtvan2025selectively,
  title={Selectively blind quantum computation},
  author={Poshtvan, Abbas and Lapiha, Oleksandra and Doosti, Mina and Leichtle, Dominik and Music, Luka and Kashefi, Elham},
  journal={arXiv preprint arXiv:2504.17612},
  year={2025}
}

@inproceedings{Dunjko2014composable,
  author    = {Vedran Dunjko and Joseph F. Fitzsimons and Christopher Portmann and Renato Renner},
  title     = {Composable Security of Delegated Quantum Computation},
  booktitle = {Advances in Cryptology -- ASIACRYPT 2014},
  series    = {Lecture Notes in Computer Science},
  volume    = {8874},
  pages     = {406--425},
  publisher = {Springer},
  year      = {2014},
  doi       = {10.1007/978-3-662-45608-8_22},
  eprint    = {1301.3662},
  archivePrefix = {arXiv},
  primaryClass  = {quant-ph}
}

@misc{Aharonov2017interactive,
  author        = {Dorit Aharonov and Michael Ben-Or and Elad Eban and Urmila Mahadev},
  title         = {Interactive Proofs for Quantum Computations},
  year          = {2017},
  eprint        = {1704.04487},
  archivePrefix = {arXiv},
  primaryClass  = {quant-ph}
}

@inproceedings{GheorghiuVidick2019RSP,
  author    = {Alexandru Gheorghiu and Thomas Vidick},
  title     = {Computationally-Secure and Composable Remote State Preparation},
  booktitle = {Proceedings of the 60th IEEE Annual Symposium on Foundations of Computer Science},
  series    = {FOCS 2019},
  pages     = {1024--1033},
  publisher = {IEEE Computer Society},
  year      = {2019},
  doi       = {10.1109/FOCS.2019.00066},
  eprint    = {1904.06320},
  archivePrefix = {arXiv},
  primaryClass  = {quant-ph}
}

@inproceedings{Mahadev2018homomorphic,
  author    = {Urmila Mahadev},
  title     = {Classical Homomorphic Encryption for Quantum Circuits},
  booktitle = {Proceedings of the 59th IEEE Annual Symposium on Foundations of Computer Science},
  series    = {FOCS 2018},
  pages     = {332--338},
  publisher = {IEEE Computer Society},
  year      = {2018},
  doi       = {10.1109/FOCS.2018.00039},
  eprint    = {1708.02130},
  archivePrefix = {arXiv},
  primaryClass  = {quant-ph}
}

@inproceedings{Alagic2018vQFHE,
  author    = {Gorjan Alagic and Yfke Dulek and Christian Schaffner and Florian Speelman},
  title     = {Quantum Fully Homomorphic Encryption with Verification},
  booktitle = {Advances in Cryptology -- EUROCRYPT 2018},
  series    = {Lecture Notes in Computer Science},
  publisher = {Springer},
  pages     = {438--467},
  year      = {2018},
  doi       = {10.1007/978-3-319-70694-8_16},
  eprint    = {1708.09156},
  archivePrefix = {arXiv},
  primaryClass  = {quant-ph}
}

@misc{Martina2021learningNoiseFingerprint,
  author        = {Stefano Martina and Lorenzo Buffoni and Stefano Gherardini and Filippo Caruso},
  title         = {Learning the Noise Fingerprint of Quantum Devices},
  year          = {2021},
  eprint        = {2109.11405},
  archivePrefix = {arXiv},
  primaryClass  = {quant-ph}
}

@inproceedings{RoyGhoshGhosh2025forensics,
  author    = {Rupshali Roy and Archisman Ghosh and Swaroop Ghosh},
  title     = {Forensics of Transpiled Quantum Circuits},
  booktitle = {Proceedings of the Great Lakes Symposium on VLSI 2025},
  series    = {GLSVLSI 2025},
  pages     = {354--359},
  publisher = {Association for Computing Machinery},
  year      = {2025},
  doi       = {10.1145/3716368.3735240},
  eprint    = {2412.18939},
  archivePrefix = {arXiv},
  primaryClass  = {quant-ph}
}

@article{eisert2020quantum,
  title={Quantum certification and benchmarking},
  author={Eisert, Jens and Hangleiter, Dominik and Walk, Nathan and Roth, Ingo and Markham, Damian and Parekh, Rhea and Chabaud, Ulysse and Kashefi, Elham},
  journal={Nature Reviews Physics},
  volume={2},
  number={7},
  pages={382--390},
  year={2020},
  publisher={Nature Publishing Group UK London}
}

@article{Priestley_2026,
	doi = {10.1088/2058-9565/ae4cc6},
	url = {https://doi.org/10.1088/2058-9565/ae4cc6},
	year = {2026},
	month = {mar},
	publisher = {IOP Publishing},
	volume = {11},
	number = {2},
	pages = {025018},
	author = {Priestley, Ben and Wallden, Petros},
	title = {A practically scalable approach to the closest vector problem for sieving via QAOA with fixed angles},
	journal = {Quantum Science and Technology}
}

@misc{priestley2025lqg,
	title={Finite-Dimensional ZX-Calculus for Loop Quantum Gravity}, 
	author={Ben Priestley},
	year={2025},
	eprint={2511.15966},
	archivePrefix={arXiv},
	primaryClass={gr-qc},
	url={https://arxiv.org/abs/2511.15966}, 
}

@inproceedings{Shor-1994,
	author={Shor, P.W.},
	booktitle={Proceedings 35th Annual Symposium on Foundations of Computer Science}, 
	title={Algorithms for quantum computation: discrete logarithms and factoring}, 
	year={1994},
	volume={},
	number={},
	pages={124-134},
	doi={10.1109/SFCS.1994.365700}
}

@article{proos2003shor,
	title={Shor's discrete logarithm quantum algorithm for elliptic curves},
	author={Proos, John and Zalka, Christof},
	journal={arXiv preprint quant-ph/0301141},
	year={2003}
}

@inproceedings{roetteler2017quantum,
	title={Quantum resource estimates for computing elliptic curve discrete logarithms},
	author={Roetteler, Martin and Naehrig, Michael and Svore, Krysta M and Lauter, Kristin},
	booktitle={International Conference on the Theory and Application of Cryptology and Information Security},
	pages={241--270},
	year={2017},
	organization={Springer}
}

@inproceedings{grassl2016applying,
	title={Applying Grover’s algorithm to AES: quantum resource estimates},
	author={Grassl, Markus and Langenberg, Brandon and Roetteler, Martin and Steinwandt, Rainer},
	booktitle={International Workshop on Post-Quantum Cryptography},
	pages={29--43},
	year={2016},
	organization={Springer}
}

@article{kaplan2015quantum,
	title={Quantum differential and linear cryptanalysis},
	author={Kaplan, Marc and Leurent, Ga{\"e}tan and Leverrier, Anthony and Naya-Plasencia, Mar{\'\i}a},
	journal={arXiv preprint arXiv:1510.05836},
	year={2015}
}

@inproceedings{kaplan2016breaking,
	title={Breaking symmetric cryptosystems using quantum period finding},
	author={Kaplan, Marc and Leurent, Ga{\"e}tan and Leverrier, Anthony and Naya-Plasencia, Mar{\'\i}a},
	booktitle={Annual international cryptology conference},
	pages={207--237},
	year={2016},
	organization={Springer}
}

@inproceedings{chailloux2021lattice,
	title={Lattice sieving via quantum random walks},
	author={Chailloux, Andr{\'e} and Loyer, Johanna},
	booktitle={International Conference on the Theory and Application of Cryptology and Information Security},
	pages={63--91},
	year={2021},
	organization={Springer}
}

@article{laarhoven2015finding,
	title={Finding shortest lattice vectors faster using quantum search},
	author={Laarhoven, Thijs and Mosca, Michele and Van De Pol, Joop},
	journal={Designs, Codes and Cryptography},
	volume={77},
	number={2},
	pages={375--400},
	year={2015},
	publisher={Springer}
}

@article{aspuru2005simulated,
	title={Simulated quantum computation of molecular energies},
	author={Aspuru-Guzik, Al{\'a}n and Dutoi, Anthony D and Love, Peter J and Head-Gordon, Martin},
	journal={Science},
	volume={309},
	number={5741},
	pages={1704--1707},
	year={2005},
	publisher={American Association for the Advancement of Science}
}

@article{peruzzo2014variational,
	title={A variational eigenvalue solver on a photonic quantum processor},
	author={Peruzzo, Alberto and McClean, Jarrod and Shadbolt, Peter and Yung, Man-Hong and Zhou, Xiao-Qi and Love, Peter J and Aspuru-Guzik, Al{\'a}n and O’brien, Jeremy L},
	journal={Nature communications},
	volume={5},
	number={1},
	pages={4213},
	year={2014},
	publisher={Nature Publishing Group UK London}
}

@article{kandala2017hardware,
	title={Hardware-efficient variational quantum eigensolver for small molecules and quantum magnets},
	author={Kandala, Abhinav and Mezzacapo, Antonio and Temme, Kristan and Takita, Maika and Brink, Markus and Chow, Jerry M and Gambetta, Jay M},
	journal={nature},
	volume={549},
	number={7671},
	pages={242--246},
	year={2017},
	publisher={Nature Publishing Group}
}

@article{reiher2017elucidating,
	title={Elucidating reaction mechanisms on quantum computers},
	author={Reiher, Markus and Wiebe, Nathan and Svore, Krysta M and Wecker, Dave and Troyer, Matthias},
	journal={Proceedings of the national academy of sciences},
	volume={114},
	number={29},
	pages={7555--7560},
	year={2017},
	publisher={National Academy of Sciences}
}

@article{google2020hartree,
	title={Hartree-Fock on a superconducting qubit quantum computer},
	author={Google AI Quantum and Collaborators*† and Arute, Frank and Arya, Kunal and Babbush, Ryan and Bacon, Dave and Bardin, Joseph C and Barends, Rami and Boixo, Sergio and Broughton, Michael and Buckley, Bob B and others},
	journal={Science},
	volume={369},
	number={6507},
	pages={1084--1089},
	year={2020},
	publisher={American Association for the Advancement of Science}
}

@article{von2021quantum,
	title={Quantum computing enhanced computational catalysis},
	author={von Burg, Vera and Low, Guang Hao and H{\"a}ner, Thomas and Steiger, Damian S and Reiher, Markus and Roetteler, Martin and Troyer, Matthias},
	journal={Physical Review Research},
	volume={3},
	number={3},
	pages={033055},
	year={2021},
	publisher={APS}
}

@article{clinton2024towards,
	title={Towards near-term quantum simulation of materials},
	author={Clinton, Laura and Cubitt, Toby and Flynn, Brian and Gambetta, Filippo Maria and Klassen, Joel and Montanaro, Ashley and Piddock, Stephen and Santos, Raul A and Sheridan, Evan},
	journal={Nature Communications},
	volume={15},
	number={1},
	pages={211},
	year={2024},
	publisher={Nature Publishing Group UK London}
}

@article{bernien2017probing,
	title={Probing many-body dynamics on a 51-atom quantum simulator},
	author={Bernien, Hannes and Schwartz, Sylvain and Keesling, Alexander and Levine, Harry and Omran, Ahmed and Pichler, Hannes and Choi, Soonwon and Zibrov, Alexander S and Endres, Manuel and Greiner, Markus and others},
	journal={Nature},
	volume={551},
	number={7682},
	pages={579--584},
	year={2017},
	publisher={Nature Publishing Group UK London}
}

@article{lloyd1996universal,
	title={Universal quantum simulators},
	author={Lloyd, Seth},
	journal={Science},
	volume={273},
	number={5278},
	pages={1073--1078},
	year={1996},
	publisher={American Association for the Advancement of Science}
}

@article{martinez2016real,
	title={Real-time dynamics of lattice gauge theories with a few-qubit quantum computer},
	author={Martinez, Esteban A and Muschik, Christine A and Schindler, Philipp and Nigg, Daniel and Erhard, Alexander and Heyl, Markus and Hauke, Philipp and Dalmonte, Marcello and Monz, Thomas and Zoller, Peter and others},
	journal={Nature},
	volume={534},
	number={7608},
	pages={516--519},
	year={2016},
	publisher={Nature Publishing Group UK London}
}

@article{ebadi2021quantum,
	title={Quantum phases of matter on a 256-atom programmable quantum simulator},
	author={Ebadi, Sepehr and Wang, Tout T and Levine, Harry and Keesling, Alexander and Semeghini, Giulia and Omran, Ahmed and Bluvstein, Dolev and Samajdar, Rhine and Pichler, Hannes and Ho, Wen Wei and others},
	journal={Nature},
	volume={595},
	number={7866},
	pages={227--232},
	year={2021},
	publisher={Nature Publishing Group UK London}
}

@article{atas20212,
	title={SU (2) hadrons on a quantum computer via a variational approach},
	author={Atas, Yasar Y and Zhang, Jinglei and Lewis, Randy and Jahanpour, Amin and Haase, Jan F and Muschik, Christine A},
	journal={Nature communications},
	volume={12},
	number={1},
	pages={6499},
	year={2021},
	publisher={Nature Publishing Group UK London}
}

@article{meth2025simulating,
	title={Simulating two-dimensional lattice gauge theories on a qudit quantum computer},
	author={Meth, Michael and Zhang, Jinglei and Haase, Jan F and Edmunds, Claire and Postler, Lukas and Jena, Andrew J and Steiner, Alex and Dellantonio, Luca and Blatt, Rainer and Zoller, Peter and others},
	journal={Nature Physics},
	volume={21},
	number={4},
	pages={570--576},
	year={2025},
	publisher={Nature Publishing Group UK London}
}

@article{li2019quantum,
	title={Quantum spacetime on a quantum simulator},
	author={Li, Keren and Li, Youning and Han, Muxin and Lu, Sirui and Zhou, Jie and Ruan, Dong and Long, Guilu and Wan, Yidun and Lu, Dawei and Zeng, Bei and others},
	journal={Communications Physics},
	volume={2},
	number={1},
	pages={122},
	year={2019},
	publisher={Nature Publishing Group UK London}
}

@article{east2021spin,
	title={Spin-networks in the ZX-calculus},
	author={East, Richard DP and Martin-Dussaud, Pierre and Van de Wetering, John},
	journal={arXiv preprint arXiv:2111.03114},
	year={2021}
}

@article{mielczarek2019spin,
	title={Spin foam vertex amplitudes on quantum computer—preliminary results},
	author={Mielczarek, Jakub},
	journal={Universe},
	volume={5},
	number={8},
	pages={179},
	year={2019},
	publisher={MDPI}
}

@article{van2023experimental,
	title={Experimental simulation of loop quantum gravity on a photonic chip},
	author={van der Meer, Reinier and Huang, Zichang and Anguita, Malaquias Correa and Qu, Dongxue and Hooijschuur, Peter and Liu, Hongguang and Han, Muxin and Renema, Jelmer J and Cohen, Lior},
	journal={npj Quantum Information},
	volume={9},
	number={1},
	pages={32},
	year={2023},
	publisher={Nature Publishing Group UK London}
}

@article{Zhu-2025a,
	title={Observation of quantum Darwinism and the origin of classicality with superconducting circuits},
	author={Zhu, Zitian and Salice, Kiera and Touil, Akram and Bao, Zehang and Song, Zixuan and Zhang, Pengfei and Li, Hekang and Wang, Zhen and Song, Chao and Guo, Qiujiang and others},
	journal={Science Advances},
	volume={11},
	number={31},
	pages={eadx6857},
	year={2025},
	publisher={American Association for the Advancement of Science}
}

@article{farhi2014quantum,
	title={A quantum approximate optimization algorithm},
	author={Farhi, Edward and Goldstone, Jeffrey and Gutmann, Sam},
	journal={arXiv preprint arXiv:1411.4028},
	year={2014}
}

@article{lucas2014ising,
	title={Ising formulations of many NP problems},
	author={Lucas, Andrew},
	journal={Frontiers in physics},
	volume={2},
	pages={74887},
	year={2014},
	publisher={Frontiers}
}

@article{ebadi2022quantum,
	title={Quantum optimization of maximum independent set using Rydberg atom arrays},
	author={Ebadi, Sepehr and Keesling, Alexander and Cain, Madelyn and Wang, Tout T and Levine, Harry and Bluvstein, Dolev and Semeghini, Giulia and Omran, Ahmed and Liu, J-G and Samajdar, Rhine and others},
	journal={Science},
	volume={376},
	number={6598},
	pages={1209--1215},
	year={2022},
	publisher={American Association for the Advancement of Science}
}

@article{rebentrost2014quantum,
	title={Quantum support vector machine for big data classification},
	author={Rebentrost, Patrick and Mohseni, Masoud and Lloyd, Seth},
	journal={Physical review letters},
	volume={113},
	number={13},
	pages={130503},
	year={2014},
	publisher={APS}
}

@article{schuld2019quantum,
	title={Quantum machine learning in feature Hilbert spaces},
	author={Schuld, Maria and Killoran, Nathan},
	journal={Physical review letters},
	volume={122},
	number={4},
	pages={040504},
	year={2019},
	publisher={APS}
}

@article{havlivcek2019supervised,
	title={Supervised learning with quantum-enhanced feature spaces},
	author={Havl{\'\i}{\v{c}}ek, Vojt{\v{e}}ch and C{\'o}rcoles, Antonio D and Temme, Kristan and Harrow, Aram W and Kandala, Abhinav and Chow, Jerry M and Gambetta, Jay M},
	journal={Nature},
	volume={567},
	number={7747},
	pages={209--212},
	year={2019},
	publisher={Nature Publishing Group UK London}
}

@article{liu2021rigorous,
	title={A rigorous and robust quantum speed-up in supervised machine learning},
	author={Liu, Yunchao and Arunachalam, Srinivasan and Temme, Kristan},
	journal={Nature physics},
	volume={17},
	number={9},
	pages={1013--1017},
	year={2021},
	publisher={Nature Publishing Group UK London}
}

@article{huang2021power,
	title={Power of data in quantum machine learning},
	author={Huang, Hsin-Yuan and Broughton, Michael and Mohseni, Masoud and Babbush, Ryan and Boixo, Sergio and Neven, Hartmut and McClean, Jarrod R},
	journal={Nature communications},
	volume={12},
	number={1},
	pages={2631},
	year={2021},
	publisher={Nature Publishing Group UK London}
}

@article{huang2022quantum,
	title={Quantum advantage in learning from experiments},
	author={Huang, Hsin-Yuan and Broughton, Michael and Cotler, Jordan and Chen, Sitan and Li, Jerry and Mohseni, Masoud and Neven, Hartmut and Babbush, Ryan and Kueng, Richard and Preskill, John and others},
	journal={Science},
	volume={376},
	number={6598},
	pages={1182--1186},
	year={2022},
	publisher={American Association for the Advancement of Science}
}

@article{yin2025experimental,
	title={Experimental quantum-enhanced kernel-based machine learning on a photonic processor},
	author={Yin, Zhenghao and Agresti, Iris and De Felice, Giovanni and Brown, Douglas and Toumi, Alexis and Pentangelo, Ciro and Piacentini, Simone and Crespi, Andrea and Ceccarelli, Francesco and Osellame, Roberto and others},
	journal={Nature Photonics},
	volume={19},
	number={9},
	pages={1020--1027},
	year={2025},
	publisher={Nature Publishing Group UK London}
}

@article{Golec_2024,
	title={Quantum cloud computing: Trends and challenges},
	volume={2},
	ISSN={2949-9488},
	url={http://dx.doi.org/10.1016/j.ject.2024.05.001},
	DOI={10.1016/j.ject.2024.05.001},
	journal={Journal of Economy and Technology},
	publisher={Elsevier BV},
	author={Golec, Muhammed and Hatay, Emir Sahin and Golec, Mustafa and Uyar, Murat and Golec, Merve and Gill, Sukhpal Singh},
	year={2024},
	month=Nov, pages={190–199} 
}

@inproceedings{smith2023fast,
	title={Fast fingerprinting of cloud-based nisq quantum computers},
	author={Smith, Kaitlin N and Viszlai, Joshua and Seifert, Lennart Maximilian and Baker, Jonathan M and Szefer, Jakub and Chong, Frederic T},
	booktitle={2023 IEEE International Symposium on Hardware Oriented Security and Trust (HOST)},
	pages={1--12},
	year={2023},
	organization={IEEE}
}

@inproceedings{mi2021short,
	title={Short paper: Device-and locality-specific fingerprinting of shared nisq quantum computers},
	author={Mi, Allen and Deng, Shuwen and Szefer, Jakub},
	booktitle={Proceedings of the 10th International Workshop on Hardware and Architectural Support for Security and Privacy},
	pages={1--6},
	year={2021}
}

@misc{mutolo2025quantumcomputerfingerprintingusing,
	title={Quantum Computer Fingerprinting using Error Syndromes}, 
	author={Vincent Mutolo and Devon Campbell and Quinn Manning and Henri Witold Dubourg and Ruibin Lyu and Simha Sethumadhavan and Daniel Rubenstein and Salvatore Stolfo},
	year={2025},
	eprint={2506.16614},
	archivePrefix={arXiv},
	primaryClass={quant-ph},
	url={https://arxiv.org/abs/2506.16614}, 
}

@article{Wu-2024,
	author = {Wu, Jindi and Hu, Tianjie and Li, Qun},
	title = {Q-ID: Lightweight Quantum Network Server Identification Through Fingerprinting},
	year = {2024},
	issue_date = {Sept. 2024},
	publisher = {IEEE Press},
	volume = {38},
	number = {5},
	issn = {0890-8044},
	url = {https://doi.org/10.1109/MNET.2024.3400893},
	doi = {10.1109/MNET.2024.3400893},
	journal = {Netwrk. Mag. of Global Internetwkg.},
	month = sep,
	pages = {146–152},
	numpages = {7}
}

@misc{coupel2025securityvulnerabilitiesquantumcloud,
	title={Security Vulnerabilities in Quantum Cloud Systems: A Survey on Emerging Threats}, 
	author={Justin Coupel and Tasnuva Farheen},
	year={2025},
	eprint={2504.19064},
	archivePrefix={arXiv},
	primaryClass={cs.CR},
	url={https://arxiv.org/abs/2504.19064}, 
}

\newcommand{\appendixpart}[1]{
  \clearpage
  {\noindent\huge\bfseries #1\par}
  \noindent\rule{\textwidth}{0.4pt}\par
  \vspace{3ex}
  \addcontentsline{toc}{part}{#1}
}

\newpage
\appendix

\part*{Appendix}

\startcontents[appendix]
\begingroup
\makeatletter
\let\oldaddvspace\addvspace
\renewcommand{\addvspace}[1]{\oldaddvspace{0.25em}}
\let\oldlpart\l@part
\renewcommand{\l@part}[2]{%
  \noindent\rule{\linewidth}{0.4pt}\par
  {\large\bfseries\oldlpart{#1}{#2}}%
  \vspace{1em}
}
\makeatother
\printcontents[appendix]{}{0}{\setcounter{tocdepth}{2}}
\endgroup

\appendixpart{Framework and Theoretical Results}

\section{Formal Framework} \label{appendix:framework}
Fix $n\in\mathbb{N}$ qubits, and let $\Omega_n:=\{0,1\}^n$ denote the set of computational-basis $n$-bit string outcomes. We use the convention $\Delta(\cdot)$ to denote sets of probability distributions; so write $\Delta(\Omega_n)$ to denote the set of probability distributions on $n$-bit strings, or $\Delta_m(\Omega_n)$ to denote the set of empirical histograms on $\Omega_n$ that can be obtained over $m$ draws.

Let $\mathcal{C}_d$ be a family of depth-$d$ circuits over $n$ qubits, and for each $C\in\mathcal{C}_d$, we will write $U_C$ to denote its ideal unitary action, producing an ideal output state $\rho_C:=U_C|0^n\rangle\langle 0^n|U_C^\dagger$. We will denote by $\mathcal{C}:=\bigcup_{d\geq 0}\mathcal{C}_d$ the collection of all executable quantum circuits over $n$ qubits.
	
Consider a finite set of quantum backends (or ``routes'') $\mathcal{D}=\{D_1,\dots,D_k\}$, and for each $D_i\in\mathcal{D}$, associate a family of effective noisy channels $\mathcal{N}_i:=\{\mathcal{N}_{i,C}:C\in\mathcal{C}\}$ where,
\[
\mathcal{N}_{i,C}:\mathcal{D}((\mathbb{C}^2)^{\otimes n})\to\mathcal{D}((\mathbb{C}^2)^{\otimes n})\quad ,
\]
is a completely positive and trace-preserving (CPTP) map. Here, we are taking the quite general perspective that allows effective noise to depend explicitly on the circuit, since real effective noise depends on compilation, layout, native gates, calibration, crosstalk, control context, time, etc. We represent measurement (without loss of generality, in the computational basis) by,
\[
\mathcal{M}:\mathcal{D}((\mathbb{C}^2)^{\otimes n})\to \Delta(\Omega_n)\quad ,
\]
and hence, for a given backend $D_i$ and circuit $C$, the idealised output law is given by,
\[
p_{i,C}:=\mathcal{M}(\mathcal{N}_{i,C}(\rho_C)) \in \Delta(\Omega_n) \quad .
\]

In practice, we will only execute the circuit over $m$ independent shots, producing random outcomes $Y_1,\dots,Y_m\sim p_{i,C}$. We can then compute from these an empirical histogram like,
\[
\hat{p}_{i,C}(x) := \frac{1}{m}\sum_{j=1}^m\mathbf{1}\{Y_j=x\} \in \Delta_m(\Omega_n) \quad .
\]

Now let $\phi:\mathcal{C}\times\Delta_m(\Omega_n)\to\mathcal{X}$ be a deterministic post-processing map taking raw observations $(C,\hat{p}_{i,C})$ into a transcript space $\mathcal{X}$ (which we leave abstract for now). For brevity, we may sometimes refer to the raw observation space simply as $\mathcal{Y}:=\mathcal{C}\times\Delta_m(\Omega_n)$.

\subsection{A Game-based definition for Backend/Route Identifiability}

We now have the necessary objects to define a game of \textit{backend/route identification} between a user and provider. Informally, the game is as follows: the provider privately selects a quantum backend, associated with a noisy channel, and asks the user to choose a probe circuit to be routed through this channel. Dutifully, the provider then executes the circuit on the backend for a finite precision of shots, before returning to the user (a post-processed transcript of) the outcome histogram. The user's task is to guess which backend produced the outcome returned to her---if only she could contain her excitement. 

More formally, we provide the following definition for a single round of the game.

\begin{definition}[Backend/Route Identification Game] \label{def:formal_identification_game}
    Fix publicly known: circuit ensembles $\{\mu_d\}_{d\geq 0}$ indexed by depth $d$, backend set $\mathcal{D}=\{D_1,\dots,D_k\}$ with prior $\pi\in\Delta(k)$, finite shot count $m\in\mathbb{N}$, and post-processing map $\phi:\mathcal{Y}\to\mathcal{X}$. The game is then played as follows:
        \begin{enumerate}
            \item Provider randomly samples a hidden route $I\sim \pi$.
            \item User chooses a circuit depth $d\in\mathbb{N}$ and randomly samples a depth-$d$ workload circuit $C\sim\mu_d$, then submits it to the provider.
            \item Provider executes $C$ over $m$ independent shots on backend $D_I$, producing a raw empirical distribution $\hat{p}_{I,C}\in\Delta_m(\Omega_n)$.
            \item Provider produces a transcript $X=\phi(C,\hat{p}_{I,C})$, then returns it to the user.
            \item User observes $X$ and outputs a guess $\tilde{I}$, winning the game if $\tilde{I}=I$.
        \end{enumerate}
\end{definition}

\begin{definition}[Persistent Routing] \label{def:persistent_routing}
    The backend/route identification game is extended over $T$ rounds by repeating steps 2 and 3; that is, in each round $t=1,\dots,T$: the route $I$ remains fixed (i.e. we have a ``persistent route''), but a new probe circuit $C_t$ may be sampled from the same fixed ensemble $\mu_d$. Ultimately, a $T$-length transcript $(X_1,\dots,X_T)$ is returned and a guess is made.

    This work restricts to \textit{passive} probing, meaning that $C_1,\dots,C_T$ are all sampled and submitted before any execution, so $C_{t+1}$ cannot be biased by knowledge of $X_1,\dots,X_{t}$.
\end{definition}

\subsection{Describing Identifiability and Anonymity}

\begin{definition}[Induced/Raw Laws] \label{def:distributions}
    Denote by $P_i(\mu,m,\phi)$ the \textit{induced law} on the post-processed transcript space $\mathcal{X}$ such that, for every measurable set $A\subseteq\mathcal{X}$,
    \[
    P_i(\mu,m,\phi)(A)
    :=
    \int_{C\in\mathcal C}
    \Pr\!\Big[\phi(C,\hat p_{i,C})\in A\Big]\ \mu(dC)\quad .
    \]
    We can also denote by $Q_i(\mu,m)$ the \textit{raw law} on the raw observation space $\mathcal{Y}=\mathcal{C}\times\Delta_m(\Omega_n)$, and then write $P_i(\mu,m,\phi)=\phi_\# Q_i(\mu,m)$ as the forward pass of $Q_i(\mu,m)$ under the map $\phi$.

    For brevity, we will often omit the function signatures when the context is clear, and write $P_i$ and $Q_i$ respectively for these laws.
\end{definition}

\begin{definition}[Success Probability] \label{def:success}
	Let $g:\mathcal{X}\to[k]$ act as a decision rule that produces a guessed backend label from an observed transcript $X$. The \textit{success probability} for $g$, under a prior $\pi=(\pi_1,\dots,\pi_k)$ on backends, is given by,
	\[
	p_s(g) := \sum_{i=1}^{k}\pi_i\ \mathbb{P}_{X\sim P_i}\Big[ g(X)=i \Big] \quad ,  
	\]
	and the \textit{optimal success probability} over all such $g$ is denoted $p^\star_s:=\sup_g p_s(g)$.
\end{definition}

\begin{definition} [Excess Advantage / Distinguishing Bias]
	In the general setting, with prior $\pi=(\pi_1,\dots,\pi_k)$ over $k$ backends, the best probability we can obtain without the use of a transcript is $p_b:=\max_{i\in[k]}\pi_i$, made by blindly guessing the most likely backend. We can then define an \textit{excess advantage} beyond blind guessing by,
	\[
	\mathrm{Adv} := \frac{p_s^\star-p_b}{1-p_b} \quad .   
	\]
	In the uniform-binary setting, $\mathcal{D}=\{D_1,D_2\}$ and $\pi=(1/2, 1/2)$, so $p_b=1/2$ and the above excess advantage expression recovers the familiar \textit{distinguishing bias};
	\[
	\beta_{D_1,D_2} := \frac{p_s^\star-\frac{1}{2}}{1-\frac{1}{2}} = 2p_s^\star - 1 \quad .
	\]
\end{definition}

We can season the backend/route identification game with cryptographic implications to change its flavour from \textit{identification} to \textit{anonymisation}. Where before we asked how well a user can recover the provider's routing choice, we now ask the equivalent provider-side security question: how small is user's excess advantage, with respect to the kinds of transcript we expose?

\begin{definition}[$\varepsilon$-Anonymity] \label{def:anonymity}
	Begin with the backend identification game played with a set of routes $\mathcal{D}=\{D_1,\dots,D_k\}$. For security parameter $\varepsilon\in(0,1)$, such a setting can be called $\varepsilon$-anonymous with respect to a class $\Phi$ of allowed post-processing maps if,
	\[
	p^\star_s \leq p_b + (1-p_b)\varepsilon \quad ,
	\]
	for all $\phi\in\Phi$. Equivalently, $\sup_{\phi\in\Phi}\mathrm{Adv}\leq\varepsilon$, or $\sup_{\phi\in\Phi}\beta_{D_1,D_2}\leq\varepsilon$ in the uniform-binary setting.
\end{definition}

\section{Formal Theoretical Results} \label{appendix:theoretical_results}
This section will present the formal theoretical results summarised informally in \cref{sec:overview_framework_and_theoretical_results}. We begin by showing that backend identification is exactly the problem of classical statistical hypothesis testing in \cref{appendix:identification_reduces_to_hypothesis}, yielding an exact characterisation of optimal binary identification with total variation distance, a pairwise sufficient condition for general routing anonymity, and describing the rate at which anonymity decays via Chernoff information. Then, in \cref{appendix:post_processing_and_utility}, we focus on the role of post-processing as a means for removing backend-specific information. We then introduce utility maps and state a precise utility-anonymity trade-off relationship that establishes fundamental bounds on anonymity. Finally, we relate statistical distinguishability of routes to the underlying noisy quantum dynamics in \cref{appendix:noise_channel_bounds}, showing that identification relies on average-case channel distances rather than any worst-case distance, and hence defining a new pseudo-distance for (families of) channels within our framework. We then give some simple sufficient conditions for routing anonymity with respect to this average-case distance.

\subsection{Reducing Identification to Classical Statistical Testing} \label{appendix:identification_reduces_to_hypothesis}

\begin{theorem}[Backend Identifiability Reduces to Hypothesis Testing] \label{thm:backend_identifiability_reduces_to_hypothesis_testing}
    Consider the backend identification game, and denote by $X\in\mathcal{X}$ the returned (singleton) transcript about which the user decides on a backend label $i\in[k]$. The user's decision is exactly equivalent to hypothesis testing with hypotheses $H_i : X\sim P_i$, where $P_i$ is the known induced law of $D_i$.
	
	Specifically, in the uniform-binary setting between backends $\mathcal{D}=\{D_1,D_2\}$ with prior $\pi=(1/2,1/2)$, over all binary decision rules $g:\mathcal{X}\to\{1,2\}$, we have that,
	\[
	p^\star_s := \sup_g p_s(g) = \frac{1}{2}\Big(1 + \mathrm{TV}(P_1,P_2)\Big) \quad .
	\]
	Equivalently, we can write the distinguishing bias $\beta_{D_1,D_2}=\mathrm{TV}(P_1,P_2)$.
\end{theorem}
\begin{proof}
    Observe that any binary decision rule $g$ splits the transcript space $\mathcal{X}$ into a region $A\subseteq\mathcal{X}$ (corresponding to guessing $D_1$) and its complement $A^c$ (resp. $D_2$). Then, under a uniform prior,
    \[
    p_s(g) = \frac{1}{2}P_1(A)+\frac{1}{2}P_2(A^c) = \frac{1}{2}\Big(1 + P_1(A) - P_2(A)\Big) \quad ,
    \]
    hence, maximising the success looks like,
    \[
    p_s^\star := \sup_g p_s(g) = \frac{1}{2}\left( 1 + \sup_{A\subseteq\mathcal{X}} P_1(A) - P_2(A) \right) \equiv \frac{1}{2}\Big(1 + \mathrm{TV}(P_1,P_2)\Big) \quad .
    \]
    and the bias $\beta_{D_1,D_2}$ easily follows by definition.
\end{proof}

In the following proposition, we observe that pairwise distinguishability, measured in the familiar total variation distance, can be seen to bound the excess advantage over any arbitrary backend set. Hence, an arbitrary collection of routes and prior can be worst-case bounded by another route identification game in the uniform-binary setting.

\begin{prop}[Pairwise Indistinguishability is Sufficient for Global Anonymity] \label{prop:pairwise_indistinguishability_is_sufficient_for_global_anonymity}
    For an arbitrary route set with prior $\pi$, choose $r^\star\in\arg\max_i\pi_i$ to be a most probable route. Then for any deterministic post-processing map, we have,
	\begin{equation} \label{eq:pairwise_indistinguishability_is_sufficient_for_global_anonymity}
	   p_s^\star\leq \pi_{r^\star} + \sum_{i\neq r^\star} \pi_i \ \mathrm{TV}(P_i, P_{r^\star}) \quad .
	\end{equation}
	In particular, $\mathrm{Adv}\leq\max_{i\neq j}\mathrm{TV}(P_i, P_j)$. Hence, in the general setting, $\varepsilon$-anonymity is guaranteed if pairwise distinguishability between any pair of distinct routes is at most $\varepsilon$ under any deterministic post-processing map in the class of interest.
\end{prop}
\begin{proof}
	Let $g:\mathcal{X}\to[k]$ be any decision rule, with preimages $A_i:=g^{-1}(\{i\})$ for each $i\in[k]$. By \cref{def:success}, we quantify the success probability for $g$ by,
	\[
		p_s(g) = \sum_{i\in[k]}\pi_i P_i(A_i) = \pi_{r^\star}P_{r^\star}(A_{r^\star}) + \sum_{i\neq r^\star}\pi_iP_i(A_i) \quad .
	\]
	For each $i\neq r^\star$, we have by the definition of total variation, for every $A_i$,
	\[
	\begin{split}
		P_i(A_i) - P_{r^\star}(A_i) &\leq \mathrm{TV}(P_i,P_{r^\star}) \quad ,\\
		\therefore\ P_i(A_i) &\leq P_{r^\star}(A_i) + \mathrm{TV}(P_i,P_{r^\star}) \quad ,
	\end{split}
	\]
	and so substituting this gives,
	\[
		p_s(g) \leq \pi_{r^\star}P_{r^\star}(A_{r^\star}) 
		+ \sum_{i\neq r^\star}\pi_i P_{r^\star}(A_i)
		+ \sum_{i\neq r^\star} \pi_i\cdot\mathrm{TV}(P_i,P_{r^\star}) \quad .
	\]
	Then, using $\pi_i\leq\pi_{r^*}$ for every $i$, we can see that,
	\[
	\pi_{r^\star}P_{r^\star}(A_{r^\star}) + \sum_{i\neq r^\star}\pi_i P_{r^\star}(A_i) \leq \pi_{r^\star}\sum_{i\in[k]}P_{r^\star}(A_i)=\pi_{r^\star} \quad ,
	\]
	since the $\{A_i\}_i$ partition $\mathcal{X}$. Hence,
	\[
	p_s(g) \leq \pi_{r^\star} + \sum_{i\neq r^\star} \pi_i\cdot \mathrm{TV}(P_i,P_{r^\star}) \quad . \tag{\ref{eq:pairwise_indistinguishability_is_sufficient_for_global_anonymity}}
	\]
	
	Now we prove the second statement. Note first that taking the supremum over the above to obtain $p_s^\star=\sup_{g}p_s(g)$ preserves the upper bound. Then second observe that $p_b=\max_{i\in[k]}\pi_i=\pi_{r^\star}$. So, using the first statement of Eq. (\ref{eq:pairwise_indistinguishability_is_sufficient_for_global_anonymity}), we have,
	\[
	\mathrm{Adv} := \frac{p_s^\star - p_b}{1 - p_b} \leq \sum_{i\neq r^\star}\frac{\pi_i}{1-\pi_{r^\star}}\cdot \mathrm{TV}(P_i,P_{r^\star}) \quad .
	\]
	Since the coefficients in this expression form a probability distribution on $[k]\backslash\{r^\star\}$, the weighted average is clearly bounded by the maximum term;
	\[
	\sum_{i\neq r^\star}\frac{\pi_i}{1-\pi_{r^\star}}\cdot \mathrm{TV}(P_i,P_{r^\star}) \leq \max_{i\neq r^\star}\ \mathrm{TV}(P_i,P_{r^\star}) \quad .
	\]
	Finally, since $[k]\backslash\{r^\star\}\subset[k]$, we can loosen the bound a little further to yield,
	\[
	\max_{i\neq r^\star}\ \mathrm{TV}(P_i,P_{r^\star}) \leq \max_{i\neq j}\ \mathrm{TV}(P_j,P_j) \quad .
	\]
	Therefore, having $\max_{i\neq j} \mathrm{TV}(P_i,P_j)\leq\varepsilon$ necessarily gives $\mathrm{Adv}\leq\varepsilon$, so if no pair of routes $i,j$ have $\mathrm{TV}(P_i,P_j)>\varepsilon$ under any post-processing map in the class, then $\varepsilon$-anonymity is satisfied.
\end{proof}

The above \cref{prop:pairwise_indistinguishability_is_sufficient_for_global_anonymity} gives a practical sufficient condition: if every pair of routes is hard to distinguish under the transcript maps that the provider returns, then the whole routing system is anonymous by \cref{def:anonymity}. This is particularly useful for certification because it avoids needing to solve the full $k$-ary decision problem.

To extend \cref{thm:backend_identifiability_reduces_to_hypothesis_testing}, we can reduce extended backend identification, played over $T$ rounds, to hypothesis testing between the product distributions over the $T$-length transcript space $\mathcal{X}^{\otimes T}$. Asymptotically, we can then say how quickly anonymity---the converse to distinguishability in this context---breaks down as we prolong the game over more rounds, with a rate of anonymity decay controlled exponentially by the Chernoff information between induced distributions.

\begin{theorem}[Persistent Routing Reduces to Chernoff Testing] \label{thm:persistent_routing_reduces_to_chernoff_testing}
    Extend the backend identification game over $T$ rounds of persistent routing, and denote by $(X_1,\dots,X_T)\in\mathcal{X}^{\otimes T}$ the returned $T$-length transcript about which the user decides on a backend label $i\in[k]$. The user's decision is exactly equivalent to hypothesis testing with hypotheses $H_i : (X_1,\dots,X_T)\sim P_i^{\otimes T}$, where $P_i$ is the known induced distribution of backend $D_i$.
	
	Specifically, in the uniform-binary setting between backends $\mathcal{D}=\{D_1,D_2\}$ with prior $\pi=(1/2,1/2)$, over all binary decision rules $g:\mathcal{X}^{\otimes T}\to\{1,2\}$, we have that,
	\[
	p^\star_{s}(T) := \sup_g p_s(g) = \frac{1}{2}\Big(1 + \mathrm{TV}\Big(P_1^{\otimes T},P_2^{\otimes T}\Big)\Big) \quad , 
	\]
	and the optimal success probability $p^\star_{s}(T)$ over the $T$ rounds satisfies,
	\[
	\lim_{T\to\infty}-\frac{1}{T}\log\big(1-p_{s}^\star(T)\big) = D_{\mathrm{Ch}}(P_1, P_2) \quad ,
	\]
	where $D_{\mathrm{Ch}}(\cdot, \cdot)$ denotes the Chernoff information between probability distributions. In particular, when $D_{\mathrm{Ch}}(P_1, P_2)$ is nonzero and finite, we have,
	\[
	1-p^\star_{s}(T) = \exp\Big( -T\ D_{\mathrm{Ch}}(P_1, P_2) + o(T) \Big) \quad .
	\]
\end{theorem}
\begin{proof}
    Recall that we have a conditional independence enforced by \cref{def:persistent_routing}, which says that the transcripts $X_1,\dots,X_T$ are drawn i.i.d. from some $P_i$. So, for any measurable region $A_1\times\cdots\times A_T\subseteq\mathcal{X}^{\otimes T}$, we have,
    \[
    \mathbb{P}_i\Big[(X_1,\dots,X_T)\in A_1\times \cdots\times A_T\Big] = \prod_{t=1}^T \mathbb{P}[X_t\in A_t]=\prod_{t=1}^T P_i(A_t) \quad ,    
    \]
    thus the law of the full $T$-length transcript is the product $P_i^{\otimes T}$. Then, by the same argument as in \cref{thm:backend_identifiability_reduces_to_hypothesis_testing}, the reduction to hypothesis testing follows.

    By the standard Chernoff theorem for Bayesian binary hypothesis testing of i.i.d. observations \cite[see e.g.][chapter 11]{CoverThomas2006}, the optimal error probability, given by $1-p_{s,T}^\star$, satisfies, 
    \[
    \lim_{T\to\infty}-\frac{1}{T}\log\big(1-p_{s,T}^\star\big) = D_{\mathrm{Ch}}(P_1, P_2) \quad ,
    \]
    which we can rearrange like,
    \[
    1-p^\star_{s,T} = \exp\Big( -T\cdot D_{\mathrm{Ch}}(P_1, P_2) + o(T) \Big) \quad .
    \]
    for some small term $o(T)$ collecting practically negligible sublinear corrections in the exponent.
\end{proof}

\subsection{Post-Processing and the Utility-Anonymity Trade-Off} \label{appendix:post_processing_and_utility}

In the quantum cloud setting, post-processing plays a very interesting role that deserves careful and deliberate consideration. Naively, we can imagine the provider executing a given circuit on behalf of the user and providing the raw bitstring counts, much like current cloud-based quantum services do in the present day. This represents the most flexible service, with the provider taking somewhat of a `no-questions-asked' stance and allowing the user to get on with whatever she wishes to do with her outcomes. In the near future, however, as the quantum cloud setting expands into an increasingly commercialised and specialised enterprise, providers may wish to perform only a particular `service' for the user. Instead of a raw empirical histogram over bitstrings, providers may prefer to release only a feature about these bitstring outcomes; e.g. if the user is interested in the ground-state energy of a particular Hamiltonian, the provider may perform quantum phase estimation on her behalf, in which case the feature of interest looks like the peak value in the histogram.

In this service-centric paradigm, post-processing acts not only as a means to directly compute features of interest on behalf of the user, but also as the principal mechanism by which the provider can remove backend-specific information from the released transcript. We formalise this with the following \cref{prop:post_processing_cannot_increase_information} in the uniform-binary case, well understood already in classical statistics.

\begin{prop}[Post-Processing Cannot Increase Information] \label{prop:post_processing_cannot_increase_information}
    Let $Q_1,Q_2$ be the raw observation laws associated with backends $D_1,D_2$, and let $\phi:\mathcal{Y}\to\mathcal{X}$ be any (measurable) deterministic post-processing map. Then,
    \[
    \mathrm{TV}(P_1,P_2) := \mathrm{TV}(\phi_\# Q_1,\phi_\# Q_2) \leq \mathrm{TV}(Q_1,Q_2) \quad .
    \]
\end{prop}
\begin{proof}
    Recall, by definition, that the total variation distance between the induced distributions $P_1,P_2$ is given by,
    \[
    \mathrm{TV}(P_1,P_2) := \sup_{A\subseteq\mathcal{X}}\Big|P_1(A)-P_2(A)\Big| \quad ,
    \]
    with a supremum over all measurable subsets $A$ of the transcript space $\mathcal{X}$. Now, using that $P_i=\phi_\# Q_i$, we have for every such $A$,
    \[
    P_i(A) = Q_i(\phi^{-1}(A)) \quad .
    \]
    Denote by $\mathcal{Y}_\phi := \{\phi^{-1}(A) : A\subseteq\mathcal{X}\}$ the set of all raw-space events that can be expressed as preimages of transcript-space events under $\phi$. Then we have,
    \[
    \mathrm{TV}(P_1,P_2) = \sup_{B\in\mathcal{Y}_\phi}\Big|Q_1(B)-Q_2(B)\Big| \quad .
    \]
    On the other hand, we also have, by definition,
    \[
    \mathrm{TV}(Q_1,Q_2) = \sup_{B\in\mathcal{Y}}\Big|Q_1(B)-Q_2(B)\Big| \quad ,
    \]
    with a supremum over all measurable subsets $B$ of the raw space $\mathcal{Y}$. Importantly, observe that $\mathcal{Y}_\phi\subseteq\mathcal{Y}$, since the assumption that $\phi$ is measurable tells us that every preimage $\phi^{-1}(\cdot)$ is itself a measurable raw event (in $\mathcal{Y}$). Hence, using the elementary fact that, if $S_1\subseteq S_2$, then $\sup_{x\in S_1}f(x)\leq \sup_{x\in S_2}f(x)$ for any real-valued $f$,
    \[
    \sup_{B\in\mathcal{Y}_\phi}\Big|Q_1(B)-Q_2(B)\Big| \leq \sup_{B\in\mathcal{Y}}\Big|Q_1(B)-Q_2(B)\Big| \quad ,
    \]
    and so, for any measurable $\phi$, we have $\mathrm{TV}(P_1,P_2) \leq \mathrm{TV}(Q_1,Q_2)$ as required.
\end{proof}

More generally, note the following corollary which states that arbitrary post-processing cannot increase distinguishability.

\begin{corollary} \label{cor:post_processing_doesnt_increase_distinguishability}
	If a map $\phi'$ is a `further' post-processing of another map $\phi$, meaning that there exists some $\psi$ such that $\phi'=\psi\circ\phi$, then,
	\[
	\beta_{D_1,D_2}(\mu,m,\phi') \leq \beta_{D_1,D_2}(\mu,m,\phi)\quad .
	\]
\end{corollary}

We can easily see that this result already implies that raw observations are informationally maximal for backend identification, since any deterministic post-processing map $\phi$ can clearly be called a `further' post-processing of the identity map.

We should now place increased scrutiny on the choice of post-processing map; of course, the provider could choose e.g. a trivial constant function so that $P_1=P_2$ and thus achieve indefinite \textit{anonymity}, but there is obviously a lack of \textit{utility} in this `service'. Our next natural progression is then a question of trade-off between anonymity and utility through the choice of post-processing map: the provider should look to remove route-specific information as much as possible without also removing task-relevant information. In this work, we model the promised service via the following \textit{utility map}---the promise of a service is the promise that a particular computation is being performed on the bitstrings, and this utility should be preserved in that its codomain remains accessible.

\begin{definition}[Utility Map / Utility Law]
	Let $\mathcal{U}$ be a measurable utility space. A \textit{utility map} is a measurable deterministic map $u:\mathcal{Y}\to\mathcal{U}$, where $\mathcal{Y}=\mathcal{C}\times\Delta_m(\Omega_n)$ is the raw observation space.
	
	For a given backend/route $D_i$, the \textit{utility law} is the distribution $P_i(\dots,u):=u_\# Q_i$ induced by the utility map $u$, as in \cref{def:distributions}, as the forward pass of the raw distribution $Q_i$ under $u$.
\end{definition}

\begin{definition}[Utility Preservation] \label{def:exact_utility_preservation}
	Let $\phi:\mathcal{Y}\to\mathcal{X}$ be any deterministic post-processing map, and let $u:\mathcal{Y}\to\mathcal{U}$ be any utility map. We say that $\phi$ \textit{preserves the utility} $u$ exactly if there exists a (measurable) decoder $\psi:\mathcal{X}\to\mathcal{U}$ such that $u=\psi\circ\phi$.
\end{definition}

Under \cref{def:exact_utility_preservation}, the provider is then constrained (perhaps contractually obligated) to choose a class $\Phi$ of post-processing maps which exactly preserve the promised utility $u$. In the following theorem, we formalise a kind of no-free-lunch evident under this restriction: that if the requested utility is itself route-sensitive, then hiding the route necessarily means changing, coarsening, randomising, or otherwise degrading the utility.

\begin{theorem}[Utility-Preserving No-Free-Lunch] \label{thm:utility_preserving_no_free_lunch}
    Let $u:\mathcal{Y}\to\mathcal{U}$ and $\phi:\mathcal{Y}\to\mathcal{X}$ be deterministic utility and post-processing maps respectively. If $\phi$ preserves the utility $u$ exactly, then for every pair of routes $\{D_1,D_2\}$, we have,
	\[
	\mathrm{TV}\Big( P_1(\dots,\phi),P_2(\dots,\phi) \Big) \geq \mathrm{TV}\Big( P_1(\dots,u),P_2(\dots,u) \Big) \quad ,
	\]
	and also,
	\[
	D_\mathrm{Ch}\Big( P_1(\dots,\phi),P_2(\dots,\phi) \Big) \geq D_\mathrm{Ch}\Big( P_1(\dots,u),P_2(\dots,u) \Big) \quad ,
	\]
	Hence, $\phi$ may only remove route-specific information which is extraneous to the promised utility from $u$; it cannot remove route-information already encoded in the utility output. 
\end{theorem}
\begin{proof}
    By the exact utility preservation of $\phi$ (\cref{def:exact_utility_preservation}), there exists a decoder $\psi:\mathcal{X}\to\mathcal{U}$ such that $u=\psi\circ\phi$. Hence, for any route $D_i$, we can write
	\[
	P_i(\dots,u) = u_\# Q_i = (\psi\circ\phi)_\# Q_i = \psi_\#(\phi_\# Q_i) = \psi_\# P_i(\dots,\phi) \quad ,
	\]
	which is to say that the utility law is itself a deterministic post-processing of the induced transcript law. Therefore, by the same argument as in \cref{cor:post_processing_doesnt_increase_distinguishability}, deterministic post-processing cannot increase distinguishability, and we are left with,
	\[
	\mathrm{TV}\Big( P_1(\dots,\phi),P_2(\dots,\phi) \Big) \geq \mathrm{TV}\Big( \psi_\# P_1(\dots,\phi),\psi_\# P_2(\dots,\phi) \Big) = \mathrm{TV}\Big( P_1(\dots,u),P_2(\dots,u) \Big) \quad .
	\]
	The second inequality can be obtained by a similar argument with Chernoff information replacing total variation, following the same idea from \cref{cor:post_processing_doesnt_increase_distinguishability}.
\end{proof}

The above \cref{thm:utility_preserving_no_free_lunch} tells us that there is no choice of utility-preserving post-processing map that can reduce the total variation distance between the induced distributions of any two routes beyond that of the utility laws without degradation. Equivalently, by \cref{thm:backend_identifiability_reduces_to_hypothesis_testing}, we can say that distinguishability cannot be reduced without also losing utility via post-processing.

It is conceivable that the user will agree to a small loss of utility if it makes the provider happier about his ability to anonymise routes. It is nice to make people happy. To that end, consider approximate versions to the above ideas for utility preservation. \cref{def:approximate_preservation} is perhaps the simplest and most straightforward formulation.

\begin{definition}[$\alpha$-Approximate Utility Preservation] \label{def:approximate_preservation}
	Let $\phi:\mathcal{Y}\to\mathcal{X}$ be any deterministic post-processing map, and let $u:\mathcal{Y}\to\mathcal{U}$ be any utility map. For $\alpha>0$, we say that $\phi$ \textit{preserves the utility $u$ to within accuracy $\alpha$} if there exists a (measurable) decoder $\psi:\mathcal{X}\to\mathcal{U}$ such that, for any route $D_i$, we have $\mathrm{TV}( P_i(\dots,\psi\circ\phi),P_i(\dots,u)) \leq \alpha$.
\end{definition}

\begin{theorem}[Approximate Utility-Preserving No-Free-Lunch] \label{thm:approx_utility_preserving_no_free_lunch}
	Let $u:\mathcal{Y}\to\mathcal{U}$ and $\phi:\mathcal{Y}\to\mathcal{X}$ be deterministic utility and post-processing maps respectively. If $\phi$ preserves the utility $u$ to within accuracy $\alpha>0$, then for every pair of routes $\{D_1,D_2\}$, we have,
	\[
	\mathrm{TV}\Big( P_1(\dots,\phi),P_2(\dots,\phi) \Big) \geq \mathrm{TV}\Big( P_1(\dots,u),P_2(\dots,u) \Big) - 2\alpha \quad .
	\]
\end{theorem}
\begin{proof}
	By the approximate utility preservation of $\phi$, there exists a decoder $\psi:\mathcal{X}\to\mathcal{U}$ such that, for any route $D_i$, we have $\mathrm{TV}( P_i(\dots,\psi\circ\phi),P_i(\dots,u)) \leq \alpha$. Hence, by a triangle inequality,
	\[
	\begin{split}
		\mathrm{TV}\Big( P_1(\dots,u), P_2(\dots,u) \Big)
		&\leq \mathrm{TV}\Big( P_1(\dots,\psi\circ\phi), P_2(\dots,\psi\circ\phi) \Big) \\
		&\phantom{\leq\ } + \mathrm{TV}\Big( P_1(\dots,\psi\circ\phi), P_1(\dots,u) \Big) \\
		&\phantom{\leq\ } + \mathrm{TV}\Big( P_2(\dots,\psi\circ\phi), P_2(\dots,u) \Big) \\
		&\leq \mathrm{TV}\Big( P_1(\dots,\psi\circ\phi), P_2(\dots,\psi\circ\phi) \Big) + 2\alpha \\ 
        &\leq \mathrm{TV}\Big( P_1(\dots,\phi),P_2(\dots,\phi) \Big) + 2\alpha \quad ,
	\end{split}
	\]
	and a simple rearrangement completes the proof.
\end{proof}

In the exact setting, \cref{thm:utility_preserving_no_free_lunch} makes it clear that the utility map itself is the most anonymising post-processing that the provider can apply---any further post-processing that preserves this utility exactly will necessarily increase distinguishability, as per \cref{cor:post_processing_doesnt_increase_distinguishability}. So, we might like to argue that the approximate setting in \cref{thm:approx_utility_preserving_no_free_lunch} is more interesting, and introduces the possibility for an accuracy parameter $\alpha$ of an acceptably small scale which dominates the term $\mathrm{TV}(P_1(\dots,u),P_2(\dots,u))$. Indeed, indefinite anonymity can be obtained with any $\alpha \geq \frac{1}{2}\mathrm{TV}(P_1(\dots,u),P_2(\dots,u))$.

Note that an approximately utility-preserving map defined as in \cref{def:approximate_preservation} could be poorly considered for particular utilities. Take, for example, a binary decision rule to be the utility, and suppose that the number of positive outcomes for the rule is exponentially small in the number of possible inputs. Then a constant post-processing which takes every input to a negative output has only an exponentially small loss of accuracy under the approximate preservation of \cref{def:approximate_preservation}. Clearly though, the meaningful utility---defining a small set of acceptable inputs---is destroyed here. To address this concern, it may be appropriate to define a slightly different notion of preservation for decision utilities. For example:

\begin{definition}[$\alpha$-Approximate Decision-Utility Preservation]
	Let $v=\psi\circ\phi$ be the decision decoded from a transcript $X=\phi(\cdots)$. Then we can call $v$ close to a utility $u$ if,
	\[
	Q_i(v=1\ |\ u=1) \geq 1-\alpha \quad ,
	\]
	for all $D_i$. Hence, among positive decision instances, the post-processed decision retains the decision with arbitrarily high probability.
\end{definition}

With such a definition for preservation, it may be possible to show that rare decision utilities are intrinsically anonymous. That is, if the probability $Q_i(u=1)$ for a positive outcome is exponentially small, then any approximately-preserving post-processing map can improve the anonymity exponentially in $m$. Extended over multiple rounds, this then looks like a polynomial increase in in the number of rounds we can play until anonymity is lost. For now, this is only conjecture; we encourage interesting future work in this direction.

\subsection{Noise-Channel Bounds and Relation to Noise Characterisation} \label{appendix:noise_channel_bounds}

Now we explore the relationship between successful identification in the uniform-binary setting and the underlying effective channel noise. Our basic operational quantity for this is the single-shot, raw distinguishing bias, as give in the following \cref{prop:distinguishing_bias_is_an_average_case_notion}.

\begin{prop}[Distinguishing Bias is an Average-Case Notion] \label{prop:distinguishing_bias_is_an_average_case_notion}
    The single-shot ($m=1$) distinguishing bias on raw observations ($\phi=\mathrm{id}$) over a circuit ensemble $\mu\in\Delta(\mathcal{C})$ is given by,
	\[
	\beta_{D_1,D_2}(\mu,1,\mathrm{id}) = \mathbb{E}_{C\sim\mu}\Big[\mathrm{TV}(p_{1,C}, p_{2,C})\Big] \quad .
	\]
\end{prop}
\begin{proof}
    With only a single shot, the raw observation is a pair $(C,Y)$, with a circuit $C\sim\mu$ and conditionally-sampled outcome $Y\sim p_{i,C}$. Hence,
	\[
	Q_i(dC,y)=\mu(dC)\ p_{i,C}(y) \quad .
	\]
	Then, since $\phi=\mathrm{id}$, we have by \cref{thm:backend_identifiability_reduces_to_hypothesis_testing} and \cref{cor:post_processing_doesnt_increase_distinguishability},
	\[
	\beta_{D_1,D_2}(\mu,1,\mathrm{id}) = \mathrm{TV}(Q_1,Q_2) \quad ,
	\]
	and expanding the total variation over $\mathcal{Y}=\mathcal{C}\times \Omega_n$ completes the proof;
	\[
    \mathrm{TV}(Q_1,Q_2)
    = \frac{1}{2}\int_C\sum_{y\in\Omega_n}\Big|p_{1,C}-p_{2,C}\Big|\ \mu(dC) 
    = \mathbb{E}_{C\sim\mu}\Big[\mathrm{TV}(p_{1,C},p_{2,C})\Big] \quad .
	\]
\end{proof}

This motivates a very practical distance measure. We can see that what matters is \textit{not} any worst-case separation of the full channels over all possible inputs and ancillas, but actually the average separation on the states actually induced by the probing protocol. Hence, we introduce a workload-probed channel (pseudo-)distance accordingly, to give us a very natural distance measure for our setting, tailored to the workload ensemble $\mu$ used in any experiment.

\begin{definition}[Workload-Probed Channel Pseudo-Distance] \label{def:mu_probed_channel_distnace}
	Let $\mu\in\Delta(\mathcal{C})$ be any ensemble of circuits. Then define the \textit{workload-probed channel pseudo-distance} by,
	\[
	\delta_\mu(\mathcal{N}_i, \mathcal{N}_j) := \mathbb{E}_{C\sim\mu}\ \frac{1}{2}\Big\|\Big( \mathcal{N}_{i,C}-\mathcal{N}_{j,C}\Big)(\rho_C) \Big\|_1 \quad ,
	\]
	where each $\mathcal{N}_i=\{\mathcal{N}_{i,C}\}_{C\in\mathcal{C}}$ is a family of circuit-dependent noise channels. 
\end{definition}

This is a \textit{pseudo}-distance since it does not generally satisfy an identity of indiscernibles, though note that is is non-negative, symmetric, and satisfies a triangle inequality. Experimentally, it can also be seen to be a more useful measure; a poorly designed ensemble (from an adversarial user's perspective) may have a very small workload-probed channel pseudo-distance and hence fail to distinguish backends accurately, but there may exist other workloads that have a much easier time of it. Distinguishability of backends is thus relative to the workload, in the most practical sense.

\begin{lemma} \label{lemma:tv_bounded_by_trace_of_channels}
	For any two induced noisy distributions $p_{i,C}=\mathcal{M}(\mathcal{N}_{i,C}(\rho_C))$ (resp. $p_{j,C}$),
	\[
	\mathrm{TV}(p_{i,C},p_{j,C}) \leq \frac{1}{2}\Big\| \mathcal{N}_{1,C}(\rho_C)-\mathcal{N}_{2,C}(\rho_C) \Big\| \quad .
	\]
\end{lemma}
\begin{proof}
	Immediately follows by the contractivity of trace distance under measurement $\mathcal{M}$.
\end{proof}

\begin{prop}[Workload-Probed Distances are Tighter than Worst-Case Distances] \label{prop:workload_probed_distance_tighter_than_diamon_norm}
    For any fixed ensemble $\mu\in\Delta(\mathcal{C})$ ,
	\[
	\beta_{D_1,D_2}(\mu,1,\phi) \leq \delta_\mu(\mathcal{N}_1,\mathcal{N}_2) \leq \frac{1}{2}\sup_{C\in\mathcal{C}}\Big\| \mathcal{N}_{1,C}-\mathcal{N}_{2,C} \Big\|_\diamond \quad .
	\]
\end{prop}
\begin{proof}
    These bounds can easily be seen via elementary facts and above results;
	\begin{align}
		\beta_{D_1,D_2} &\overset{\ref{prop:distinguishing_bias_is_an_average_case_notion}}{=} \mathbb{E}_{C\sim\mu}\Big[ \mathrm{TV}(p_{1,C},p_{2,C}) \Big] \phantom{\frac{1}{2}} \nonumber \\
		&\leq \mathbb{E}_{C\sim\mu}\ \frac{1}{2}\Big\| \mathcal{N}_{1,C}(\rho_C)-\mathcal{N}_{2,C}(\rho_C) \Big\|_1 \label{eq:def_of_tv} \\
		&\overset{\ref{def:mu_probed_channel_distnace}}{=} \delta_\mu(\mathcal{N}_i, \mathcal{N}_j) \phantom{\frac{1}{2}} \nonumber \\
		&\leq \mathbb{E}_{C\sim\mu}\ \frac{1}{2}\Big\| \mathcal{N}_{1,C}-\mathcal{N}_{2,C} \Big\|_\diamond \label{eq:diamond_is_worse_than_tv} \\
		&\leq \frac{1}{2}\ \sup_{C\in\mathcal{C}}\ \Big\| \mathcal{N}_{1,C}-\mathcal{N}_{2,C} \Big\|_\diamond \quad . \label{eq:sup_is_worse_than_expect}
	\end{align}
	In the above, Eq. (\ref{eq:def_of_tv}) follows by \cref{lemma:tv_bounded_by_trace_of_channels}; then Eq. (\ref{eq:diamond_is_worse_than_tv}) follows by definition of the diamond norm as a supremum over input states and thus a bound on the trace distance of any $\rho_C$; and finally Eq. (\ref{eq:sup_is_worse_than_expect}) is a simple use of the fact that the expectation over a set is trivially bounded by the supremum over that set.
\end{proof}

Establishing the above relationship in \cref{prop:workload_probed_distance_tighter_than_diamon_norm} allows us to immediately suggest some practically useful sufficient conditions to achieve a desired degree of routing anonymity.

\begin{corollary}
	Under any deterministic post-processing map, $\mathrm{Adv}\leq m \cdot\max_{i\neq j}\delta_\mu(\mathcal{N}_i,\mathcal{N}_j)$, therefore a sufficient condition for $\varepsilon$-anonymity with respect to any map class is,
	\[
	\max_{i\neq j}\ \delta_\mu(\mathcal{N}_i,\mathcal{N}_j) \leq \frac{\varepsilon}{m} \quad .
	\]
\end{corollary}

\begin{corollary}
	Suppose that there exist channels $\Lambda_i$ such that $\mathcal{N}_{i,C}=\Lambda_i$ for every backend $D_i$ and circuit $C$. Then, for any post-processing map, $\mathrm{Adv} \leq \frac{m}{2}\max_{i\neq j}\big\| \Lambda_i-\Lambda_j \big\|_\diamond $ hence a sufficient condition for $\varepsilon$-anonymity is,
	\[
	\max_{i\neq j}\big\|\Lambda_i-\Lambda_j\big\|_\diamond  \leq \frac{2\varepsilon}{m} \quad .
	\]
\end{corollary}

\section{Intermediate-Depth Principle in the PTM Model} \label{appendix:depth_window_principle}
Let $\mathcal{P}_n$ denote the $n$-qubit Pauli basis. Every $n$-qubit density operator $\rho$ can be expanded like,
\[
\rho = \frac{1}{2^n}\left( I + \hspace{-7px}\sum_{P\in\mathcal{P}_n\backslash\{I\}}\hspace{-7px} \mathrm{Tr}(P\rho) P \right) \quad ,
\]
and we call the vector $\vec{r}(\rho):=(\mathrm{Tr}(P\rho))_{P\neq I}\in\mathbb{R}^{4^n-1}$ the \textit{traceless Pauli vector} of $\rho$, recording its non-identity Pauli expectation values. 

A quantum channel $\mathcal{N}$ can be represented in the same basis by its \textit{Pauli Transfer Matrix (PTM)}, which, for a trace-preserving channel, takes the block form,
\[
\begin{pmatrix}
	1 & 0 \\
	t_\mathcal{N} & N_\mathcal{N}
\end{pmatrix} \quad ,
\]
where $N_\mathcal{N}$ is the block acting on the traceless Pauli coordinates, and $t_\mathcal{N}$ is an affine, non-unital drift. Therefore, the traceless part transforms as $\vec{r}(\rho)\mapsto t_\mathcal{N}+N_\mathcal{N}\vec{r}(\rho)$, and if the channel is unital, then $t_\mathcal{N}=0$ and evolution cleans up to $\vec{r}(\rho)\mapsto N_\mathcal{N}\vec{r}(\rho)$.

PTMs have been widely used across the literature because they give a concrete linear representation of quantum processes in the Pauli basis. We see them in gate-set tomography and quantum process characterisation, where they provide a useful diagnostic representation of gate errors \cite{greenbaum2015introduction}; in learning problems, wherein it can often be easier to learn entries and expectation values associated with the PTM of unknown quantum processes \cite{caro2022learning}; and spectral quantum tomography, where eigenvalues of the PTM of a noisy gate can act as useful gate diagnostics \cite{helsen2019spectral}. These applications motivate the use of PTMs to yield a simple coordinate system wherein differences between noisy processes can be represented, learned, and compared.

\subsection{Output Distributions from Traceless Pauli Signals}

Let $\vec{r}_C:=\vec{r}(\rho_C)$ denote the traceless Pauli vector of the ideal output state $\rho_C$ for a given depth-$d$ circuit $C$. Under the execution of a backend $D_i$, the computational-basis output distribution $p_{i,C}$ can be modelled like,
\begin{equation} \label{eq:output_distribution_from_pauli_signal}
p_{i,C} = q + MN_i^d\vec{r}_C \in \Delta(\Omega_n) \quad ,
\end{equation} 
where $q\in\Delta(\Omega_n)$ is a common asymptotic fixed-point distribution, $M:\mathbb{R}^{4^n-1}\to\mathbb{R}^{2^n}$ is a fixed linear map transforming traceless Pauli coordinates into computational-basis probability deviations, and $N_i$ is the effective traceless PTM block describing the noisy action of $D_i$ per circuit layer. We can thus view $N_i^d\vec{r}_C$ as the residual circuit-dependent traceless Pauli signal after $d$ noisy layers.

There are several simplifying assumptions present in Eq. (\ref{eq:output_distribution_from_pauli_signal}): (1) noisy evolution is treated as Markovian and depth-homogeneous at the level of an effective layer noise map, applying the same traceless PTM block $N_i$ after each layer; (2) $N_i$ is assumed diagonal, or effectively diagonal after Pauli twirling; (3) affine non-unital effects are either negligible or common among compared backends and thus absorbed into $q$; (4) the measurement map is assumed to be backend-independent so that backend-specific readout effects are either negligible or common and thus absorbed into $q$; and (5) we assume either that the effective layer noise is unital, or that all compared backends share a common fixed point and we work in coordinates centred at this fixed point (in the latter case, $q$ denotes the corresponding computational-basis fixed-point distribution and $\vec{r}_C$ should be interpreted as the centred ideal signal).

Sufficiently mixing noisy dynamics are often modelled by traceless PTM blocks whose relevant eigenvalues lie strictly inside the unit circle; e.g. simple depolarising channels act as $\vec{r}(\rho)\to\lambda\vec{r}(\rho)$ with $\lambda\in(0,1)$ and thus have corresponding traceless PTM block $\lambda I$, and amplitude-damping noise has subunit eigenvalues like $\sqrt{1-p}$ and $1-p$ with damping rate $p>0$ \cite{helsen2019spectral}. \citet{raginsky2002strictly} define these kinds of \textit{strictly contractive} channel as those which reduce trace distance by a factor less than $1$, interpreting them as making quantum states less distinguishable.  

This motivates a key dominant-mixing assumption for the intermediate-depth principle: decompose the traceless PTM block $N_i$ of a backend $D_i$ like,
\begin{equation} \label{eq:decomposed_ptm_block}
N_i = \lambda I + E_i \quad ,
\end{equation}
where $\lambda I$ is a common depolarising-like contraction with $\lambda\in(0,1)$, and $E_i$ is a backend-specific perturbation with $\|E_i\|_{1\to 1} \leq \varepsilon$ such that $\lambda+\varepsilon<1$. That is, we will assume that the backend-specific perturbation $E_i$ is sufficiently small so as not to destroy the common contraction on the traceless PTM subspace. 

This is a standard form of modelling assumption for analysing noisy quantum dynamics. Strictly contractive quantum channels have long been used as mathematical models of non-ideal quantum evolution, especially for settings wherein finite experimental precision and repeated noise make perfectly reversible dynamics unrealistic \cite{raginsky2002strictly}. The broader theory of quantum contraction coefficients similarly studies how noisy channels contract operational distances and divergences between states \cite{hiai2016contraction}. Moreover, randomised benchmarking and randomised compiling provide both theoretical and experimental support for replacing complicated coherent hardware errors by effective stochastic Pauli or depolarising-like noise models: randomised compiling was introduced to tailor coherent errors into stochastic Pauli errors \cite{wallman2016noise}; and Pauli-frame randomisation has been experimentally demonstrated on superconducting hardware to suppress non-Markian and non-Pauli signatures while leaving the computation unchanged \cite{ware2021experimental}. Recent noisy-circuit analyses also use strict contractivity as a depth-dependent mixing hypothesis, with amplitude damping and suitable Pauli channels appearing as canonical examples of strictly contractive noise maps \cite{schumann2024emergence}.

The following \cref{lemma:dominant_ptm_contraction} formally presents the corresponding strict contractivity statement in our PTM model's $\ell_1$-coordinate norm as a consequence of our assumptions about the norm gap $\|E_i\|_{1\to 1}<1-\lambda$ in Eq. (\ref{eq:decomposed_ptm_block}). Note that this is a coordinate-level contraction assumption in the traceless Pauli representation. It should not be confused with trace-norm contractivity unless additional norm-comparison constants are introduced.

\begin{lemma} \label{lemma:dominant_ptm_contraction}
	Let $N_i=\lambda I + E_i$ be a linear map on a traceless PTM coordinate space equipped with an $\ell_1$-norm. If $\lambda\in(0,1)$ and $\|E_i\|_{1\to 1}\leq\varepsilon$ such that $\lambda+\varepsilon<1$, then $\|N_i\|_{1\to 1}\leq\lambda+\varepsilon<1$.
	
	Consequently, for any traceless PTM vector $\vec{r}$ and $d\geq 0$, $\|N_i^d \vec{r}\ \|_1 \leq (\lambda+\varepsilon)^d \|\vec{r}\ \|_1$. Hence, we have that $N_i^d \vec{r}\to 0$ in $\ell_1$-norm as $d\to\infty$.
\end{lemma}
\begin{proof}
	For any $\vec{r}$, write $N_i \vec{r}=\lambda\vec{r}+E_i\vec{r}$. By a triangle inequality and the definition of the induced operator norm, we then have,
	\[
	\|N_i \vec{r}\ \|_1 \leq \lambda \|\vec{r}\ \|_1 + \|E_i\vec{r}\ \|_1 \leq (\lambda + \|E_i\|_{1\to 1})\|\vec{r}\ \|_1 \leq (\lambda + \varepsilon)\|\vec{r}\ \|_1 \quad .
	\]
	Taking the supremum over $\vec{r}\neq 0$ yields the first result; $\|N_i\|_{1\to 1}\leq\lambda+\varepsilon<1$. Then, by submultiplicativity of induced norms, $\|N_i^d\|_{1\to 1} \leq \|N_i\|_{1\to 1}^d \leq (\lambda+\varepsilon)^d$. Hence, the second result;
	\[
	\|N_i^d\vec{r}\ \|_{1} \leq \|N_i^d\|_{1\to 1}\|\vec{r}\ \|_1 \leq \|N_i\|_{1\to 1}^d \|\vec{r}\ \|_1 \leq (\lambda+\varepsilon)^d \|\vec{r}\ \|_1 \quad .
	\]
	Finally, since $0<\lambda+\varepsilon<1$, the factor $(\lambda+\varepsilon)^d\to 0$ as $d\to \infty$, hence $\|N_i^d \vec{r}\ \|_1\to 0$.
\end{proof}

\subsection{Intermediate-Depth Identifiability Principle}

We now show that the contraction statement of \cref{lemma:dominant_ptm_contraction}, following our standard strict contractivity assumption, leads to a useful practical intuition about the relationship between depth and backend identifiability. The quantity of interest is the \textit{expected total variation distance} between output distributions produced by two backends over an ensemble $\mu_d$ of depth-$d$ circuits, as the average single-shot distinguishing behaviour constrained to a particular depth. 

The following theorem formalises the idea that, under the dominant-mixing PTM model, the backend-specific contribution to distinguishing bias is necessarily suppressed at large depths, leading to a degradation of backend identifiability.

\begin{theorem}[Identifiability Degrades at Large Depths] \label{thm:identifiability_degrades_at_large_depths}
	Fix any two backends $D_i$ and $D_j$, and let $\mu_d$ be an ensemble of depth-$d$ circuits. Suppose that, for every $C\sim\mu_d$, the computational-basis output distributions for $D_i$ (similarly for $D_j$) satisfy,
	\[
	p_{i,C} = q + MN_i^d\vec{r}_C \quad ,
	\]
	where $q$ is common to both backends, $M$ is a fixed linear map, and we assume that,
	\[
	N_i = \lambda I + E_i\ ,
	\qquad
	\text{with}\quad
	\|E_i\|_{1\to 1}\leq\varepsilon
	\quad\text{s.t.}\quad
	\lambda+\varepsilon<1 \quad ,
	\]
	and that the ideal traceless Pauli vectors $\vec{r}_C$ are bounded on average; i.e. $\mathbb E_{C\sim\mu_d}\|\vec{r}_C\|_1\leq R$.
	
	Denoting by $\beta_{D_i,D_j}(d):=\mathbb E_{C\sim\mu_d}\mathrm{TV}(p_{i,C},p_{j,C})$ the expected distinguishing bias between $D_i$ and $D_j$ over depth-$d$ circuit ensemble $\mu_d$, we have, for every $d\geq 1$, 
	\[
	\beta_{D_i,D_j}(d) \leq \frac{1}{2} \|M\|_{1\to1} R\, d(\lambda+\varepsilon)^{d-1} \|N_i-N_j\|_{1\to1} = \mathcal{O}(d(\lambda+\varepsilon)^{d-1}) \quad .
	\]
	In particular, $\beta_{D_i,D_j}(d)\to0$ as $d\to\infty$.
\end{theorem}
\begin{proof}
	Fix a circuit $C\sim\mu_d$. Under our PTM model, after $d$ noisy layers,
	\[
	p_{i,C}-p_{j,C} = q+M N_i^d\vec{r}_C - q-M N_j^d\vec{r}_C, = M(N_i^d-N_j^d)\vec{r}_C \quad ,
	\]
	hence the distinguishing bias for $C$ can be written,
	\[
	\mathrm{TV}(p_{i,C},p_{j,C}) = \frac{1}{2} \|M(N_i^d-N_j^d)\vec{r}_C\|_1 \leq \frac{1}{2}\|M\|_{1\to 1} \|N_i^d-N_j^d\|_{1\to1} \|\vec{r}_C\|_1 \quad .
	\]
	
	By the telescoping identity,
	\[
	N_i^d-N_j^d = \sum_{\ell=0}^{d-1} N_i^\ell (N_i-N_j) N_j^{d-1-\ell} \quad ,
	\]
	thus, by then taking induced norms and using submultiplicity, we get,
	\[
	\|N_i^d-N_j^d\|_{1\to1} \leq \sum_{\ell=0}^{d-1} \|N_i^\ell\|_{1\to1} \|N_i-N_j\|_{1\to1} \|N_j^{d-1-\ell}\|_{1\to1} \quad .
	\]
	
	Now, we can utilise our assumptions; by \cref{lemma:dominant_ptm_contraction}, $\|N_i^k\|_{1\to1}\leq(\lambda+\varepsilon)^k$ for any $k\geq 0$ (similarly for $N_j$). Hence, each summand in the above satisfies,
	\[
	\begin{split}
	\|N_i^\ell\|_{1\to1} \|N_i-N_j\|_{1\to1} \|N_j^{d-1-\ell}\|_{1\to1} 
	&\leq (\lambda+\varepsilon)^\ell \|N_i-N_j\|_{1\to1} (\lambda+\varepsilon)^{d-1-\ell} \\
	&= (\lambda+\varepsilon)^{d-1} \|N_i-N_j\|_{1\to1} \quad ,
	\end{split}
	\]
	yielding,
	\[
	\|N_i^d-N_j^d\|_{1\to1} \leq \sum_{\ell=0}^{d-1} (\lambda+\varepsilon)^{d-1} \|N_i-N_j\|_{1\to1} \leq d (\lambda+\varepsilon)^{d-1} \|N_i-N_j\|_{1\to1} \quad .
	\]
	
	Substituting this into the distinguishing bias at depth $d$ and taking its expectation over $C\sim\mu_d$, 
	\[
	\beta_{D_i,D_j}(d) := \mathbb E_{C\sim\mu_d} \mathrm{TV}(p_{i,C},p_{j,C}) \leq \frac{1}{2}\|M\|_{1\to1} d(\lambda+\varepsilon)^{d-1} \|N_i-N_j\|_{1\to1} \mathbb E_{C\sim\mu_d}\|\vec r_C\|_1 \quad ,
	\]
	then using the assumed bound $\mathbb{E}_{C\sim\mu_d}\|\vec r_C\|_1\leq R$ gives the targeted expression. And finally, since $0<\lambda+\varepsilon<1$, the term $d(\lambda+\varepsilon)^{d-1}\to0$ as $d\to\infty$, hence $\beta_{D_i,D_j}(d)\to0$. 
\end{proof}

This establishes that distinguishability is not a high-depth phenomenon. Next, we formalise the shallow-depth side of the picture by showing that, over a perturbative range, the leading distinguishable signal builds up with depth and thus is not a low-depth phenomenon either (relatively speaking). We do this in the following theorem by presenting a lower bound for the same expected distinguishability object that holds over a perturbative range of shallow depths. Some work is offloaded into the preceding lemma.

\begin{lemma} \label{lemma:identifiability_grows_at_small_depths_residual_bound}
	Let $N_i=\lambda I+E_i$ with $\|E_i\|_{1\to1}\leq \varepsilon$ (similarly for $N_j$). Then,
	\[
	\|\mathcal{R}_d\|_{1\to1}:=
	\left\|\ \sum_{k=2}^d \binom{d}{k} \lambda^{d-k}(E_i^k-E_j^k)\ \right\|_{1\to 1} \hspace{-5px} \leq\  d\lambda^{d-1} \left[ \left(1+\frac{\varepsilon}{\lambda}\right)^{d-1}-1 \right] \|N_i-N_j\|_{1\to1} \quad .
	\]
\end{lemma}
\begin{proof}
	For each $k\geq2$, the telescoping identity gives that,
	\[
	E_i^k-E_j^k = \sum_{\ell=0}^{k-1} E_i^\ell(E_i-E_j)E_j^{k-1-\ell} \quad ,
	\]
	then taking induced $\ell_1\to\ell_1$ norms and using submultiplicativity,
	\[
	\|E_i^k-E_j^k\|_{1\to1} 
	\leq \sum_{\ell=0}^{k-1} \|E_i\|_{1\to1}^{\ell} \|E_i-E_j\|_{1\to1} \|E_j\|_{1\to1}^{k-1-\ell} 
	\leq k \varepsilon^{k-1} \|E_i-E_j\|_{1\to1} \quad ,
	\]
	thus we can write,
	\[
	\|\mathcal{R}_d\|_{1\to1} \leq \sum_{k=2}^{d}\binom{d}{k}\lambda^{d-k}\|E_i^k-E_j^k\|_{1\to1} \leq \sum_{k=2}^{d} \binom{d}{k}\lambda^{d-k} k \varepsilon^{k-1} \|E_i-E_j\|_{1\to1} \quad .
	\]
	Now we can use $\sum_{k=1}^{d}\binom{d}{k}k\lambda^{d-k}\varepsilon^{k-1}=d(\lambda+\varepsilon)^{d-1}$ and remove the $k=1$ summand to yield,
	\[
	\|\mathcal R_d\|_{1\to1} \leq d((\lambda+\varepsilon)^{d-1}-\lambda^{d-1}) \|N_i-N_j\|_{1\to1} \quad ,
	\]
	and finally factor out $d\lambda^{d-1}$ to leave the stated bound.
\end{proof}

\begin{theorem}[Identifiability Grows at Small Depths] \label{thm:identifiability_grows_at_small_depths}
	Assume the same hypotheses as in \cref{thm:identifiability_degrades_at_large_depths}, and further suppose that for some perturbative range of depths $1\leq d\leq d_0$ satisfying, 
	\[
	\|M\|_{1\to1} \left[\left(1+\frac{\varepsilon}{\lambda}\right)^{d-1}-1\right] R \leq \frac{\kappa}{2} \quad ,
	\]	
	the ensemble $\mu_d$ of depth-$d$ circuits is $\kappa$-nondegenerate, in the sense that,
	\[
	\mathbb{E}_{C\sim\mu_d} \| M (N_i-N_j) \vec{r}_C \|_1 \geq \kappa \|N_i-N_j\|_{1\to 1} \quad .
	\]
	Then, for every $1\leq d\leq d_0$,
	\[
	\beta_{D_i,D_j}(d) \geq \frac{\kappa}{4} d\lambda^{d-1} \|N_i-N_j\|_{1\to 1} = \Omega(d\lambda^{d-1}) \quad .
	\]
\end{theorem}
\begin{proof}
	Fix $1\leq d\leq d_0$ and a circuit $C\sim\mu_d$. Under our PTM model, after $d$ noisy layers,
	\[
	p_{i,C}-p_{j,C} = q+M N_i^d\vec{r}_C - q-M N_j^d\vec{r}_C = M(N_i^d-N_j^d)\vec{r}_C \quad .
	\]
	Since $I$ commutes with $E_i$, a binomial expansion of $N_i=\lambda I+E_i$ gives,
	\[
	N_i^d = \lambda^d I + d\lambda^{d-1}E_i + \sum_{k=2}^d\binom{d}{k}\lambda^{d-k}E_i^k \quad ,
	\]
	and similarly for $N_j$. Taking the difference of the two expansions gives,
	\[
	N_i^d - N_j^d = d\lambda^{d-1}(E_i-E_j) + \underbrace{\ \sum_{k=2}^d\binom{d}{k}\lambda^{d-k}(E_i^k-E_j^k)}_{\text{Higher-order remainder $=:\mathcal{R}_d$}\ } \quad .
	\]
	Note that $E_i-E_j=N_i-N_j$. Now write the distinguishing bias for $C$, 
	\[
	\begin{split}
		\frac{1}{2}\|p_{i,C}-p_{j,C}\|_1
		&= \frac{1}{2}\|M(N_i^d-N_j^d)\vec{r}_C\|_1 \\
		&\geq \frac{1}{2}d\lambda^{d-1}\|M(N_i-N_j)\vec{r}_C\|_1
		- \frac{1}{2}\|M\mathcal{R}_d\vec{r}_C\|_1 \\
		&\geq \frac{1}{2}d\lambda^{d-1}\|M(N_i-N_j)\vec{r}_C\|_1
		- \frac{1}{2}\|M\|_{1\to1}\|\mathcal{R}_d\|_{1\to1}\|\vec{r}_C\|_1
	\end{split}
	\]
	using a reverse triangle inequality. Then, taking the expectation over $C\sim\mu_d$,
	\[
	\begin{split}
		\beta_{D_i,D_j}(d) 
		&\geq \frac{1}{2}d\lambda^{d-1}\ \mathbb{E}_{C\sim\mu_d} \|M(N_i-N_j)\vec{r}_C\|_1
		- \frac{1}{2}\|M\|_{1\to1}\|\mathcal{R}_d\|_{1\to1}\ \mathbb{E}_{C\sim\mu_d} \|\vec{r}_C\|_1 \\
		&\geq \frac{\kappa}{2} d\lambda^{d-1} \|N_i-N_j\|_{1\to1}
		- \frac{1}{2}\|M\|_{1\to1}\|\mathcal{R}_d\|_{1\to1} R \\
		&\overset{\ref{lemma:identifiability_grows_at_small_depths_residual_bound}}{\geq} \frac{\kappa}{2} d\lambda^{d-1} \|N_i-N_j\|_{1\to1}
		- \frac{1}{2}d\lambda^{d-1} \|M\|_{1\to1} \left[ \left(1+\frac{\varepsilon}{\lambda}\right)^{d-1}-1 \right] R \|N_i-N_j\|_{1\to1} \\
		&\geq \frac{\kappa}{2} d\lambda^{d-1} \|N_i-N_j\|_{1\to1}
		- \frac{\kappa}{4} d\lambda^{d-1} \|N_i-N_j\|_{1\to1} \\
		&= \frac{\kappa}{4} d\lambda^{d-1} \|N_i-N_j\|_{1\to 1} \quad .
	\end{split}
	\]
\end{proof}

Together, \cref{thm:identifiability_degrades_at_large_depths} and \cref{thm:identifiability_grows_at_small_depths} present our formal intuition for the intermediate-depth principle; shallow circuits may not accumulate enough backend-specific noise to separate, while very deep circuits are dominated by common mixing, hence optimal distinguishability lies between these extrema. Recall \cref{fig:intermedaite_depth_phenomenon_illustration} from the main text, which illustrates this mechanism.

\appendixpart{Implementation Details and Experimental Results} 

\section{Probe Circuits} \label{appendix:experiment_circuits}
Our experiments are run remotely on three different quantum processors (`\textit{devices}') via AWS Braket: two superconducting QPUs (Ankaa-3 and Garnet) and one ion-trap QPU (Aria-1). These devices are deployed by Rigetti, IQM, and IonQ respectively. Throughout this section, we refer to comparisons between the superconducting devices as \textit{like-type}, and to comparisons between superconducting and ion-trap devices as \textit{differing-type}.

At the high level, each individual \textit{call} to a device takes the following simple form: (1) a circuit of depth $d$ is defined; (2) the device is remotely prepared with the circuit; (3) the circuit is executed numerous times, specified by the number of \textit{shots}; then (4) a histogram of final measurements is returned, which we read as an outcome probability distribution over discrete bitstrings. 

Our experiments in this work use circuits at a fixed width of $n=5$ qubits, and hence lead to outcome distributions over $2^5=32$ possible bitstrings. A sensible choice for the number of shots is much larger than this. For example, consider that each bitstring has an equal probability of $1/32\approx{0}.03125$ to be measured. Then $2^6=64$ shots is expected to give us only 2 measured outcomes per bitstring and so is extremely vulnerable to random noise greatly skewing our estimated probabilities; e.g. if only a single measurement error is made, then one outcome will falsely be perceived to have triple probability of another. On the other hand, with $2^8=256$ shots, we have 8 measurements per outcome, and so errors in measurement can be tolerated more comfortably. The number of shots thus dictates the \textit{resolution} with which we describe noise apparent in each call, and subsequently influences how precisely we may learn that description. Following preliminary experimentation, we have used our best judgment to select an appropriate number of shots.

In an effort to explore many different questions about device noise, it becomes necessary to run different forms of experiment. We categorise our device calls into two suites of experiments: \textit{depth-varied} and \textit{time-varied}.

\begin{widefigure}
    \centering
    \begin{subfigure}[t]{0.48\linewidth}
        \centering
        \includegraphics[width=\linewidth]{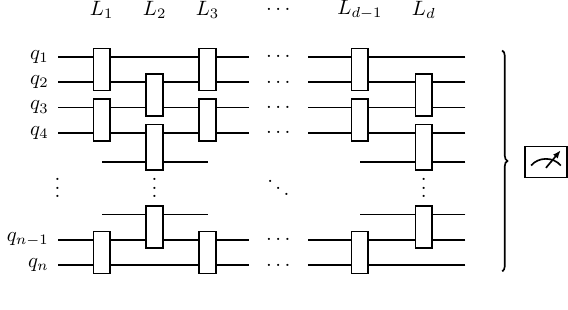}
        \caption{Brickwork architecture (depth-varied).}
        \label{fig:brickwork_circuit}
    \end{subfigure}\hfill
    \begin{subfigure}[t]{0.48\linewidth}
        \centering
        \includegraphics[width=\linewidth]{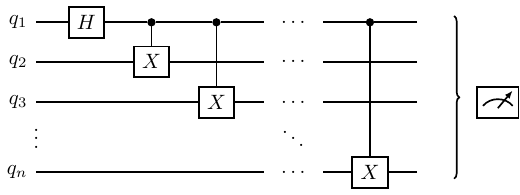}
        \caption{GHZ circuit (time-varied).}
        \label{fig:ghz_circuit}
    \end{subfigure}
    \caption{Probe circuit structures used in our cloud experiments. We have: (a) \textit{brickwork circuit architecture}, wherein each layer $L_i$ comprises (roughly) $n/2$ two-qubit unitaries, which offset-interlock with those in layers $L_{i-1}$ and $L_{i+1}$, randomly sampled from the unitary Haar measure, and we call such a circuit ``depth $d$'' if it has $d$ layers $L_1,\dots,L_d$; and (b) \textit{GHZ circuit} which prepares the state $|\mathrm{GHZ}_n\rangle=\frac{1}{\sqrt{2}}(|0\overset{n}{\dots}0\rangle+|1\overset{n}{\dots}1\rangle)$ over $n$ qubits.}
    \label{fig:circuits}
\end{widefigure}

\subsection{Depth-Varied Experiments}

This suite of experiments explores questions about circuit depth, and whether the relationship between noise and depth depends strongly on the device. 

In each call to a device, we define a depth-$d$ circuit by sampling Haar-random unitaries to be stacked together in an alternating brickwork pattern of $d$ layers (see \cref{fig:brickwork_circuit}). Importantly, note that the circuit seen in each individual call to a device is independently sampled (and almost certainly unique), thus in any given run it cannot easily be inferred how much noise is present---or how it presents---in the outcome distribution until we provide knowledge about the circuit. Designing the experiments in this way allows us the opportunity to ask more nuanced questions about the noise fingerprints both implicitly (without circuit knowledge) and explicitly (accounting for the circuit).

Another motivation for the randomness of our circuits is to remove bias. With sufficiently many experiment calls, and increasingly with deeper circuits, we can expect that our probability distributions approach uniformly random. This opens the possibility to ask whether well-spread distributions, as we increasingly expect from deeper and deeper circuits, dilute the influence of a device's unique character. And perhaps most usefully for the name-sake objective of this suite of experiments, we can roughly expect the same true outcome distribution (i.e. in the noise-free setting) beyond a sufficiently large depth, and thus we can better understand how noise fingerprints change and become harder to effectively describe with increased depth. 

On each of the two superconducting devices (Ankaa-3 and Garnet), we make calls with 500 random depth-$d$ circuits, with $2^9=512$ shots of resolution, at depths $d=5,10,15,20$. On the ion-trap device (Aria-1), we make 350 calls, with the same resolution, at depths $d=5,10$. This choice of minimum depth is sufficient to allow for roughly uniformly random outcome distributions at our chosen resolution, and further allows the possibility for each qubit to interact with every other. 

\subsection{Time-Varied Experiments}

This suite of experiments explores questions about the dynamics of device noise, both in the short(er)-term and long(er)-term. Being frugal, we can---and will---look toward classification tasks with data from these experiments, but they are primarily run with the purpose of understanding how well we can reason about the progression of noise in and between devices.

Every call to a device is made with the same circuit preparing a GHZ state over our $n=5$ qubits (illustrated in \cref{fig:ghz_circuit}); the true probability distribution is known to be an equal weighting between the all-0s bitstring and the all-1s bitstring. These circuits are constrained to a linear depth in $n$, plus measurement, and so (since our $n$ is reasonably small) we should not expect much deviation from the true distribution. For this reason, we judge it necessary to increase our resolution to $2^{10}=1024$ shots to capture the far smaller probabilities leaking into the remaining 30 possible bitstrings. Note that this identical circuit preparation has the contrasting setup to our Haar-random circuits from the depth-varied suite of experiments; we can ask contrasting questions about whether highly concentrated distributions are easier to learn, as well as tie in temporal questions to understand whether the concentration of the distribution spills out in a predictable way, with respect to the device or its physical platform.

Experiments in this suite are grouped into three \textit{batches}; a batch of calls are made in a series of contiguous \textit{time steps} within the same day, with reasonably small and similar delays between successive steps. Calls between batches are separated by days or more, with the interval between batches being significantly larger than any interval between calls within the same batch.

In the first two batches, each superconducting device (Ankaa-3 and Garnet) is called in 100 contemporaneous time steps, which we sort chronologically according to completion time. As a post-processing step, we can take each completion time relative to some fixed time to assign to each call a single numerical quantity. The final batch is similar, and further includes 45 time steps of calls to the ion-trap device (Aria-1) at the same resolution. For clarity, note that each time step is an independent call, with freshly prepared initial states and GHZ circuit.

\section{Data Processing}
Suppose we have run, on some quantum device, an $n$-qubit circuit whose theoretical output distribution should be given by $\mathbf{y}=(y_{1},\dots,y_{N})$ for $N=2^n$, but for which we actually measure an output distribution of $\mathbf{y}'=(y_{1}',\dots,y'_{N})$. There are two things we want to do with these data: in the \textit{implicit} case, we would like to directly apply ML techniques to $\mathbf{y}$; and in the \textit{explicit} case, we would like to find a suitable function $f(\mathbf{y},\mathbf{y}')$ which produces an input feature vector effectively isolating and describing some aspect of noise particular to that device.

The implicit case is specified by the trivial function $f(\mathbf{y},\mathbf{y}')=\mathbf{y}$, so we should ask the more interesting question: \textit{which functions best isolate ``noise'' in the explicit case?} In particular, it would be useful to discover how best we can specify and/or enhance the characteristics of a given device's noise fingerprint that make it unique/distinctive. The nature of $f$ then helps us to decorate any learned characteristics, and representations thereof, with useful theoretical interpretations.

To structure the discussion, we will discern three categories for the form of function $f$, which increase in the strictness of the imposed structure, though do so with an increasing interpretability: \textit{neurally-guided}, \textit{per-outcome}, and \textit{per-qubit}. 

\subsection{Neurally-Guided Features}

Sitting at the lenient end, we can use any arbitrary neural network architecture to transform some combination of $\mathbf{y}$ and $\mathbf{y'}$ (ideally, retaining most information from each) directly into a feature vector without any premonition for its interpretation. In fact, the only `strictness' involved in this approach is to enforce a dimensionality, call it $N'$, for the desired latent representation.

At its simplest, we might like to feed, as input to our ``neural network'' $f$, the concatenation $\mathbf{x}:=(y_{1},\dots,y_{N},y'_{1},\dots,y'_{N})$. Any neural network is suitable; e.g. in this work, we will remain general and use a standard multilayer perceptron (MLP). At the output of $f$ is an $N'$-dim vector representation of $\mathbf{x}$, which we can then pass through a final fully-connected layer $g$ (or a secondary model accepting a latent representation for input) to perform a final classification task.

Then we have some ``full-model'' $g\circ f$ which, once trained to classify between devices given these input $\mathbf{x}$, contains a ``sub-model'' $f$ transforming the concatenation of two probability distributions (true and obtained) into an $N'$-dim feature vector which can be used to distinguish two devices. In this way, the feature vector obtained from $f$ is an optimal description of a noise fingerprint for some device \textit{with respect to other devices}. Moreover, $g$ can be understood to be reading the character of a given fingerprint's representation for backend identification. \cref{fig:latent_representations} in the main text presents these kinds of neurally-guided feature from our later experiments.

A less task-dependent alternative to produce a representation of the noise fingerprint in arbitrary dimension is to use a compressing architecture, such as a variational autoencoder (VAE). These models do not use a secondary process for contextualising how the representations should be learned, but instead work by composing an ``encoder'' with a ``decoder''. The encoder works to reduce the input to some latent space, and the decoder works in reverse to reconstruct the original input from the latent space. Hence, the latent representation we learn is specifically designed to retain as much information as possible in lower dimensions. Such approaches can also be combined with a downstream task (e.g. classifying between devices) by re-purposing the encoder in a new model once trained (i.e. used as $f$ in our above discussion). We leave the use of VAE-type noise fingerprint representation to future work.

As a final note before moving on to other approaches: the interpretation we can draw from $f$ is thus inexplicably tied to which devices are used for training the full-model, as well as the nature of the training data. Especially, the set of circuits $\mathcal{C}$ used to generate the data is of extreme importance, and allows us to enforce an interpretation on our representations. For example, if we choose $\mathcal{C}$ to be a suitably general set of circuits, then we might expect that our latent representations will also be very general in nature---most likely learned to be as distinguishable per-device as possible without bias. On the other hand, we could choose $\mathcal{C}$ to be very specific (e.g. stabiliser circuits, quantum Fourier transforms, low-depth random circuits, etc.) to hone in on some precise aspect of noise \textit{with respect to a specific task/computation}. If we have two devices performing a particular task in some controlled context, then learning how their noise fingerprints differ \textit{with respect to this context} could be extremely useful. The beauty of this neurally-guided approach is that it affords us the flexibility to change the set of circuits $\mathcal{C}$, the downstream task, the latent dimension $N'$, etc.

\subsection{Per-Outcome Features}

Increasing the strictness somewhat, we can instead opt to act directly on the probability distributions, retaining their shape and ordering. The most straightforward things we could want to do are to take residuals. We can then weight them in interesting ways to come up with different interpretations on the data. 

A per-outcome feature $\mathbf{x}=f(\mathbf{y},\mathbf{y}')$ works pairwise on elements as $x_{i}=f(y_{i},y_{i}')$. For example,
\begin{itemize}
    \item Standard residuals;
    \[
    x_{i}=y'_{i}-y_{i}\quad\text{or}\quad x_{i}=|y_{i}'-y_{i}|\quad ;
    \]
    \item Relative residuals; for some small $\varepsilon>0$,
    \[
    x_{i}=\frac{y'_{i}-y_{i}}{y_{i}+\varepsilon}\quad\text{or}\quad x_{i}=\frac{|y'_{i}-y_{i}|}{y_{i}+\varepsilon}\quad ;
    \]
    \item Log-ratios; again, for some small $\varepsilon>0$,
    \[
    x_{i}=\log\frac{y'_{i}+\varepsilon}{y_{i}+\varepsilon}\quad;
    \]
    \item $Z$-scores, which may be useful for isolating noise beyond sampling fluctuations: for $m$ the number of shots, compute the variance as $\sigma_{i}^2=y'_{i}(1-y_{i}')/m$, then for small $\varepsilon>0$,
    \[
    x_{i}=\frac{y'_{i}-y_{i}}{\sqrt{ \sigma^2+\varepsilon }}\quad.
    \]
\end{itemize}

We use the first three (together with their absolute versions) in this work, and leave other transformations, such as $Z$-scores, to future work. 

\subsection{Per-Qubit Features}

Now we restrict $f$ to considerations on the qubit-scale, operating on qubit indices to define marginal distributions over our bitstrings for computing features. We can, as before, look toward residual-based features, taking sums with respect to each qubit index to produce $n$-dim representations.

Assume that outcomes are ordered as bitstrings $b=b_{n-1},\dots,b_{0}$. Then we can operate directly on these bits in e.g. the following ways:
\begin{itemize}
    \item Single-qubit marginals; for a qubit $i$,
    \[
    x_{i}=\sum_{b:b_{i}=0}y'_{b}-y_{b}\quad\text{or}\quad x_{i}=\sum_{b:b_{i}=0}|y'_{b}-y_{b}|\quad;
    \]
    \item Pauli-$Z$ expectation per qubit;
    \[
    x_{i}=\langle Z_{i}\rangle'-\langle Z_{i}\rangle=\sum_{b}(-1)^{b_{i}}(y_{b}'-y_{b})\quad;
    \]
    \item Two qubit correlators; producing $n(n-1)/2$ features, with the $i$-th and $j$-th qubits giving,
    \[
    \begin{split}
        x_{i,j}
        &=\langle Z_{i}Z_{j}\rangle'-\langle Z_{i}Z_{j}\rangle\\
        &=\sum_{b}(-1)^{b_{i}\oplus b_{j}}(y'_{b}-y_{b})\quad ;
    \end{split}
    \]
    \item And so on into higher orders;
    \[
    \begin{split}
        x_{i_{1},\dots,i_{m}}
        &=\langle Z_{i_{1}}\cdots Z_{i_{m}}\rangle'-\langle Z_{i_{1}}\cdots Z_{i_{m}}\rangle\\
        &=\sum_{b}(-1)^{b_{i_{1}}\oplus\dots \oplus b_{i_{m}}}(y'_{b}-y_{b})\quad .
    \end{split}
    \]
\end{itemize}

We only extend into two-qubit correlator features in our experimentation, leaving higher orders to future work; with just $n=5$ qubits, two-qubit correlations are sufficient.

\section{Dimensionality Reduction Visualisation}
We now visualise embeddings of our experimental data in two dimensions using two common techniques: \textit{principal component analysis (PCA)} and \textit{$t$-distributed stochastic neighbour embedding ($t$-SNE)}. The former is a linear dimensionality reduction technique, which will allow us to understand the variance of our data projected into principal components and gauge the linear separability of the data classes (device membership, batch, etc.); the latter is a non-linear reduction which will much better allow us to understand the clustering of data classes as an indication for trainability. Such clustering patterns simply may not present in linear spaces, but can be exploited via the nonlinearity of deeply-layered neural networks with chained activation functions.


\begin{widefigure}
	\centering
	\begin{subfigure}[t]{0.48\linewidth}
		\centering
		\includegraphics[width=\columnwidth]{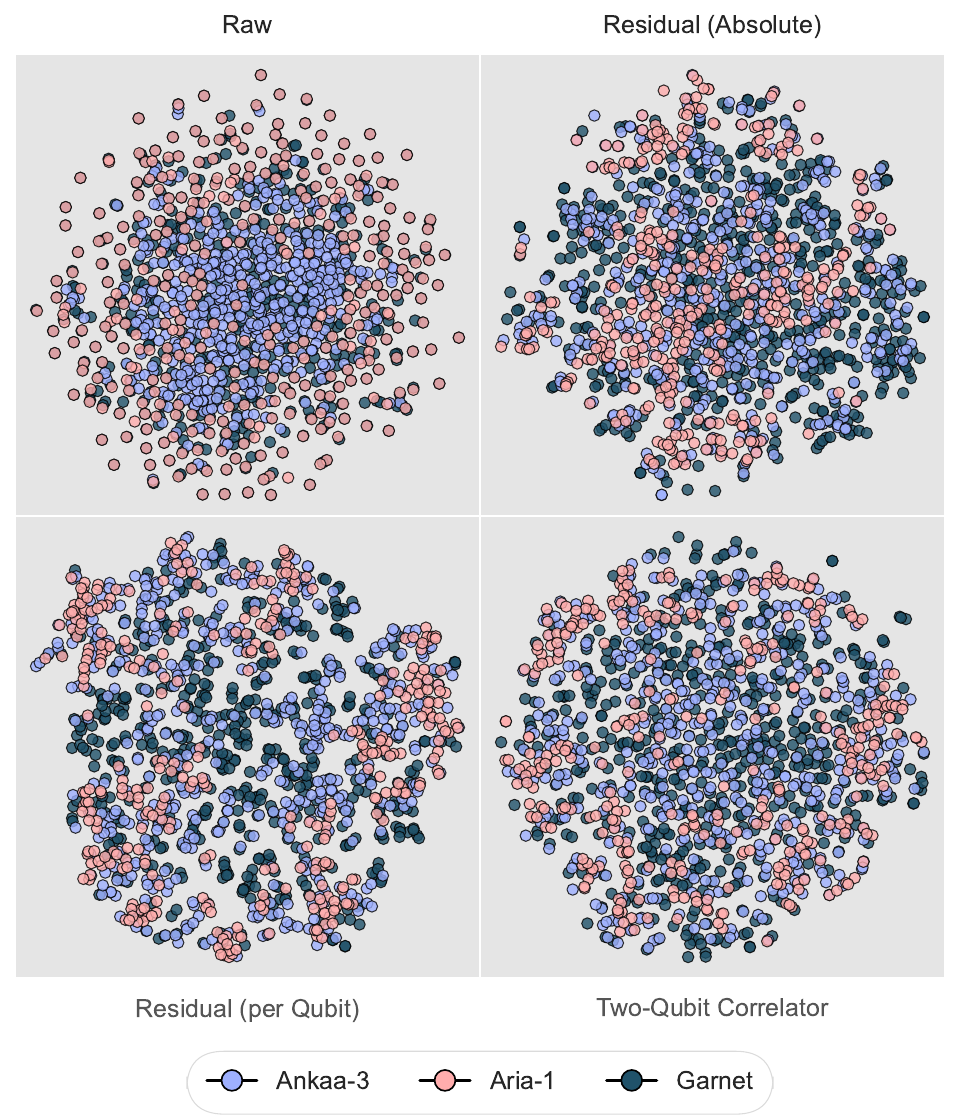}
		\caption{$t$-SNE plots for data produced on Haar-random circuits at a depth $d=5$.}
		\label{fig:tsne_examples}
	\end{subfigure}\hfill
	\begin{subfigure}[t]{0.48\linewidth}
		\centering
		\includegraphics[width=\columnwidth]{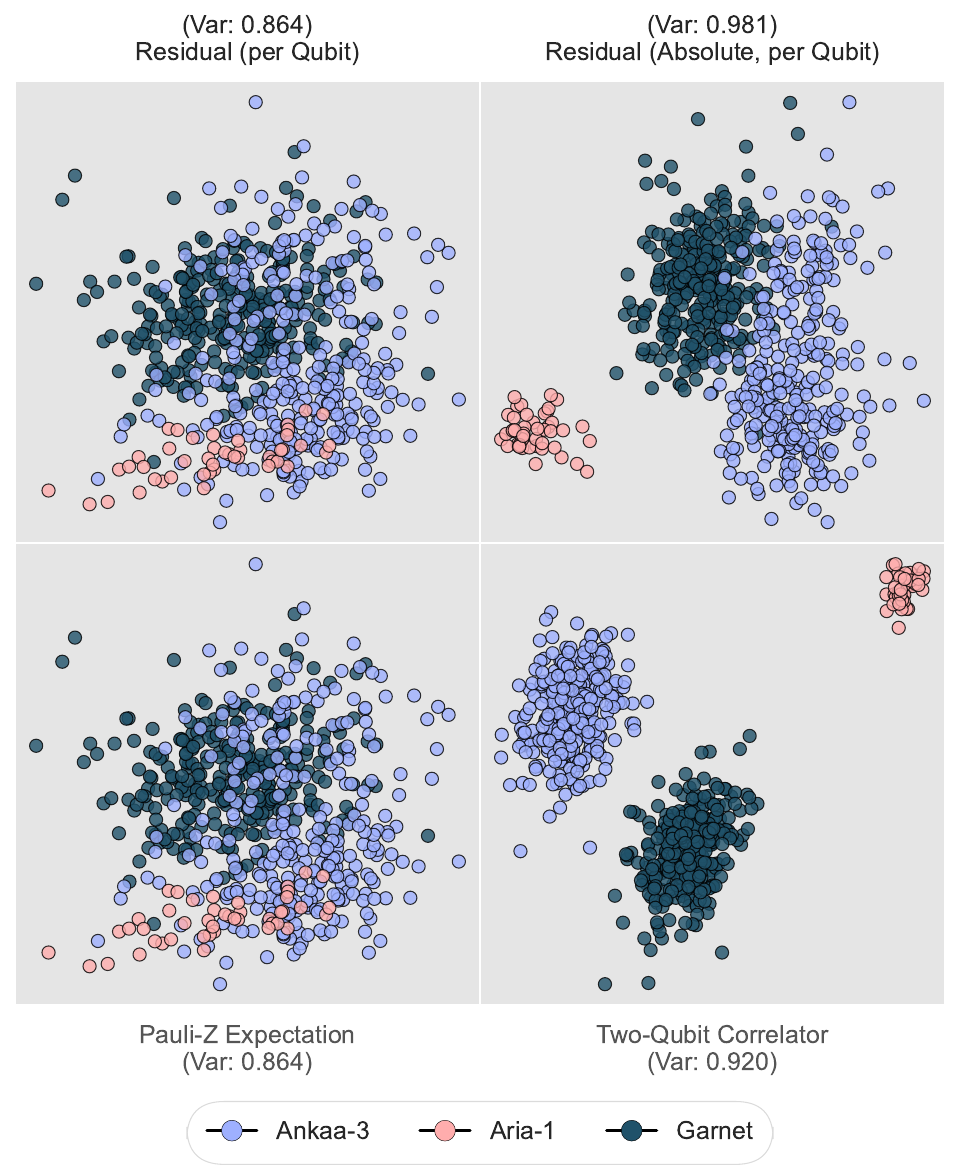}
		\caption{PCA plots for data produced on GHZ circuits over all batches.}
		\label{fig:pca_bell_state_examples}
	\end{subfigure}
	\caption{Example dimensionality reduction plots of quantum device feature vectors, hue by device. All features are shown via out repository, for both $t$-SNE and PCA reductions. In our PCA plots, the variance explained by the first two principal components (plotted) is given in parentheses.}
\end{widefigure}

\cref{fig:tsne_examples} demonstrates higher-dimensional clustering via $t$-SNE for the shallowest of the Haar-random circuit experiments, for a select few example features. It is immediately evident that like-type data points occupy similar neighbourhoods over many features, whereas differing-type data points indicate that at least a partial separation can be observed. Similarly, see \cref{fig:pca_bell_state_examples} and \cref{fig:pca_bell_state_by_batch_examples}.

\begin{figure}[t]
	\centering
	\includegraphics[width=\textwidth]{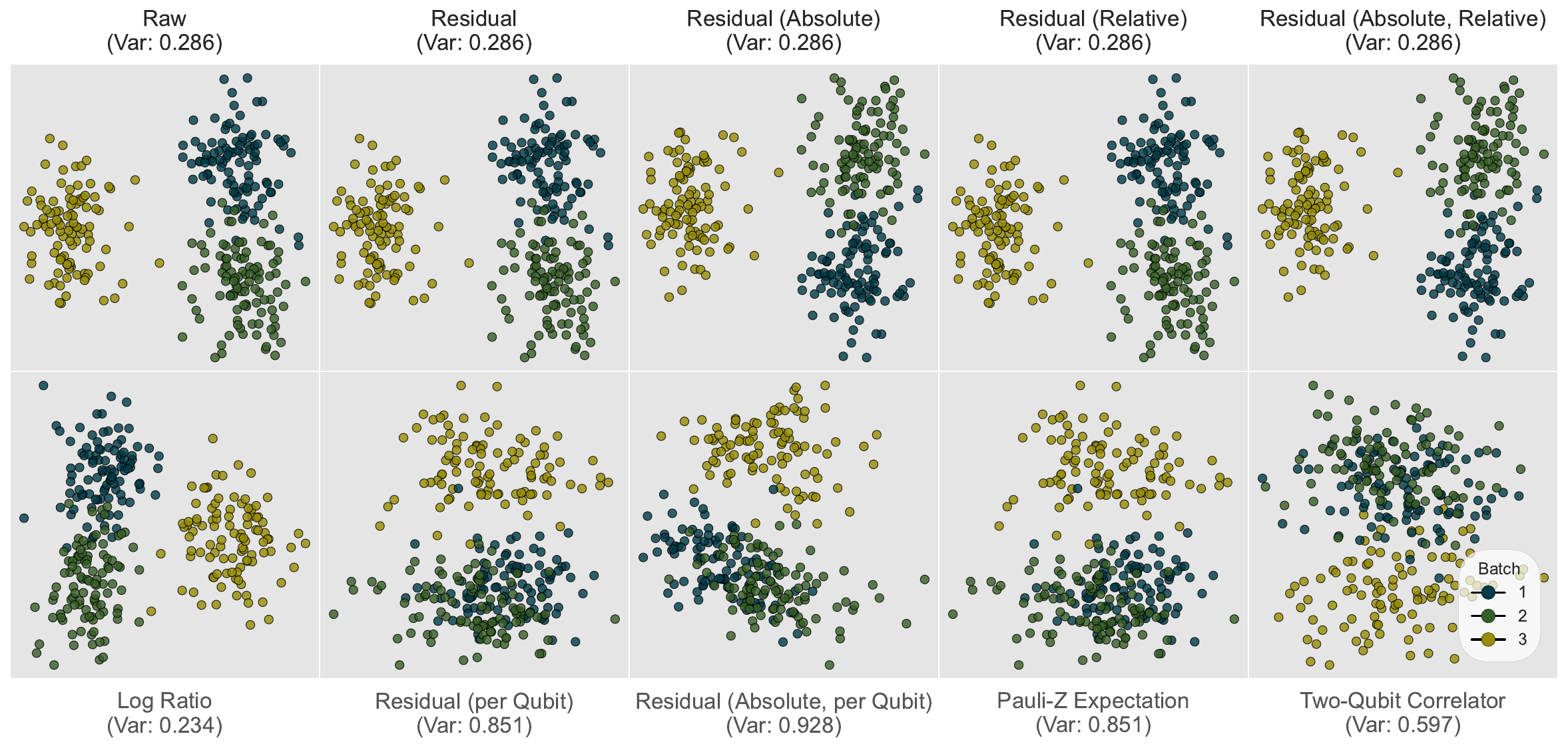}
	\caption{Dimensionality reduction visualisation (PCA) of quantum device feature vectors, hue by batch, for data produced on GHZ circuits.}
	\label{fig:pca_bell_state_by_batch_examples}
\end{figure}

\section{Classification Results} \label{appendix:classification_results}
We now consider some experimental classification tasks designed to empirically demonstrate that output distributions carry learnable route information. In the context of the backend identification game of \cref{def:formal_identification_game}, a classifier is a concrete decision rule for guessing the route that produced a given feature vector (representing a transcript derived from a finite-shot histogram). The test accuracies that we report here are then an empirical lower bound on the success probability achievable from the corresponding transcript class.

\subsection{Device Classification}

Device classification experiments use a multilayer perceptron with two hidden layers of widths 256 and 64 respectively, using batch normalisation, ReLU activation, and dropout after each layer, followed by a final classification layer. 

\begin{table*}[t]
	\centering
	\textbf{Experimental Backend Identification Success Probability}\vspace{1em}
	\small
	\renewcommand{\arraystretch}{1.25}
	\begin{tabular*}{\textwidth}{@{\extracolsep{\fill}}r|ccccc|cc}
		\toprule
		& \multicolumn{5}{c|}{\textbf{Like-Type}}
		& \multicolumn{2}{c}{\textbf{Differing-Type}} \\
		\cline{2-8}
		\textbf{Feature}
		& {$d=5$}
		& {$d=10$}
		& {$d=15$}
		& {$d=20$}
		& {GHZ}
		& {$d=5$}
		& {$d=10$} \\
		\midrule
		
		Two-Qubit Correlator   & ---/.69 & .52/.63 & .50/.60 & .51/.59
		& 1. 
		& ---/.79 & .70/.76 \\
		Pauli-Z Expectation    & ---/.68 & .53/.63 & .56/.65 & .59/.61
		& 1. 
		& ---/.78 & .60/.74 \\
		Residual (Abs., Qubit) & ---/.60 & .56/.60 & .58/.57 & .59/.62 
		& 1.
		& ---/.66 & .61/.60 \\
		Residual (Qubit)       & ---/.68 & .51/.64 & .56/.65 & .59/.61 
		& 1.
		& ---/.79 & .67/.71 \\
		
		\midrule
		
		\textbf{Log Ratio}     & \textbf{---/.77} & \textbf{.59/.73} & \textbf{.51/.67} & \textbf{.53/.66}
		& \textbf{1.}
		& \textbf{---/.96} & \textbf{.86/.90} \\
		Residual (Abs., Rel.)  & ---/.55 & .50/.54 & .50/.52 & .54/.52 
		& 1.
		& ---/.56 & .54/.58 \\
		Residual (Rel.)        & ---/.54 & .49/.51 & .51/.55 & .54/.52 
		& 1.
		& ---/.56 & .49/.65 \\
		Residual (Abs.)        & ---/.65 & .61/.67 & .53/.64 & .49/.62 
		& 1.
		& ---/.79 & .79/.81 \\
		Residual               & ---/.73 & .51/.67 & .60/.62 & .53/.63 
		& 1.
		& ---/.78 & .81/.83 \\
		
		\midrule
		
		\textbf{Raw}           & \textbf{---/.89} & \textbf{.90/.84} & \textbf{.89/.82} & \textbf{.87/.80}
		& \textbf{1.}
		& \textbf{---/.96} & \textbf{1./.98} \\
		
		\bottomrule
	\end{tabular*}
	\renewcommand{\arraystretch}{1.0}
	\caption{Binary classification test accuracies. Columns 1--4 and 6--7 are depth-varied (i.e. with $d$-layer brickwork circuit) experiments, while column 5 is the GHZ circuit (selecting best batch). Columns 1--5 are classification between like-type devices (Ankaa-3 vs Garnet); columns 6--7 are classification between differing-type (Ankaa-3 vs Aria-1). A column with depth label $d$ indicates accuracy obtained by a classifier trained on data \textit{at} depth $d$ / \textit{up to} depth $d$.}
	\label{tab:device_classification}
\end{table*}

\cref{tab:device_classification} shows that device identity can be recovered well above chance from finite-shot output data, most accurately from raw histogram features (as indicated by \cref{prop:post_processing_cannot_increase_information}). High like-type classification accuracies are particularly interesting, as they indicate that the distinguishability is not only a platform-level effect and can exist between even physically very similar backends. Differing-type classification is expectedly near-perfect in the raw case; distinguishing differing physical platforms is likely made relatively easy due to different native noise mechanisms etc. It is also experimentally event that backend information is able to survive several lower-dimensional or more interpretable post-processing transformations, most notably log ratio features. 

While the quantity and range of depths probed in these experiments is modest, the intermediate-depth principle formalised in Appendix \ref{appendix:depth_window_principle} is supported by observations in \cref{tab:device_classification}. Several features exhibit a non-negligible improvement in performance as we jump from depth $d=5$ to $10$, particularly residual-based features in differing-type classification, and a general trend of decaying distinguishability thereafter can be observed across all features in the deeper like-type experiments. In fact, for almost all post-processing features computed in this work (excluding raw transcripts), we can observe convergence on near-zero distinguishing bias (i.e. test accuracy exceeding $1/2$) in increasing depth. We have trained classifiers with both a single probed depth (``at depth'') and a cumulative pooling across depths (``up to depth'') to demonstrate more clearly this additional variance introduced that is not route-specific---incorporating deeper observations into the training set often degrades the classification performance exactly due to this intermediate-depth principle punishing the addition of less-distinguishable, high-depth transcripts. In further support of the principle, some features are aided/hindered by pooling across depths at different rates, owing to their specific optima for the most distinguishable depth (their peak in \cref{fig:intermedaite_depth_phenomenon_illustration}).

It is worth briefly noting the perfect classification we observe when distinguishing devices via the GHZ circuit (column 5 in \cref{tab:device_classification}). Here we have a fixed and highly structured circuit with a particularly non-uniform ideal outcome distribution, all of which makes the task of distinguishing as direct as possible. This is an unsurprising result, but it should not be overlooked; in the context of our theoretical discussions of pseudo channel distances, this emphasises the role of the workload ensemble in allowing users to violate routing anonymity, and the ease with which they may do so.

Recall \cref{fig:latent_representations}, which visualises the latent representations learned by the above classifiers (on raw transcripts). Similar clusterings across train, validation, and test data splits serves as a comforting visual marker of the generality with which backend-dependent noise signals can be learned. \cref{fig:all_devices_raw_5} gives some particularly interesting insight about how separable our classifiers understand the like-type superconducting devices to be from the ion-trap device; Ankaa-3 and Aria-1 (of differing type) separate their representations very strongly in the latent space (strengthened in the differing-type-only classifier's embeddings in \cref{fig:differ_type_raw_10}), with the other superconducting device occupying the middle ground between the two, and densely mixing with its like-type counterpart while only sparsely mixing with the differing type.

\subsection{Temporal (Batch) Classification}

Now we ask whether time itself---the dynamics of backend-specific noise---imprints a learnable fingerprint in output distributions. Rather than asking \textit{which} backend produced a given transcript, we fix the backend and ask \textit{when} the transcript was produced (in a relative sense). Specifically, we classify a device's transcripts by batch; since batches are separated by substantially longer intervals than calls within a batch, successful classification in the context of these experiments implies a long-range dynamic character to backend-specific noise signals.

These GHZ experiments are supplemented with an additional feature: \textit{bell summary}, consisting of the two GHZ peak probabilities, total leakage into the remaining outcome states, peak imbalance, $\ell_1$-distance from the ideal GHZ distribution, and entropy. The intent here is to introduce a feature that understands broadly how much noise there is with respect to leakage away from the bell state's support, but which is agnostic to specific detail about where that noise has ended up. If we see more specific features classifying better than bell summaries, then we can infer that our classifiers are identifying device-specific noise signals in shape and character, rather than simply learning that one device is `noisier' than another.

\begin{table*}[t]
	\centering
	\textbf{Experimental Batch Classification (and Negative Control)}\vspace{1em}
	\small
	\renewcommand{\arraystretch}{1.25}
	\begin{tabular*}{\textwidth}{@{\extracolsep{\fill}}r|ccc|cc}
		\toprule
		
		& \multicolumn{3}{c|}{\textbf{Batch Classification}}
		& \multicolumn{2}{c}{\textbf{Control}} \\
		\cline{2-6}
		\textbf{Feature}
		& \textbf{Ankaa-3}
		& \textbf{Garnet}
		& \textbf{Combined}
		& \textbf{Ankaa-3}
		&\textbf{Garnet} \\
		\midrule
		
		Bell Summary & .61 $\pm$ .048 & .61 $\pm$ .036 & .46 $\pm$ .034 & .28 $\pm$ .022 & .36 $\pm$ .049 \\
		
		\midrule
		
		Two-Qubit Correlator    & .92 $\pm$ .021 & .78 $\pm$ .051 & .74 $\pm$ .030 & .33 $\pm$ .059 & .30 $\pm$ .056 \\
		Pauli-Z Expectation     & .92 $\pm$ .024 & .59 $\pm$ .085 & .62 $\pm$ .043 & .33 $\pm$ .035 & .30 $\pm$ .073 \\
		Residual (Abs., Qubit)  & .89 $\pm$ .027 & .43 $\pm$ .032 & .57 $\pm$ .044 & .36 $\pm$ .069 & .36 $\pm$ .052 \\
		Residual (Qubit)        & .92 $\pm$ .024 & .59 $\pm$ .085 & .62 $\pm$ .043 & .33 $\pm$ .035 & .30 $\pm$ .073 \\
		
		\midrule
		
		Log Ratio    & .98 $\pm$ .015 & .81 $\pm$ .039 & .84 $\pm$ .012 & .32 $\pm$ .036 & .33 $\pm$ .022 \\
		Residual (Abs., Rel.)   & .99 $\pm$ .013 & .80 $\pm$ .047 & .83 $\pm$ .007 & .29 $\pm$ .045 & .34 $\pm$ .020 \\
		Residual (Rel.)         & .99 $\pm$ .013 & .80 $\pm$ .047 & .83 $\pm$ .007 & .29 $\pm$ .045 & .34 $\pm$ .020 \\
		Residual (Abs.)         & .99 $\pm$ .013 & .80 $\pm$ .047 & .83 $\pm$ .007 & .29 $\pm$ .045 & .34 $\pm$ .020 \\
		Residual                & .99 $\pm$ .013 & .80 $\pm$ .047 & .83 $\pm$ .007 & .29 $\pm$ .045 & .34 $\pm$ .020 \\
		
		\midrule
		
		\textbf{Raw}          & \textbf{.99} $\pm$ \textbf{.013} & \textbf{.80} $\pm$ \textbf{.047} & \textbf{.83} $\pm$ \textbf{.007} & \textbf{.29} $\pm$ \textbf{.045} & \textbf{.34} $\pm$ \textbf{.020} \\
		
		\bottomrule
	\end{tabular*}
	\renewcommand{\arraystretch}{1.0}
	\caption{Three-way temporal (batch) classification test accuracies for the time-varied GHZ experiments, together with a randomly permuted-label negative control experiment. We report the mean test accuracy over 5 runs $\pm$ one standard deviation. Column 3 is a combined and balanced dataset on both Ankaa-3 and Garnet.}
	\label{tab:batch_classification}
\end{table*}

\cref{tab:batch_classification} presents logistic-regression classifiers on balanced device-batch groups, with standardised input features evaluated over five random stratified train-test splits. Note that we have three distinct batches (for the two superconducting devices), so performance is measured against a majority baseline $1/3$. We observe an extremely strong temporal structure. Retaining high classification accuracies in the combined dataset indicates that this temporal signal is not merely an artefact of training separate per-device classifiers; batch-dependent structure persists largely independently of apparent device signal. We include a control experiment in \cref{tab:batch_classification} wherein the dataset's labels are randomly permuted. The same logistic-regression models fall to no better than majority guessing, supporting the conclusion that our high batch-classification accuracies are not caused by data leakage, class imbalance, or overparameterised classifiers.

These results naturally follow the persistent routing setting of \cref{def:persistent_routing}, in which we repeatedly probe a fixed backend. In practice, we expect the induced law associated with a backend to drift over time, which can then produce the kind of temporal signal that we show above can be extremely identifying about the route. Great care should be taken with persistent routing to ensure that such identifying temporal structure is not also leaked via transcripts.

We also consider a shorter time-scale control by splitting each batch into an `early' section and `late' section, then perform a binary classification between the two. This gives us an empirical view for the learnability of the perhaps more erratic yet muted short-term drift over the lifespan of a contiguous batch. The intra-batch results presented in \cref{tab:early_vs_late_classification} yield a better prospect for short-term routing anonymity than the long-term identification in \cref{tab:batch_classification}. It is clear, then, that the dominant temporal effect in these data are batch-scale effects rather than any drift within batches. This is useful insight for the difficulty of \textit{forecasting} backend-specific noise in the short-term, as we do next.

\begin{table*}[t]
	\centering
	\textbf{Experimental Intra-Batch (Early vs Late) Classification}\vspace{1em}
	\small
	\renewcommand{\arraystretch}{1.25}
	\begin{tabular*}{\textwidth}{@{\extracolsep{\fill}}r|ccc|ccc|c}
		\toprule
		
		& \multicolumn{3}{c|}{\textbf{Ankaa-3}}
		& \multicolumn{3}{c|}{\textbf{Garnet}}
		& \textbf{Aria-1} \\
		\cline{2-8}
		\textbf{Feature}
		& $B=1$
		& $B=2$
		& $B=3$
		& $B=1$
		& $B=2$
		& $B=3$
		& $B=3$ \\
		\midrule
		
		Bell Summary & .54 & .53 & .49 & .70 & .45 & .47 & .41 \\
		
		\midrule
		
		Two-Qubit Correlator    & .51 & .54 & .49 & .64 & .44 & .56 & .50 \\
		Pauli-Z Expectation     & .46 & .55 & .54 & .58 & .50 & .53 & .55 \\
		Residual (Abs., Qubit)  & .47 & .52 & .41 & .52 & .50 & .58 & .53 \\
		Residual (Qubit)        & .46 & .55 & .54 & .58 & .50 & .53 & .55 \\
		
		\midrule
		
		Log Ratio    & .53 & .57 & .46 & .55 & .46 & .58 & .56 \\
		Residual (Abs., Rel.)   & .49 & .61 & .45 & .54 & .48 & .57 & .53 \\
		Residual (Rel.)         & .49 & .61 & .45 & .54 & .48 & .57 & .54 \\
		Residual (Abs.)         & .49 & .61 & .45 & .54 & .48 & .57 & .53 \\
		Residual                & .49 & .61 & .45 & .54 & .48 & .57 & .54 \\
		
		\midrule
		
		Raw          & \textbf{.49} & \textbf{.61} & \textbf{.45} & \textbf{.54} & \textbf{.48} & \textbf{.57} & \textbf{.54} \\
		
		\bottomrule
	\end{tabular*}
		\renewcommand{\arraystretch}{1.0}
	\caption{Within each batch, early vs late classification for time-varied GHZ experiments. We order data chronologically by completion time, then split batch data into two groups and classify membership between them. Similar to \cref{tab:batch_classification}, these are logistic-regression classifiers over 5 random stratified samples, with mean test accuracy reported (standard deviations omitted for readability). Aria-1's $B=3$ means that its runs were scheduled alongside Ankaa-3 and Garnet's $B=3$ batch.}
	\label{tab:early_vs_late_classification}
\end{table*}

\section{Distribution Forecasting} \label{appendix:forecasting}
Our final experiment is exploratory and not intended to be a central contribution of the present work. The purpose of this section is to illustrate a natural continuation of the persistent routing formulation: if backend-specific noise signals have a strong (long-term) temporal structure, then is it possible to estimate precisely how noise will present in the future (that is, to ``forecast'' the noise fingerprint)? There are some intriguing possibilities surrounding this question; e.g. can we pre-emptively correct noisy outcome distributions consistently, or can we reverse-forecast to an early time to predict ideal distributions? Such questions make for interesting future work---we only indicate at the possibility for forecasting in this section.

We take the same time-ordered GHZ data and train a long short-term memory (LSTM) recurrent model to predict the next empirical outcome distribution from a short prior of previous distributions. Concretely, for a sequence of length $L$, an input prior takes the form $(\mathbf{y}'_{t-L},\dots,\mathbf{y}'_{t-1})$ and ask for a prediction of $\mathbf{y}'_t$ (note, we forecast the \textit{measured} outcome distribution, not the true distribution). Our exploratory model uses 10 hidden layers of 512 dimensions, dropout, and a softmax normalised output layer, using a KL-divergence-based loss function.

An example prediction, alongside the actual measured distribution, is presented in \cref{fig:forecasting}, taking a prior window of $L=10$ previous successive measurement steps \textit{in batch}. Perhaps unsurprisingly, we are able to near-perfectly capture the precise leakage away from the all-0s and all-1s states, but more impressively, the rough localisation of the leakage across the remaining states can be broadly understood even at this short-term window scale.

\begin{figure}[t]
	\centering
	\includegraphics[width=\textwidth]{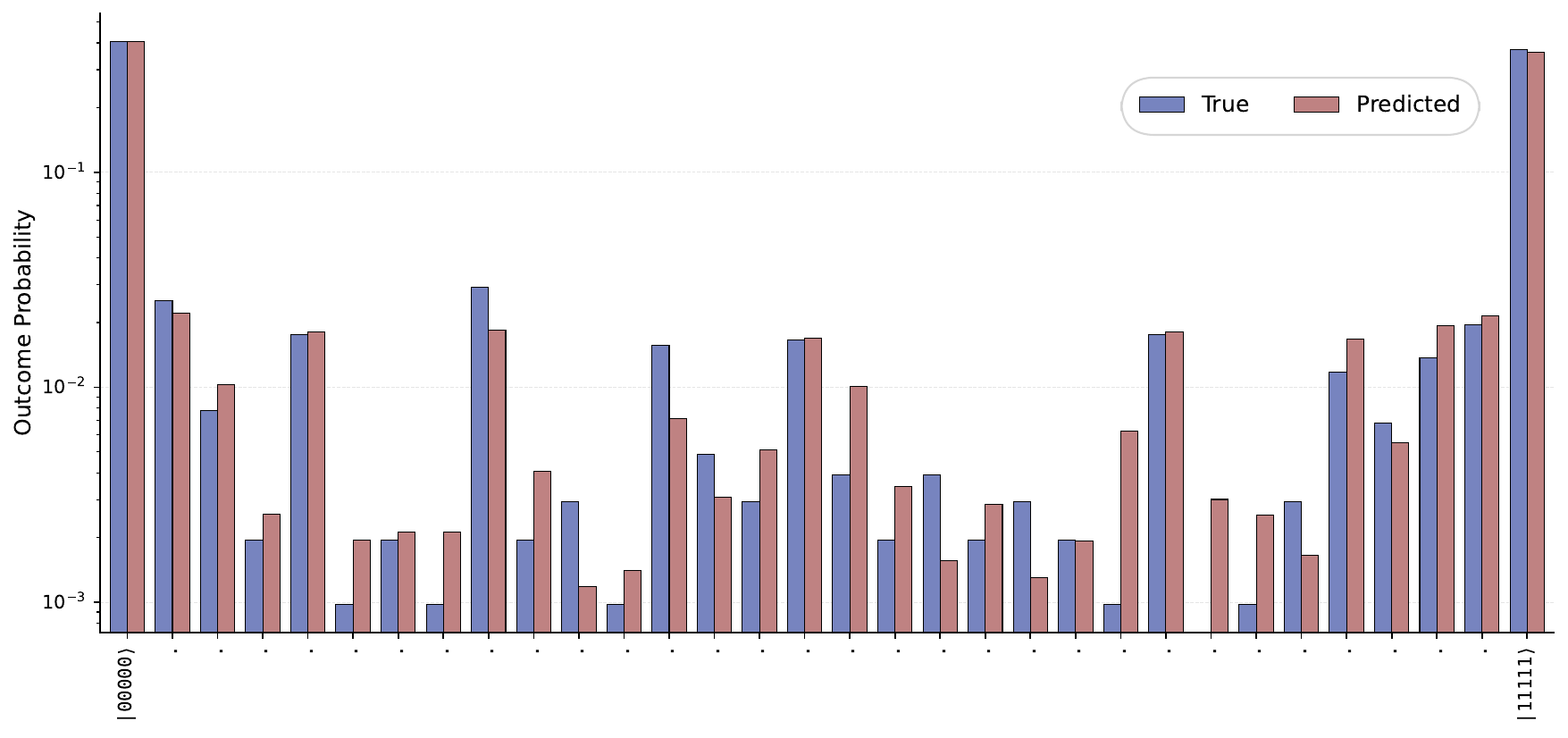}
	\caption{Forecasting time-varied experimental data for Ankaa-3 using an LSTM model; given a prior of sequential measurement steps (here, 10), predict the next chronological step.}
	\label{fig:forecasting}
\end{figure}

At this stage, the result of \cref{fig:forecasting} is merely a proof-of-concept. More complete study would ideally evaluate the correctness of forecasted distributions over many more metrics (we only use KL divergence here), and account for the irregularity of cloud-based computation times. In the context of routing anonymity, the forecasting problem presents an additional security vulnerability---that users may be able to strategically adapt their expectations about backend signals over time, even without an adaptive access model. Conversely, a provider looking to retain anonymity of their routing choices should not only be concerned about the distinguishability of fixed transcript laws, but also the temporal stability and predictability of their released transcripts over repeated service calls (particularly over longer time intervals).

\end{document}